\documentclass[10pt]{article}
\usepackage[margin=1in]{geometry}
\usepackage[T1]{fontenc}\usepackage[utf8]{inputenc}
\usepackage{amsmath,amssymb,amsthm,mathtools,mathrsfs,booktabs,array,tabularx,xcolor,enumitem,float,graphicx,subcaption,tikz,xspace}
\usetikzlibrary{arrows.meta,positioning,fit,calc}
\usepackage{natbib,hyperref}\usepackage{url}
\hypersetup{hidelinks}
\setlist{nosep,leftmargin=1.4em,topsep=2pt}\renewcommand{\arraystretch}{1.12}
\newtheorem{assumption}{Assumption}\newtheorem{definition}{Definition}\newtheorem{lemma}{Lemma}\newtheorem{theorem}{Theorem}\newtheorem{corollary}{Corollary}
\newcommand{\TV}{d_{\mathrm{TV}}}\newcommand{\push}{\mathbin{\#}}
\newcommand{\pihat}{\widehat{\boldsymbol\pi}^{\mathrm{lex}}}
\tikzset{bbnode/.style={draw=black,rounded corners=2pt,fill=white,align=center,text width=25mm,minimum height=12mm,inner sep=3pt},bbmap/.style={-{Latex[length=2mm]},thick,black},bbdash/.style={-{Latex[length=2mm]},thick,dashed,black}}

\title{Calibrating Semantic Uncertainty from Observable Language-Model Probabilities}
\author{Matthew Dixon\thanks{Corresponding author. Email: \texttt{matthew.dixon@aifi.edu}}\\[-1pt]\small Artificial Intelligence Finance Institute (AIFI), 69 Charlton St, New York, NY 10014, USA}
\date{27 July 2026}

\begin{document}\maketitle
\begin{abstract}
As generative artificial intelligence enters scientific and professional work, its uncertainty must be defined on the states that matter for inference and decision-making. Language models assign probabilities to words, whereas applications require uncertainty over meaningful states such as diagnoses, hypotheses or operational conditions. We introduce a \emph{semantic map}: a prespecified, testable bridge from probabilities over verbal responses to a posterior over declared finite states. The language distribution remains unrestricted; held-out calibration connects it to a reference posterior. We derive posterior-error bounds and conditions for existence, conditional uniqueness, presentation stability and stable inverse recovery. This distinction matters because language probabilities depend on prompt wording, while the target posterior should not change under information-equivalent rewording.

Experiments use professional market text compiled from Federal Reserve economic and financial series, together with controlled simulations having exact posteriors. Across two fitted language models, language-derived probabilities outperform printed numerical confidence, recover held-out posteriors with valid uncertainty coverage, remain largely stable under paraphrase and respond appropriately to altered evidence. \textbf{Prompt engineering optimises a wording-dependent response; robust scientific use requires validated stability of application-relevant meaning.} The proposed map turns semantic uncertainty in generative systems into an identifiable and testable statistical measurement problem and, when its acceptance conditions hold, yields an auditable posterior estimate.
\end{abstract}


\section{Introduction}
Language models are now used to assist judgement in medicine, law, cybersecurity, science, engineering, finance and many other areas. Their adoption extends far beyond a specialist community: a survey of 100,000 workers in exposed occupations found that one half had used ChatGPT \citep{humlum2025adoption}. Although these systems appear conversational, their native statistical output is a conditional distribution over possible words and phrases. Professional decisions, by contrast, require uncertainty over alternatives that have meaning in the application: competing diagnoses, hypotheses, threat conditions or operational states. This is a fundamental uncertainty-quantification problem for generative systems: the observable probability law and the scientifically relevant estimand live on different spaces. Can language probabilities shaped by evidence selection and prompt presentation recover a posterior that is stable under information-equivalent changes in the prompt's words, order and numerical format? We refer informally to such a statistically defined bridge from language probabilities to state probabilities as a \emph{semantic map}.

\noindent\textbf{Semantic map:} Consider a simple clinical example in which ``urgent review'' and ``immediate escalation'' express the same patient-assessment state. One prompt may divide probability $.45$ and $.15$ between them, while an information-equivalent rewording divides it $.20$ and $.40$. The most probable phrase changes, but the combined urgent meaning remains $.60$. Prompt engineering can therefore redistribute probability among synonymous expressions without defining their shared scientific meaning or establishing a calibrated state posterior. A prespecified semantic map is needed before posterior recovery and stability can be tested.

Constructing the map remains non-trivial because different probability quantities must not be conflated. The candidate state vectors $(0.51,0.48,0.01)$ and $(0.98,0.01,0.01)$ have the same modal state but very different concentration. Asking a model to print ``the probability of urgency is 70\%'' produces a \emph{numerically elicited probability}, which may reflect rounding or instruction-following conventions. A more direct language-based measurement is the \emph{continuation probability}: the conditional probability of a complete phrase such as ``urgent review'', calculated from its conditional token probabilities when available.\footnote{N.B. the continuation probability and the elicited number are separate observable measurements; neither is assumed to equal the other or the posterior.} Although closer to the model's native statistical output, it remains a prompt-dependent probability over language rather than a prompt-invariant state posterior.

The evaluator chooses the evidence and its prompt; the fitted service returns conditional token probabilities, which may be multiplied to obtain probabilities for complete phrases. The desired quantity lies one step further on: a posterior over meaningful states. We call the route from prompt to phrase probabilities the \emph{language channel}; the semantic map is the proposed bridge from that channel to the posterior target.

This distinction is more than terminological. Prompt design is an experimental factor, not a substitute for an estimand, calibration or validation. Rewording may alter a probability, classification, recommendation or refusal without altering the information supplied. Fluency and repeatability are therefore weak evidence of semantic stability. Informally, \emph{prompting} means choosing the wording, order, instructions and response format through which the evidence is presented. The resulting language probabilities may depend on these choices; the target posterior, under a fixed evidence set and reference model, should not depend merely on information-equivalent wording.

\textbf{The central methodological contribution is a scientifically defined and empirically testable map from phrase probabilities to $P(\text{state}\mid E)$.} Meaning-equivalent grouping has useful precedents \citep{farquhar2024semantic,kuhn2023semantic}, but the language model does not itself supply this map. We place the grouping within a semiparametric measurement model, retain unmatched language explicitly, and test calibration, identification and stability on untouched observations. The language distribution remains unrestricted; a low-dimensional calibration map connects it to the reference posterior.

The paper is organised around three questions. \emph{First}, does semantic measurement add information beyond a probability printed by the same fitted model? \emph{Second}, can it recover the target posterior on new cases, with calibrated uncertainty? \emph{Third}, is it stable under changes in presentation while remaining responsive to genuine changes in evidence? These questions move from practical value, through statistical recovery, to robustness. They also determine the order of the theory, experimental design and results.

The answers also provide a basis for runtime governance. Rather than trusting a fluent response or printed confidence, a deployed procedure can admit a state estimate only within its validated scope and otherwise abstain, request review or seek additional evidence.

In the running example, phrase probabilities $.45$, $.15$, $.25$, $.10$ and $.05$ may correspond to ``urgent review'', ``immediate escalation'', ``routine follow-up'', ``insufficient information'' and other responses. Semantic grouping assigns $.60$ to urgent language; held-out calibration may map the grouped vector to a state posterior such as $(.51,.48,.01)$. The two distributions are intentionally distinct. Words are measurements, not scientific states, and their probabilities must earn a posterior interpretation.

The proposed formulation draws together several statistical traditions that have largely developed separately.

\noindent\textbf{Calibration and forecast evaluation.} Classical work distinguishes calibration from sharpness and uses proper scores and prequential evaluation to assess probabilistic forecasts \citep{dawid1982well,dawid1984prequential,degroot1983comparison,gneiting2007strictly}. Neural-network calibration and conformal inference extend this concern to flexible predictors and finite-sample coverage \citep{angelopoulos2023conformal,guo2017calibration,lei2018distribution,vovk2005algorithmic}. Yet these theories begin with a declared outcome and an already defined predictive distribution. Here the difficulty arises earlier: the fitted model supplies a distribution over language, while the scientifically relevant state space and its relation to that language still require definition. Repeated prompt adjustment on the same examples also recalls the familiar danger of specification search \citep{chatfield1995model}.

\noindent\textbf{Confidence and uncertainty in language models.} Token probabilities, answer likelihoods and stated confidence can all contain useful signal \citep{jiang2021calibration,lin2022teaching,tian2023justask}. Rank calibration offers a complementary assessment without forcing heterogeneous scores onto one numerical scale \citep{huang2024rank}. This literature also makes plain that self-reported numbers and language probabilities are different measurements. Neither is thereby a posterior over a declared latent state.

\noindent\textbf{Semantic uncertainty.} Different strings may express the same meaning. Semantic-entropy methods group equivalent generations and can assist hallucination detection \citep{farquhar2024semantic,kuhn2023semantic}; decision-theoretic formulations offer an alternative \citep{wang2025subjective}. We ask a different inverse question: can the observable language distribution recover a separately defined finite-state posterior, and can failure be attributed to semantic grouping, probability calculation, uncovered language or unstable inversion?

\noindent\textbf{Inverse problems, experiments and compositions.} Treating language as a stochastic measurement brings familiar theory into view. Comparison of experiments describes information loss and sufficiency \citep{blackwell1953equivalent,lecam1964sufficiency,torgersen1991comparison}; inverse-problem theory explains ill-conditioned recovery \citep{engl1996regularization}; and compositional analysis supplies coordinates for probability vectors \citep{aitchison1982statistical,egozcue2003isometric}. The semantic map connects these ideas to a fitted model's conditional language distribution.

\subsection{Contributions of the paper}

\paragraph{Research gap.} Existing work does not jointly provide (i) a declared map from complete verbal responses to finite states, including unmatched language; (ii) calibration against an independently specified posterior; (iii) decomposition of semantic, missing-mass and inverse-recovery error; and (iv) held-out tests of identification and presentation stability. Together these distinguish a semantic map from an informal phrase dictionary or another prompt-engineering recipe.

\paragraph{Theoretical contribution.} The mathematical novelty is not a claim to have rediscovered image measures, total-variation contraction or singular-value inequalities. Those are established results and are cited as such. The contribution is to construct a new semantic measurement experiment around them and to derive guarantees that the general results do not state: an error decomposition for finite scored approximations of semantic cells; amplification by unexpressed mass; propagation of prompt-dependent language-law perturbations through semantic conditioning and calibration; and a combined semantic-to-posterior error budget. These results connect quantities that can be audited in the language experiment to error in the posterior estimand.

\paragraph{Structure of the theory.} The theoretical roadmap has four links. A standard image-measure lemma is first specialised to the unexpressed-state construction and its required normalisation. The paper's approximation theorem then supplies the new semantic error decomposition. A perturbation theorem combines standard data processing with the conditioning step needed for declared-state probabilities. Finally, a standard inverse-recovery result is joined to those semantic errors in a new combined bound, whose affine corollary gives the singular-value diagnostic an operational interpretation.

\paragraph{Empirical contribution.} The experiments answer the three central questions rather than offer a model leaderboard. Professional text compares semantic measurement with probabilities printed by the same fitted model. Controlled finite-state experiments test exact posterior recovery and conformal coverage across two fitted models. Finally, paraphrases, reorderings, omissions and contradictions separate presentation stability from responsiveness to evidence. The answers are encouraging but qualified: semantic measurement adds information, recovers held-out posteriors with calibrated uncertainty on the sampled design, and is largely stable under paraphrase; raw word uncertainty is not posterior uncertainty, however, and evidence order remains consequential.

\paragraph{Practical implication.} The practical claim is deliberately modest in form, though consequential in use. We do not claim access to a model's inaccessible ``true belief'', nor should every generated probability be trusted. Rather, an observable language distribution can become an auditable posterior estimator when explicit conditions hold. This supplies a runtime-governance rule: permit downstream use within the validated scope, but trigger abstention, additional evidence or human review when semantic coverage, inverse conditioning or predictive coverage fails.

\subsection{Organisation of the paper}
The remainder of the paper is as follows. Section~2 gives the statistical setup, defining both the reference probability model and the indirect observation supplied by the response distribution. Section~3 formulates identification, estimation and uncertainty. Section~4 develops the mathematical theory. Section~5 describes the experimental design used to test its observable implications, Section~6 reports the experimental results, and Section~7 concludes. The appendix provides notation, methodological lineage and step-by-step proofs.

\def\BeliefLensSetupOnly{1}
\ifdefined\BeliefLensSetupOnly
\newcommand{\BeliefReferenceHeading}{\section{Statistical setup}\subsection{Reference probability model and estimand}}
\newcommand{\BeliefPromptsHeading}{\subsection{Indirect observation through a response distribution}\subsubsection{Prompts as designed inputs}}
\newcommand{\BeliefContinuationHeading}{\subsubsection{Conditional continuation laws and semantic structure}}
\newcommand{\BeliefPushforwardHeading}{\subsubsection{Semantic pushforward}}
\newcommand{\BeliefMismatchHeading}{\section{Identification, estimation and uncertainty}\subsection{Statistical mismatch and inferential target}}
\newcommand{\BeliefCalibrationHeading}{\subsection{Calibration and verification}}
\newcommand{\BeliefNotationHeading}{\noindent For convenience, Table~\ref{tab:notation-guide} summarizes the primary mathematical objects used in the setup and subsequent analysis.}
\newcommand{\BeliefNotationBegin}{\begin{table}[H]\centering\small\caption{Primary mathematical objects in the statistical formulation.}\label{tab:notation-guide}}
\newcommand{\BeliefNotationEnd}{\end{table}}
\else
\documentclass[11pt]{article}
\usepackage[margin=1in]{geometry}
\usepackage[T1]{fontenc}
\usepackage{amsmath,amssymb,mathtools,mathrsfs,enumitem,microtype,booktabs,tabularx}
\usepackage{tikz}
\usetikzlibrary{arrows.meta,positioning}
\usepackage[hidelinks]{hyperref}
\setlist[enumerate]{leftmargin=*,itemsep=8pt,topsep=5pt}
\renewcommand{\arraystretch}{1.18}
\title{Measure-Theoretic Background for the Semantic Measurement Problem\\\large A short introductory note}
\author{Matthew Dixon}
\date{1 August 2026}

\begin{document}
\maketitle

\noindent
\textbf{Purpose.} This note introduces the small amount of measure theory
needed for the statistical setup. The notation is the same as in the companion
setup note. The central idea is simple: a probability model requires a set of
possible outcomes, a declared collection of observable events, and a rule that
assigns probabilities to those events.

The main development follows Chapters~2 and~3 of Capi\'nski and Kopp:
probability spaces and Borel sets first, followed by measurable functions,
random variables, generated $\sigma$-fields and induced distributions. The
final discussion of conditional probability kernels is clearly marked as the
one extension needed by the paper.

\noindent\textbf{Terminology inherited from Capi\'nski and Kopp.} We use
\[
 (\Omega,\mathscr F,\mathbb P)
\]
for a probability space. The set $\Omega$ is the \emph{sample space}; a point
$\omega\in\Omega$ is an \emph{outcome} or \emph{sample point}; and a set
$A\in\mathscr F$ is an \emph{event}. Thus $\omega$ is not an event and
$\Omega$ is not an event space. For example, when a die is rolled,
\[
 \Omega=\{1,2,3,4,5,6\},\qquad \omega=4,
 \qquad A=\{2,4,6\}\in\mathscr F.
\]
Here $\omega=4$ is the realized outcome, while $A$ is the event ``an even
number is rolled.'' Capi\'nski and Kopp's Definition~2.7 calls the members of
$\mathscr F$ events.

Chapter~3 of that text reserves \emph{random variable} for a measurable map
$Z:\Omega\longrightarrow\mathbb R$. The paper uses the common statistical
extension \emph{categorical random variable} for a measurable map into a
finite state space. If strict Chapter~3 notation is desired, choose a numerical
coding $\iota:\mathcal X\longrightarrow\{1,\ldots,K\}\subset\mathbb R$;
then $\iota\circ X:\Omega\longrightarrow\mathbb R$ is a random variable in
the book's literal sense. No probability statement depends on the arbitrary
numerical coding.

\section*{How this note builds the bridge}

The note is not intended as a second textbook. Its purpose is to begin with
the Chapter~2--3 construction and alter one component at a time. The complete
bridge is:
\newcommand{\BeliefReferenceHeading}{\subsection*{Reference probability model and estimand}}
\newcommand{\BeliefPromptsHeading}{\subsection*{Prompts as designed inputs}}
\newcommand{\BeliefContinuationHeading}{\subsection*{Conditional continuation laws and semantic structure}}
\newcommand{\BeliefPushforwardHeading}{\subsection*{Semantic pushforward}}
\newcommand{\BeliefMismatchHeading}{\subsection*{Statistical mismatch and inferential target}}
\newcommand{\BeliefCalibrationHeading}{\subsection*{Calibration and verification}}
\newcommand{\BeliefNotationHeading}{\section*{Compact notation guide}}
\newcommand{\BeliefNotationBegin}{\begin{center}\small}
\newcommand{\BeliefNotationEnd}{\end{center}}
\fi

\BeliefReferenceHeading
This section fixes the probability model independently of the language-based
measurement introduced later. It first specifies the latent state and observed
evidence, then defines the conditional state distribution that the method is
intended to estimate.

\noindent\textbf{Reference probability space and random variables.}
Begin with
$(\Omega,\mathscr F,\mathbb P)$. This notation follows the latent-state
convention, not the regression convention: $X$ is the unobserved scientific
state and $Y$ is its observed evidence. Using the conventional symbol
$\Omega$ for their common sample space, define
\[
 X:(\Omega,\mathscr F)\longrightarrow
   (\mathcal X,\mathcal B(\mathcal X)),
 \qquad
 Y:(\Omega,\mathscr F)\longrightarrow
   (\mathcal Y,\mathscr Y).
\]
$X$ represents the latent scientific state and $Y$ represents the observed
evidence. They are jointly distributed under $\mathbb P$.

\noindent\textbf{Sampling model.}
For $n$ sampled scenarios,
\[
 (X_i,Y_i)\overset{\mathrm{i.i.d.}}{\sim}\mathbb P_{X,Y},
 \qquad i=1,\ldots,n.
\]
$\mathscr F$ contains measurable events such as
$\{X=\text{urgent}\}$ or $\{Y_1>39\}$; these are statistical events, not
prompts. The measure $\mathbb P$ assigns probabilities to them.

\noindent\textbf{Reference posterior estimand.}
The estimand is the conditional distribution of the latent state given the
evidence. For $B\in\mathcal B(\mathcal X)$, write
\[
 \Pi^\star(B\mid y)=\mathbb P(X\in B\mid Y=y),
\]
and collect its probabilities on the $K$ declared states as
\[
 \boldsymbol\pi^\star(y)
 =\bigl(\Pi^\star(\{x_1\}\mid y),\ldots,
         \Pi^\star(\{x_K\}\mid y)\bigr)\in\Delta_K.
\]
Here
\[
 \Delta_K=\left\{\boldsymbol z\in[0,1]^K:
                    \sum_{k=1}^K z_k=1\right\}
\]
is the $K$-state probability simplex.
Thus $\Pi^\star$ is the reference conditional probability kernel and
$\boldsymbol\pi^\star(y)$ is the corresponding finite-dimensional posterior
vector. Both are defined by the scientific probability model and the observed
evidence; neither is defined by the fitted language model or by the wording
used to present that evidence.

\BeliefPromptsHeading
The reference posterior is not observed directly through the fitted language
model. Instead, the evidence is presented under a declared experimental
design, and the observable output is a conditional probability distribution
over verbal responses. This section constructs that indirect measurement in
the order evidence $\rightarrow$ context $\rightarrow$ response law
$\rightarrow$ semantic state probabilities.

\noindent\textbf{Fitted model and context.}
In this experiment, the mathematical input to the fitted language model
$\mathcal M_{\widehat\theta}$ is a complete context $c\in\mathcal C$,
referred to informally as a \emph{prompt}. The fitted parameters
$\widehat\theta$ remain fixed: the model's training procedure, parameters and
internal representations lie outside the experimental scope, while its
observable conditional probability laws are the measurements studied here.
The contexts are generated by the family of measurable maps
\[
 C_u:(\mathcal Y,\mathscr Y)
 \longrightarrow(\mathcal C,\mathscr C),
 \qquad u\in\mathcal U,
\]
where $C_u$ inserts evidence into presentation rule $u$ and returns the
context supplied to the fitted language model. Here
$(\mathcal Y,\mathscr Y)$ is the evidence space,
$(\mathcal C,\mathscr C)$ is the space of admissible prompts or complete
contexts, and $\mathcal U$ is a prespecified finite index set of presentation
rules,
\[
 \mathcal U=\{u_1,\ldots,u_m\}.
\]
In this paper, $(\mathcal Y,\mathscr Y)$ is a standard Borel evidence space,
so $\mathscr Y=\mathcal B(\mathcal Y)$, while the finite or countable raw
context space $\mathcal C$ has the discrete topology and hence
$\mathscr C=\mathcal B(\mathcal C)=\mathcal P(\mathcal C)$.

\noindent\textbf{Presentation rule.}
Note that an element $u\in\mathcal U$ is not the realised prompt string
itself. It simply identifies a complete rule for constructing that string,
including the task
instruction, placement and formatting of evidence, ordering of fields, and
permitted response phrases. Each $C_u$ is $\mathscr Y$-measurable: the
collection of evidence values that produce any declared context event must be
a measurable evidence event.

\noindent\textbf{Minimal example.} Take $\mathcal Y=\mathbb R$ with
$\mathscr Y=\mathcal B(\mathbb R)$ and let
\[
 \mathcal C
 =\{\text{high-temperature context},
      \text{ordinary-temperature context}\}
\]
with the discrete $\sigma$-algebra $\mathscr C=\mathcal P(\mathcal C)$, where
$\mathcal P(\mathcal C)$ denotes the power set, or collection of all subsets,
of $\mathcal C$. If
\[
 C_u(y)=
 \begin{cases}
 \text{high-temperature context},&y\geq39,\\
 \text{ordinary-temperature context},&y<39,
 \end{cases}
\]
then
\[
 C_u^{-1}\{\text{high-temperature context}\}=[39,\infty)
 \in\mathcal B(\mathbb R),
\]
and the analogous inverse image for the ordinary context is
$(-\infty,39)$. Hence $C_u$ is measurable.

Since $Y$ is random and $C_u$ is measurable, $C_u(Y)$ is a context-valued
random variable. For scenario $i$, the realised context is
\[
 c_{u,i}=C_u(y_i)\in\mathcal C.
\]

In the baseline experiment, prompts have no sampling space: the rules in
$\mathcal U$ are fixed experimental-design settings and are deliberately
crossed with evidence scenarios.

\BeliefContinuationHeading
The preceding subsection specifies what is supplied to the fitted model. We
now specify what is observed from it: a probability law over complete verbal
responses, followed by the prespecified grouping that assigns those responses
to scientific states.

Before relating language probabilities to scientific states, we must specify
the model's native probabilistic output and the measurable space on which its
law is defined.
So far, we have specified the evidence and prompts supplied as inputs to the
fitted language model, together with the scientific states about which we seek
inference. We now specify the model's language output and, more particularly,
the conditional probability assigned to each possible next token.
Let $\mathcal V$ denote the fitted language model's finite or countable token
vocabulary. Given a context $c\in\mathcal C$, the model supplies successive
conditional token probabilities. If a complete verbal continuation is the
finite token sequence
\[
 r=(v_1,\ldots,v_m)\in\mathcal V^{*},
\]
including its prescribed termination, then its conditional probability is
determined by the successive next-token distributions. Thus the observable
language probabilities concern token sequences, whereas our inferential
target is a probability distribution over the scientific state space
$\mathcal X$. Herein lies the gap that this paper seeks to fill.

The bridge has two stages. First, the successive token probabilities are
assembled into a prompt-dependent probability law over complete verbal
continuations, referred to informally below as \emph{phrases}. Second, those
continuations are grouped according to their
scientific meaning and the resulting law is transferred to the state space.
For the first stage, introduce the continuation space
$\mathcal R\subseteq\mathcal V^{*}$ containing the task-relevant complete
verbal continuations. It is the domain of the observable continuation law.

We endow the finite or countable set $\mathcal R$ with the discrete topology
$\tau_{\mathcal R}=\mathcal P(\mathcal R)$, and its Borel $\sigma$-algebra is
$\mathscr R:=\mathcal B(\mathcal R)=\mathcal P(\mathcal R)$.
Thus $(\mathcal R,\mathscr R)$ is the measurable continuation space used in
the probability model; no additional metric structure is required.

\noindent\textbf{Conditional continuation law.}
For each presentation rule $u\in\mathcal U$, the fitted model defines the
conditional continuation kernel
\[
 Q_u:\mathcal C\times\mathscr R\longrightarrow[0,1],
 \qquad (c,A)\longmapsto Q_u(A\mid c),
\]
from $(\mathcal C,\mathscr C)$ to $(\mathcal R,\mathscr R)$. Hence, for every
fixed context $c\in\mathcal C$, the conditional continuation law
\[
 Q_u(\cdot\mid c):\mathscr R\longrightarrow[0,1]
\]
is a probability measure. In particular, $Q_u(A\mid c)$ is the probability
that the fitted model, when supplied the input context $c$, generates a
continuation belonging to the phrase event $A\in\mathscr R$; for a single
phrase $r\in\mathcal R$, this is $Q_u(\{r\}\mid c)$. For every fixed
$A\in\mathscr R$, the map
$c\mapsto Q_u(A\mid c)$ is $\mathscr C$-measurable.
Each singleton mass $Q_u(\{r\}\mid c)$ is obtained by chaining the fitted
model's next-token probabilities along the continuation, with the probability
of each token conditioned on $c$ and on all preceding tokens in $r$. Since
$\mathcal R$ is countable, for every $A\in\mathscr R$,
\[
 Q_u(A\mid c)=\sum_{r\in A}Q_u(\{r\}\mid c).
\]
Consequently,
$\bigl(\mathcal R,\mathscr R,Q_u(\cdot\mid c)\bigr)$ is a probability space
for every fixed $(u,c)$. With realised evidence $y_i$, the fitted model
$\mathcal M_{\widehat\theta}$ supplies
$Q_u(\cdot\mid C_u(y_i))$. The language-model index is suppressed because
$\mathcal M_{\widehat\theta}$ is fixed. Since the context
$C_u(y_i)$ depends on the presentation rule, this continuation law generally
remains prompt-dependent even when the evidence $y_i$ is unchanged.

For this example, let
\[
 \mathcal R=\{r_1,r_2,r_3\}
 =\{\text{urgent review},\text{immediate escalation},
      \text{routine follow-up}\}.
\]
Hold fixed a patient record $y_i$ and present it first through
an evidence-first rule $u_1$ and then through an information-equivalent
question-first rule $u_2$. For the same three continuations, the fitted model
might return
\[
\begin{array}{c|ccc}
 & \text{urgent review} & \text{immediate escalation}
 & \text{routine follow-up}\\ \hline
 Q_{u_1}(\cdot\mid C_{u_1}(y_i)) & 0.45&0.15&0.40\\
 Q_{u_2}(\cdot\mid C_{u_2}(y_i)) & 0.20&0.40&0.40
\end{array}
\]
The evidence and scientific question are unchanged, but the two conditional
continuation laws differ because the contexts constructed by $u_1$ and $u_2$
use different wording or ordering.

\noindent\textbf{State augmentation.}
The second stage begins by accommodating a continuation that does not express
a declared state. Augment
the state space by one unexpressed category,
\[
 \mathcal X_0=\mathcal X\cup\{x_0\},
\]
where $x_0$ denotes \emph{unexpressed by the declared states}.

\noindent\textbf{Semantic map.}
The continuation law assigns probabilities to verbal phrases, whereas the
inferential target assigns probabilities to scientific states. Before these
phrase probabilities can be transferred to the state space, we must specify
which phrases express the same scientific meaning and how out-of-scope phrases
are treated. This semantic structure is introduced by the prespecified
$\mathscr R$-measurable map
\[
 \phi:(\mathcal R,\mathscr R)\longrightarrow
       (\mathcal X_0,\mathcal B(\mathcal X_0)).
\]
The map groups meaning-equivalent phrases and sends irrelevant, ambiguous or
out-of-scope phrases to $x_0$. For example,
$\phi(r_1)=\phi(r_2)=\text{urgent}$ and
$\phi(r_3)=\text{routine}$, so
$\phi^{-1}\{\text{urgent}\}=\{r_1,r_2\}\in\mathscr R$.
Once $\phi$ has been specified, the probabilities of meaning-equivalent
phrases can be aggregated into probabilities over the declared scientific
states. Because the underlying continuation law may depend on $u$, these
aggregated probabilities may also remain prompt-dependent. By contrast, the
reference posterior
$\Pi^\star(\cdot\mid y)$ is defined by the evidence and scientific model and
does not depend on how the same evidence is worded. It is, however,
evidence-dependent. Establishing when the
language-derived probabilities recover that prompt-independent target is the
calibration and stability problem studied in the paper.

Figure~\ref{fig:semantic-bridge} previews the complete construction. The upper path is observable; the lower path defines the independent inferential target against which the calibrated estimate is evaluated.

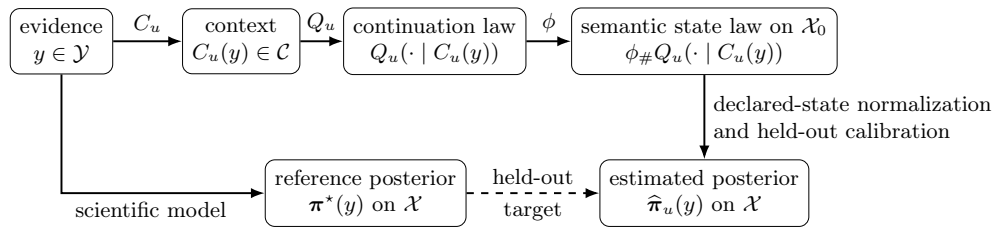
\begin{figure}[b]
\centering
\resizebox{0.80\linewidth}{!}{%
\begin{tikzpicture}[
  box/.style={draw,rounded corners,align=center,inner sep=4pt,font=\small},
  arr/.style={-{Latex[length=2mm]},thick},
  every node/.style={font=\small}
]
\node[box] (e) at (0,0) {evidence\\$y\in\mathcal Y$};
\node[box] (c) at (2.6,0) {context\\$C_u(y)\in\mathcal C$};
\node[box] (q) at (5.4,0) {continuation law\\$Q_u(\cdot\mid C_u(y))$};
\node[box] (s) at (9.3,0) {semantic state law on $\mathcal X_0$\\$\phi_\#Q_u(\cdot\mid C_u(y))$};
\node[box] (hat) at (9.3,-2.2) {estimated posterior\\$\widehat{\boldsymbol\pi}_u(y)$ on $\mathcal X$};
\node[box] (pi) at (4.4,-2.2) {reference posterior\\$\boldsymbol\pi^\star(y)$ on $\mathcal X$};
\draw[arr] (e) -- node[above] {$C_u$} (c);
\draw[arr] (c) -- node[above] {$Q_u$} (q);
\draw[arr] (q) -- node[above] {$\phi$} (s);
\draw[arr] (s) -- node[right,align=left] {declared-state normalization\\and held-out calibration} (hat);
\draw[arr,dashed] (pi) -- node[above] {held-out} node[below] {target} (hat);
\draw[arr] (e.south) |- node[pos=.72,below] {scientific model} (pi.west);
\end{tikzpicture}
}
\caption{From evidence to a posterior estimate. The semantic map transfers the continuation law to declared states; normalization and held-out calibration are assessed against the independently defined reference posterior.}
\label{fig:semantic-bridge}
\end{figure}

\BeliefPushforwardHeading
The continuation law and semantic map have now been defined separately. Their
composition gives the observable state-level measurement used by the
subsequent statistical analysis.

Recall that the goal of the experiment is to recover a probability
distribution over the scientific states from the observable probabilities
over verbal continuations. Given evidence $y$ and presentation rule $u$, the
conditional continuation kernel $Q_u$ defined above supplies the
continuation measure
$Q_u(\cdot\mid C_u(y))$ on $(\mathcal R,\mathscr R)$. The semantic map
$\phi$ transfers this measure to $(\mathcal X_0,\mathcal B(\mathcal X_0))$ by
\[
 (\phi_\#Q_u)(B\mid C_u(y))
 =Q_u\!\left(\phi^{-1}(B)\mid C_u(y)\right),
 \qquad B\in\mathcal B(\mathcal X_0).
\]
As Figure~\ref{fig:semantic-bridge} shows, the semantic pushforward is an intermediate measurement; declared-state normalization and held-out calibration are still required before comparison with the reference posterior.

To see this numerically, fix $u$ and $c$ and write
$Q=Q_u(\cdot\mid c)$ for the probability measure on the measurable space
$(\mathcal R,\mathscr R)$. Suppose, for the phrases
$\{r_1,r_2,r_3\}$, that
\[
 Q(\{r_1\})=0.45,
 \qquad Q(\{r_2\})=0.15,
 \qquad Q(\{r_3\})=0.40.
\]
Here $Q$ represents the probability law assigned to complete verbal
continuations. It assigns probabilities to sets of phrases; it is not the
reference probability law $\mathbb P$ on scientific outcomes. The
pushforward $\phi_\#Q$ represents the resulting probability distribution over
declared meanings after phrases have been grouped by $\phi$.
Then, using the inverse image calculated in the preceding semantic-map
construction,
\[
\begin{aligned}
 (\phi_\#Q)(\{\text{urgent}\})
 &=Q\!\left(\phi^{-1}\{\text{urgent}\}\right)\\
 &=Q(\{r_1,r_2\})\\
 &=Q(\{r_1\})+Q(\{r_2\})\\
 &=0.45+0.15=0.60.
\end{aligned}
\]
The pushforward $\phi_\#Q$ is therefore a probability measure on
$(\mathcal X_0,\mathcal B(\mathcal X_0))$.

\BeliefMismatchHeading
The pushforward places the observable response law on the target state space
but does not identify the reference posterior. It is a prompt-dependent
semantic measurement: $\phi$ groups continuations by declared meaning, but
does not establish that their aggregated probabilities have a posterior
interpretation. By contrast, the target posterior is defined independently
from the evidence on the $K$ scientific states, whereas $\phi_\#Q_u$ also
retains the unexpressed state $x_0$. Keeping measurement and target separate
prevents circular calibration and makes recovery falsifiable on held-out
observations. This section therefore formulates the remaining inverse problem,
defines its estimator and states the conditions under which recovery is
defensible.

\BeliefCalibrationHeading
Having separated the observable semantic measurement from the estimand, we
can define the finite-dimensional calibration input and the resulting
posterior estimator.

Calibration requires an observable $K$-state input on the same simplex as the
target posterior. The semantic pushforward is instead defined on the augmented
space $\mathcal X_0$, which also contains the unexpressed category $x_0$.
We must therefore normalize the probability mass assigned to the $K$ declared
states before it can be used as the calibration input. The excluded mass is
retained separately as a coverage diagnostic. Define
\[
\begin{aligned}
 M_u(c)
 &=\sum_{k=1}^{K}Q_u\!\left(\phi^{-1}(\{x_k\})\mid c\right)\\
 &=1-Q_u\!\left(\phi^{-1}(\{x_0\})\mid c\right),\\[2pt]
 p_{u,k}(c)
 &=\frac{Q_u\!\left(\phi^{-1}(\{x_k\})\mid c\right)}{M_u(c)},
 \qquad k=1,\ldots,K.
\end{aligned}
\]
whenever $M_u(c)>0$, and write
$\boldsymbol p_u(c)=(p_{u,1}(c),\ldots,p_{u,K}(c))\in\Delta_K$.
Thus $x_0$ is deliberately omitted from the coordinates of
$\boldsymbol p_u(c)$. Equivalently, $\boldsymbol p_u(c)$ is the conditional
distribution over $\mathcal X$ given that the continuation is mapped into one
of the declared states rather than into $x_0$. Its excluded probability
$1-M_u(c)$ remains reported separately.
More precisely, let
\[
 \mathcal C_u^+=\{c\in\mathcal C:M_u(c)>0\},
 \qquad
 \mathscr C_u^+=\{D\cap\mathcal C_u^+:D\in\mathscr C\}.
\]
Because $Q_u$ is a probability kernel and $\phi$ is measurable, every map
$c\mapsto Q_u\{\phi^{-1}(\{x_k\})\mid c\}$ is $\mathscr C$-measurable. Hence
$M_u$ is measurable and
\[
 \boldsymbol p_u:(\mathcal C_u^+,\mathscr C_u^+)
 \longrightarrow(\Delta_K,\mathcal B(\Delta_K))
\]
is measurable. Its composition $\boldsymbol p_u\circ C_u$ is therefore
measurable on the corresponding evidence domain
$\{y\in\mathcal Y:C_u(y)\in\mathcal C_u^+\}$, equipped with the trace of
$\mathscr Y$. The map $\boldsymbol p_u$ is only an intermediate, observable
calibration input; it is not the target posterior. The expressed mass $M_u$
is retained as a coverage diagnostic.

Let $\mathcal P_u\subseteq\Delta_K$ denote the prespecified calibration
domain. The paper then estimates the held-out calibration map
\[
 \widehat g_u:(\mathcal P_u,\mathcal B(\mathcal P_u))
 \longrightarrow(\Delta_K,\mathcal B(\Delta_K)),
 \qquad
 \widehat{\boldsymbol\pi}_u(y)
 =\widehat g_u\!\left[\boldsymbol p_u\{C_u(y)\}\right],
\]
from the observable lexical composition to the reference posterior and
derives conditions under which recovery exists, is conditionally unique and
is stable. Identification requires the population forward relation to be
one-to-one, with adequate separation, on the declared operating domain; the
theory then asks whether $\widehat g_u$ approximates its inverse there with
controlled error.

Thus the complete posterior estimator includes prompt construction,
the continuation law, semantic grouping, declared-state normalization and
held-out calibration; it does not terminate at $\phi_\#Q_u$.
The vector $\boldsymbol p_u$ is therefore the calibration input,
$M_u$ is a coverage diagnostic, $\widehat{\boldsymbol\pi}_u$ is the estimated
posterior, and $\boldsymbol\pi^\star$ is the held-out calibration and
validation target.
Under the sampling assumption stated in the probability-space setup, the calibration sample
contains the paired measurements and targets
\[
 \left\{\left(
 \boldsymbol p_u\{C_u(Y_i)\},
 \boldsymbol\pi^\star(Y_i)
 \right):i=1,\ldots,n\right\}.
\]
Repeated language measurements for one scenario remain in one cluster and
are not treated as additional independent evidence samples.
The continuation-law and semantic-structure subsection introduces the
paper-specific measurement construction. The statistical-mismatch and
calibration subsections state the substantive inferential problem; the
remaining components use standard measurable-space, kernel and pushforward
constructions.

\BeliefNotationHeading

\BeliefNotationBegin
\begin{tabularx}{\textwidth}{@{}l l X@{}}
\toprule
Notation & Mathematical type & Meaning \\
\midrule
$(\Omega,\mathscr F,\mathbb P)$ & Probability space & Underlying reference probability model \\
$(\mathcal X,\mathcal B(\mathcal X))$ & Finite measurable space & Declared scientific state space \\
$(\mathcal Y,\mathscr Y)$ & Standard Borel space & Evidence space \\
$X$ & Measurable map $\Omega\to\mathcal X$ & Latent scientific state \\
$Y$ & Measurable map $\Omega\to\mathcal Y$ & Observed evidence \\
$\Pi^\star$ & Probability kernel & Reference conditional state distribution \\
$\boldsymbol\pi^\star(y)$ & Element of $\Delta_K$ & Reference state-probability vector \\
$C_u$ & Measurable map $\mathcal Y\to\mathcal C$ & Context constructed from evidence \\
$Q_u$ & Probability kernel & Conditional continuation law \\
$\phi$ & Measurable map $\mathcal R\to\mathcal X_0$ & Semantic map with unexpressed category $x_0$ \\
$\phi_\#Q_u$ & Pushforward measure & Continuation law expressed on semantic states \\
$M_u(c)$ & Scalar in $[0,1]$ & Probability mass covered by declared states \\
$\boldsymbol p_u$ & Measurable map $\mathcal C_u^+\to\Delta_K$ & Intermediate normalized calibration map \\
$\boldsymbol p_u(c)$ & Element of $\Delta_K$ & Realised calibration input at context $c$ \\
$\widehat g_u$ & Fitted Borel-measurable map $\mathcal P_u\to\Delta_K$ & Held-out calibration map \\
$\widehat{\boldsymbol\pi}_u(y)$ & Element of $\Delta_K$ & Estimated reference posterior \\
\bottomrule
\end{tabularx}
\BeliefNotationEnd

\ifdefined\BeliefLensSetupOnly
\def\BeliefLensAfterSetup{ }
\else
\def\BeliefLensAfterSetup{}
\fi
\BeliefLensAfterSetup

\section*{Extra notes: measure-theoretic background}

\section*{Calibration and verification}

The notation $\widehat g_u$ uses a hat because the calibration map is fitted
from the calibration sample. Before that sample is observed,
$\widehat g_u$ is a random function through its dependence on the random
sample. Conditional on the realised sample and the fixed fitting rule, it is
a deterministic Borel-measurable map from the declared calibration domain
$\mathcal P_u$ into $\Delta_K$.

The phrase \emph{affine log-ratio calibration} refers to the particular
finite-dimensional calibration class used in the reported experiment. A
probability vector lies on a simplex, so its coordinates cannot vary freely.
Choose one state as the reference coordinate and define the log-ratio
coordinate map (conventionally called the additive log-ratio transformation)
\[
 \boldsymbol\ell(\boldsymbol p)
 =\left(\log\frac{p_1}{p_K},\ldots,
         \log\frac{p_{K-1}}{p_K}\right)
 \in\mathbb R^{K-1},
 \qquad \boldsymbol p\in\operatorname{int}(\Delta_K).
\]
The fitted calibration is affine in these unconstrained coordinates:
\[
 \boldsymbol\ell\{\widehat g_u(\boldsymbol p)\}
 =\widehat A_u\boldsymbol\ell(\boldsymbol p)+\widehat{\boldsymbol b}_u,
\]
or, equivalently,
\[
 \widehat g_u(\boldsymbol p)
 =\boldsymbol\ell^{-1}
   \left\{\widehat A_u\boldsymbol\ell(\boldsymbol p)
          +\widehat{\boldsymbol b}_u\right\}.
\]
Here ``affine'' means a linear transformation plus an intercept. The log-ratio
transformation moves the interior of the simplex to Euclidean coordinates;
the inverse transformation returns a valid probability vector. This fitted
map is continuous, and hence Borel measurable, on the interior of its
calibration domain. Boundary probabilities require a separately declared
zero-handling rule and are not silently covered by the displayed formula.

No global injectivity, surjectivity or bijectivity of $\widehat g_u$ is
assumed. For the displayed affine class, bijectivity on the simplex interior
would require $\widehat A_u$ to have full rank, but fitting alone does not
guarantee that condition. More importantly, the paper's identification claim
concerns the population forward relation on a declared operating domain, not
global invertibility of a fitted map. Verification therefore uses untouched
observations to assess recovery error and predictive coverage, while
conditioning diagnostics assess whether relevant posterior directions remain
distinguishable.

\section*{1. Sets, points and inverse images}

Let $S$ and $T$ be sets and let $f:S\to T$ be a function. For a set
$B\subseteq T$, the \emph{inverse image} is
\[
   f^{-1}(B)=\{s\in S:f(s)\in B\}\subseteq S.
\]
It is the event in the input space corresponding to the statement
``$f(s)$ lies in $B$.'' Inverse images, rather than numerical derivatives or
matrix inverses, are the basic operation used to define measurability.

\noindent\textbf{Running example.} Let $S=T=\mathbb R$, let
$f:S\to T$ be $f(s)=s^2$, and let $B=[1,4]\subseteq\mathbb R$. Then
\[
   f^{-1}(B)=[-2,-1]\cup[1,2].
\]
The set $B$ belongs to the codomain, whereas $f^{-1}(B)$ belongs to the
domain. Both are Borel subsets of $\mathbb R$.

\noindent\textbf{Paper connection and added complexity.} The paper replaces
$f$ by a semantic map $\phi_u$ from complete verbal continuations to declared
states. Then $\phi_u^{-1}(\{x_k\})$ is the collection of phrases assigned to
state $x_k$. Unlike $s\mapsto s^2$, this map records a scientific judgement
about meaning, so it must be declared and validated.

\section*{2. Sigma-algebras}

We write $\mathcal P(S)$ for the \emph{power set}, the collection of all
subsets of a set $S$. This script $\mathcal P$ has nothing to do with the
probability measure $\mathbb P$. For example, if $S=\{a,b\}$, then
\[
  \mathcal P(S)=\{\varnothing,\{a\},\{b\},\{a,b\}\}.
\]
Although the power set is occasionally useful for comparison, it is not needed
for the substantive argument below. We shall instead work directly with a declared
$\sigma$-algebra $\mathscr S$ and use $\mathcal B(S)$ for the Borel
$\sigma$-algebra when $S$ has a topology.

A \emph{$\sigma$-algebra} $\mathscr S$ on a set $S$ is a collection of subsets
of $S$ satisfying:
\begin{enumerate}[label=(\roman*)]
\item $S\in\mathscr S$;
\item if $A\in\mathscr S$, then $S\setminus A\in\mathscr S$;
\item if $A_1,A_2,\ldots\in\mathscr S$, then
      $\bigcup_{j=1}^{\infty}A_j\in\mathscr S$.
\end{enumerate}
The pair $(S,\mathscr S)$ is called a \emph{measurable space}. The sets in
$\mathscr S$ are the measurable sets or events. The axioms ensure that
negations, countable alternatives and countable intersections of events remain
measurable.

\noindent\textbf{Worked example.} Let $S=\{a,b\}$ and
\[
 \mathscr S=\{\varnothing,\{a\},\{b\},S\}.
\]
The complement of $\{a\}$ is $\{b\}$, and every union of members of
$\mathscr S$ remains in $\mathscr S$; hence this is a $\sigma$-algebra.
By contrast, $\{\varnothing,\{a\},S\}$ is not a $\sigma$-algebra because it
omits the complement $\{b\}$.

The notation $\sigma(\mathcal G)$ has a different meaning. If
$\mathcal G$ is a collection of subsets of $S$, then
\[
 \sigma(\mathcal G)
 =\bigcap\{\mathscr S:\mathscr S\text{ is a $\sigma$-algebra on $S$ and }
                         \mathcal G\subseteq\mathscr S\}
\]
is the smallest $\sigma$-algebra containing $\mathcal G$. The operator
$\sigma(\cdot)$ therefore acts on a \emph{collection of subsets}, not usually
on the underlying set of points. Writing $\mathscr F=\sigma(S)$ is consequently
ambiguous. The standard forms are
\[
 \mathscr F=\sigma(\mathcal G)
 \quad\text{for a collection of generators }\mathcal G,
 \qquad\text{or}\qquad
 \mathcal B(S)=\sigma\{G\subseteq S:G\text{ is open}\}.
\]
For $S=\mathbb R$, the Borel $\sigma$-algebra can equivalently be generated by
the open intervals:
\[
\mathcal B(\mathbb R)
 =\sigma\{(a,b):a<b,\ a,b\in\mathbb R\}.
\]

\noindent\textbf{Worked example of generation.} Let
$S=\{a,b,c\}$ and $\mathcal G=\{\{a\}\}$. Closure under complements and
unions gives
\[
 \sigma(\mathcal G)=\{\varnothing,\{a\},\{b,c\},S\}.
\]
This $\sigma$-algebra records whether the outcome is $a$, but it does not
distinguish $b$ from $c$.

\noindent\textbf{Why is $(S,\mathscr S)$ called a measurable space when no
measure has yet been specified?} This terminology is easy to misread. The
object $S$ is the underlying set, and $\mathscr S$ is a $\sigma$-algebra of
subsets of $S$. The pair $(S,\mathscr S)$ is called a \emph{measurable space}
because it declares which subsets are eligible to be measured; it does not yet
assign them numerical sizes. A particular set $A\subseteq S$ is called a
\emph{measurable set} when $A\in\mathscr S$. A numerical measure is introduced
only in the triple $(S,\mathscr S,\mu)$. Thus the standard terminology
separates three structures:
\[
\begin{array}{lll}
(S,\mathscr S) & \text{measurable space} &
   \text{declares the measurable sets},\\[2pt]
(S,\mathscr S,\mu) & \text{measure space} &
   \text{also assigns a measure }\mu,\\[2pt]
(S,\mathscr S,\mathbb P) & \text{probability space} &
   \text{has }\mathbb P(S)=1.
\end{array}
\]
The pair is therefore intentional. A triple is required only after a measure
or probability measure has been introduced.

No topology is required to form a measurable space. For example, with
$S=\{a,b,c\}$ and
\[
 \mathscr S=\{\varnothing,\{a\},\{b,c\},S\},
\]
the pair $(S,\mathscr S)$ is already a measurable space. The set $\{a\}$ is a
measurable set because $\{a\}\in\mathscr S$, whereas $\{b\}$ is not measurable
for this choice because $\{b\}\notin\mathscr S$.

A topology is needed only when we want to construct the particular
$\sigma$-algebra called the Borel $\sigma$-algebra. If $(S,\tau)$ is a
topological space, then
\[
 \mathcal B(S)=\sigma(\tau),
\]
the smallest $\sigma$-algebra containing the open sets $\tau$. Hence
$(S,\mathcal B(S))$ is a measurable space obtained from the topology, but a
general measurable space $(S,\mathscr S)$ need not arise this way.

\noindent\textbf{Paper connection and added complexity.} The paper uses
several measurable spaces: the sample space $\Omega$, finite state space
$\mathcal X$, evidence space $\mathcal Y$, context space $\mathcal C$, and
verbal-continuation space $\mathcal R$. They do not need one common topology or
$\sigma$-algebra; each receives the structure appropriate to its elements.

\noindent\textbf{Finite example.} If
$\mathcal X=\{x_1,x_2,x_3\}$ has the discrete topology, then
$\mathcal B(\mathcal X)$ contains all eight subsets of $\mathcal X$. This is
the $\sigma$-algebra used for the finite state space in the paper.

\noindent\textbf{Why not always use every subset?} For finite and countable
discrete spaces, the power set is natural. On an uncountable space such as
$\mathbb R$, assigning probabilities consistently to every conceivable subset
causes mathematical difficulties. Probability is therefore normally defined
on the Borel $\sigma$-algebra.

\section*{3. Borel sets and Borel measurable spaces}

This section slows the construction down because two objects are easily
confused:
\[
 \boxed{B\in\mathcal B(\mathbb R)}
 \qquad\text{versus}\qquad
 \boxed{\mathcal B(\mathbb R)}.
\]
Here $B$ is one set of real numbers. By contrast,
$\mathcal B(\mathbb R)$ is a collection of sets: the Borel $\sigma$-algebra.

\subsection*{3.0 What is the topology of $\mathbb R$?}

A \emph{topology} $\tau$ on a set $S$ is a collection of subsets, called open
sets, satisfying:
\begin{enumerate}[label=(\roman*)]
\item $\varnothing,S\in\tau$;
\item an arbitrary union of members of $\tau$ belongs to $\tau$; and
\item a finite intersection of members of $\tau$ belongs to $\tau$.
\end{enumerate}
The pair $(S,\tau)$ is a topological space. Unlike a $\sigma$-algebra, a
topology need not be closed under complements or countably infinite
intersections.

The \emph{standard topology} on $\mathbb R$, denoted here by
$\tau_{\mathbb R}$, is generated by the open intervals. Equivalently, a set
$G\subseteq\mathbb R$ is open when every $x\in G$ has some radius
$\varepsilon>0$ such that
\[
 (x-\varepsilon,x+\varepsilon)\subseteq G.
\]
For example, $G=(1,4)$ is open. If $x=2$, one may take
$\varepsilon=1/2$, since
\[
 (2-1/2,2+1/2)=(1.5,2.5)\subseteq(1,4).
\]
For a point very near the boundary, say $x=1.1$, one may take
$\varepsilon=0.05$. By contrast, $[1,4]$ is not open: at the endpoint $x=1$,
every interval $(1-\varepsilon,1+\varepsilon)$ contains points below $1$ and
therefore cannot lie inside $[1,4]$.

The topology expresses the ordinary geometry of the real line: which points
are near each other, which sets are neighbourhoods, and which functions are
continuous. For instance, $f(s)=s^2$ is continuous because sufficiently small
changes in $s$ produce small changes in $f(s)$; equivalently, the inverse image
of every open set is open.

The Borel $\sigma$-algebra is obtained by placing the $\sigma$-algebra rules
around this topology:
\[
 \boxed{\mathcal B(\mathbb R)=\sigma(\tau_{\mathbb R}).}
\]
The relationship can therefore be read from left to right as
\[
 \boxed{\text{open intervals}}
 \quad\longrightarrow\quad
 \boxed{\text{standard topology }\tau_{\mathbb R}}
 \quad\longrightarrow\quad
 \boxed{\text{Borel }\sigma\text{-algebra }\mathcal B(\mathbb R)}.
\]
In this precise sense, the Borel construction is related to the topology of
$\mathbb R$: open intervals generate its usual open sets, and those open sets
in turn generate its Borel $\sigma$-algebra.  The topology itself contains
only the open sets; the generated $\sigma$-algebra is larger because it must
also be closed under complements and countable set operations.  For example,
although $[a,b]$ is not open, it is Borel because
\[
 [a,b]
 =\mathbb R\setminus\bigl(( -\infty,a)\cup(b,\infty)\bigr).
\]
Thus every open set becomes a measurable event, as do its complement and all
countable unions and intersections of such sets. For a real-valued random
variable $Y$, ordinary numerical questions are therefore measurable. Examples
include
\[
 \{Y<4\}=Y^{-1}(( -\infty,4)),
\]
\[
 \{1\leq Y\leq4\}=Y^{-1}([1,4]),
\]
and, for a proposed value $y_0$ and tolerance $\varepsilon>0$,
\[
 \{|Y-y_0|<\varepsilon\}
 =Y^{-1}((y_0-\varepsilon,y_0+\varepsilon)).
\]
This is what it means to say that Borel sets connect probability to the
topology of $\mathbb R$: the sets naturally used to describe numerical
proximity, thresholds and ranges are all admitted as events to which
probabilities may be assigned.

\subsection*{3.1 An informal construction recipe}

A \emph{Borel set} is any subset of $\mathbb R$ obtainable from open intervals
using the operations permitted by a $\sigma$-algebra:
\begin{enumerate}[label=\textbf{\arabic*.},leftmargin=*,itemsep=4pt]
\item begin with open intervals such as $(1,4)$;
\item take complements;
\item take countable unions; and
\item take countable intersections.
\end{enumerate}
These operations may be repeated. The collection of every set obtainable in
this way is $\mathcal B(\mathbb R)$.

\noindent\textbf{Example 1: an open interval.}
\[
 (1,4)\in\mathcal B(\mathbb R)
\]
by definition, because open intervals are the starting sets.

\noindent\textbf{Example 2: a closed interval.} Since
\[
 [1,4]
 =\mathbb R\setminus\bigl((-\infty,1)\cup(4,\infty)\bigr),
\]
the set $[1,4]$ is Borel: it is the complement of a union of open sets.

\noindent\textbf{Example 3: a single point.} The singleton $\{2\}$ is Borel
because
\[
 \{2\}
 =\bigcap_{n=1}^{\infty}(2-1/n,2+1/n).
\]
Each interval is open, and the intersection is countable.

\noindent\textbf{Example 4: the integers.} Since every singleton is Borel,
\[
 \mathbb Z=\bigcup_{k\in\mathbb Z}\{k\}
\]
is Borel. The union is countable because $\mathbb Z$ is countable.

\noindent\textbf{Example 5: the rational and irrational numbers.} Similarly,
\[
 \mathbb Q=\bigcup_{q\in\mathbb Q}\{q\}
 \in\mathcal B(\mathbb R).
\]
Its complement
\[
 \mathbb R\setminus\mathbb Q
\]
is therefore also Borel. Thus a Borel set need not resemble a single interval:
the rationals and irrationals are both spread throughout the real line.

\noindent\textbf{Example 6: a union of separated intervals.}
\[
 [-2,-1]\cup[1,2]\in\mathcal B(\mathbb R)
\]
because each closed interval is Borel and a finite union is a countable union.
This is the inverse image of $[1,4]$ under $f(s)=s^2$.

\noindent\textbf{What is not Borel?} Non-Borel subsets of $\mathbb R$ exist,
but they cannot be produced from open intervals by the preceding countable
operations. Such sets are necessarily much less explicit than the events
normally used in elementary statistics. The important practical point is not
to identify one by formula; it is that $\mathcal B(\mathbb R)$ is an extremely
rich, but not all-inclusive, collection of subsets.

\noindent\textbf{Why does $\mathcal B(\mathbb R)$ not contain every subset of
$\mathbb R$?} Every singleton $\{x\}$ is Borel, and every subset
$A\subseteq\mathbb R$ can be written formally as
\[
 A=\bigcup_{x\in A}\{x\}.
\]
If $A$ is countable, this is a countable union, so $A$ is Borel. If $A$ is
uncountable, however, the displayed union is generally uncountable. A
$\sigma$-algebra is required to be closed under \emph{countable} unions, not
arbitrary uncountable unions. Consequently, this representation does not show
that an arbitrary uncountable subset is Borel.

There is also a counting argument. The real line has cardinality
$\mathfrak c$. Because open subsets of $\mathbb R$ can be described using
countably many intervals with rational endpoints, there are only
$\mathfrak c$ Borel sets, even after repeated countable $\sigma$-algebra
operations. The power set $\mathcal P(\mathbb R)$ has strictly larger
cardinality $2^{\mathfrak c}$. Hence
\[
 \mathcal B(\mathbb R)\subsetneq\mathcal P(\mathbb R),
\]
so most subsets of $\mathbb R$ are not Borel.

The restriction to countable operations matches countable additivity of
probability measures and is essential for limit arguments in probability.
The Lebesgue $\sigma$-algebra enlarges $\mathcal B(\mathbb R)$ by adding
subsets of null sets, but it still does not contain every subset of
$\mathbb R$.

\subsection*{3.2 Starting from open sets}

The real line already has a notion of openness. For example,
\[
 (1,4)=\{s\in\mathbb R:1<s<4\}
\]
is open. The Borel $\sigma$-algebra is defined by
\[
 \mathcal B(\mathbb R)
 =\sigma\{G\subseteq\mathbb R:G\text{ is open}\}.
\]
In words, begin with every open subset of $\mathbb R$ and add every set
required by the $\sigma$-algebra rules: complements, countable unions and
countable intersections. The resulting collection is
$\mathcal B(\mathbb R)$.

\noindent\textbf{Running example 1: an open Borel set.} Because $(1,4)$ is
one of the generating open sets,
\[
 (1,4)\in\mathcal B(\mathbb R).
\]
This expression says that the interval $(1,4)$ is one member of the much
larger collection $\mathcal B(\mathbb R)$.

\subsection*{3.3 Obtaining closed sets by complementation}

Every $\sigma$-algebra is closed under complements. Since
$(-\infty,1)\cup(4,\infty)$ is open, its complement is Borel:
\[
 \mathbb R\setminus\bigl[(-\infty,1)\cup(4,\infty)\bigr]=[1,4].
\]
Therefore
\[
 B=[1,4]\in\mathcal B(\mathbb R).
\]
This is why every closed subset of $\mathbb R$ is Borel, even though the
definition began with open sets.

\noindent\textbf{Running example 2: half-open intervals.} The interval
$[1,4)$ is also Borel because
\[
 [1,4)= [1,4]\cap(-\infty,4),
\]
an intersection of two Borel sets. Thus open, closed and half-open intervals
are all members of $\mathcal B(\mathbb R)$.

\subsection*{3.4 Obtaining more complicated sets}

Singletons are closed, so $\{r\}\in\mathcal B(\mathbb R)$ for every
$r\in\mathbb R$. A countable set is a countable union of singletons and is
therefore Borel. For example,
\[
 \mathbb Q=\bigcup_{r\in\mathbb Q}\{r\}
 \in\mathcal B(\mathbb R).
\]
Similarly,
\[
 [-2,-1]\cup[1,2]\in\mathcal B(\mathbb R)
\]
because it is a union of two closed, hence Borel, intervals.

\noindent\textbf{What this does not say.} It does not say that
$\mathcal B(\mathbb R)$ contains literally every subset of $\mathbb R$.
It contains those obtained from open sets through the countable operations
allowed by a $\sigma$-algebra. That collection is already rich enough for the
ordinary numerical events used in statistics.

\subsection*{3.5 Returning to $f(s)=s^2$}

Let $S=T=\mathbb R$, let $f:S\to T$ be $f(s)=s^2$, and let
$B=[1,4]\in\mathcal B(T)$. The inverse image is
\[
\begin{aligned}
 f^{-1}(B)
 &=\{s\in\mathbb R:f(s)\in[1,4]\}\\
 &=\{s\in\mathbb R:1\leq s^2\leq4\}\\
 &=[-2,-1]\cup[1,2].
\end{aligned}
\]
The answer belongs to $\mathcal B(S)$ because it is a union of closed
intervals. Notice the direction: $B$ is a set in the output space $T$, while
$f^{-1}(B)$ is the corresponding set in the input space $S$.

The same calculation can be repeated for other output events:
\[
\begin{array}{c|c}
\text{Borel output set }B\subseteq T &
\text{Borel inverse image }f^{-1}(B)\subseteq S\\ \hline
(-\infty,4) & (-2,2)\\
(1,4) & (-2,-1)\cup(1,2)\\
\{4\} & \{-2,2\}\\
(-\infty,0) & \varnothing\\[2pt]
[0,\infty) & \mathbb R
\end{array}
\]
Every set in the right-hand column is Borel. This is the concrete content of
the statement that $f$ is Borel measurable.

\subsection*{3.6 The same idea in several dimensions}

For a $d$-dimensional numerical observation, the usual measurable space is
\[
 (\mathbb R^d,\mathcal B(\mathbb R^d)).
\]
For example, let $\mathcal Y=\mathbb R^2$, where
$y=(y_1,y_2)$ records temperature and oxygen saturation. The rectangle
\[
 D=(39,\infty)\times(0,92)
\]
is Borel and represents ``temperature above $39$ and oxygen saturation below
$92$.'' If $Y:\Omega\to\mathbb R^2$ is the recorded observation, then
\[
 Y^{-1}(D)
 =\{\omega\in\Omega:Y_1(\omega)>39\text{ and }Y_2(\omega)<92\}
\]
is the corresponding event in the underlying sample space.

A measurable space $(S,\mathscr S)$ is \emph{standard Borel} if it is
measurably isomorphic to a Borel subset of a complete separable metric space.
Informally, it has the same measurable structure as a sufficiently regular
subset of a familiar metric space. Euclidean spaces, finite spaces, countable
discrete spaces and many function spaces used in statistics are standard
Borel.

In this paper, the evidence space is written
\[
   (\mathcal Y,\mathscr Y),\qquad
   \mathscr Y=\mathcal B(\mathcal Y),
\]
after choosing a standard Borel representation of $\mathcal Y$. This
assumption is broad enough for ordinary statistical applications and ensures
the existence of regular conditional distributions.

The preceding $\mathbb R$ and $\mathbb R^2$ examples are standard Borel
spaces. For the present paper, the practical point is that ordinary finite
states and Euclidean evidence spaces possess the regular measurable structure
needed later.

\noindent\textbf{Checkpoint.} In the running example:
\begin{enumerate}[label=(\roman*)]
\item $S=T=\mathbb R$ are the underlying sets;
\item $\mathcal B(S)$ and $\mathcal B(T)$ are collections of admissible sets;
\item $B=[1,4]$ is one member of $\mathcal B(T)$;
\item $f^{-1}(B)=[-2,-1]\cup[1,2]$ is one member of $\mathcal B(S)$; and
\item the fact that this works for every $B\in\mathcal B(T)$ makes $f$
Borel measurable.
\end{enumerate}

\subsection*{3.7 From numbers to categories, words and phrases}

Measure theory does not require the possible observations to be numbers. Its
starting point is simply a set $S$ together with a $\sigma$-algebra
$\mathscr S$. Numerical spaces such as $\mathbb R$ are only one example.

\noindent\textbf{Worked categorical example.} Let
\[
 \mathcal X=\{\text{urgent},\text{routine},\text{insufficient}\}.
\]
Because this is a finite set, we use the discrete $\sigma$-algebra containing
every subset:
\[
\begin{aligned}
 \mathcal B(\mathcal X)=\{&\varnothing,
 \{\text{urgent}\},\{\text{routine}\},\{\text{insufficient}\},\\
 &\{\text{urgent},\text{routine}\},
 \{\text{urgent},\text{insufficient}\},
 \{\text{routine},\text{insufficient}\},\mathcal X\}.
\end{aligned}
\]
Thus ``the state is urgent'' and ``the state is urgent or routine'' are both
measurable events. A probability measure might assign
\[
 \mathbb P_X(\{\text{urgent}\})=0.60,
 \quad \mathbb P_X(\{\text{routine}\})=0.30,
 \quad \mathbb P_X(\{\text{insufficient}\})=0.10.
\]
Nothing in this construction requires the category names to be numerical.

\noindent\textbf{Worked word example.} Let the vocabulary be
\[
 \mathcal V=\{\text{urgent},\text{immediate},\text{routine},\text{review},
 \text{escalation}\}.
\]
Again, $\mathcal V$ is finite and carries the discrete $\sigma$-algebra. If
$W$ is the next-word random variable, the event
\[
 A=\{\text{urgent},\text{immediate}\}\subseteq\mathcal V
\]
means ``the next word is either urgent or immediate.'' If
\[
 \mathbb P(W=\text{urgent})=0.35,
 \qquad \mathbb P(W=\text{immediate})=0.20,
\]
then
\[
 \mathbb P(W\in A)=0.35+0.20=0.55.
\]

\noindent\textbf{Worked phrase example.} A complete phrase is an element of a
set $\mathcal R$ of finite word sequences. For a deliberately small example, let
\[
 \mathcal R=\{\text{urgent review},\text{immediate escalation},
             \text{routine follow-up}\}.
\]
With the discrete $\sigma$-algebra on $\mathcal R$, the set
\[
 A_{\rm urgent}
 =\{\text{urgent review},\text{immediate escalation}\}
\]
is a measurable event. If a language distribution assigns masses $0.45$,
$0.15$ and $0.40$ to the three phrases respectively, then
\[
 Q(A_{\rm urgent})=0.45+0.15=0.60.
\]
This is ordinary probability on a finite set; the elements happen to be
phrases rather than numbers.

\noindent\textbf{Countably many strings.} With a finite or countable
vocabulary, the collection of all finite strings is countable. We may again
use the discrete $\sigma$-algebra containing every collection of strings.
Consequently, any declared group of phrases is measurable. The difficulty in
the paper is not whether these sets can be measured; it is whether grouping
phrases by meaning is scientifically defensible and stable.

\noindent\textbf{Worked semantic-map example.} Define
\[
 \phi:\mathcal R\to\mathcal X
\]
by
\[
 \phi(\text{urgent review})=\phi(\text{immediate escalation})
   =\text{urgent},
 \qquad
 \phi(\text{routine follow-up})=\text{routine}.
\]
For the categorical event $B=\{\text{urgent}\}\subseteq\mathcal X$,
\[
 \phi^{-1}(B)
 =\{\text{urgent review},\text{immediate escalation}\}
 =A_{\rm urgent}.
\]
Both spaces are discrete, so this inverse image is measurable. The induced
state probability is
\[
 (\phi_\#Q)(B)=Q\{\phi^{-1}(B)\}=0.60.
\]
This is the direct analogue of
$f^{-1}([1,4])=[-2,-1]\cup[1,2]$: the numerical function $f$ is replaced by
a semantic function $\phi$, and numerical intervals are replaced by declared
categories.

\noindent\textbf{Probability vectors return us to Euclidean space.} Once the
three category probabilities have been calculated, they form a vector such as
\[
 \boldsymbol p=(0.60,0.30,0.10)
 \in\Delta_3
 =\{(p_1,p_2,p_3)\in[0,1]^3:p_1+p_2+p_3=1\}.
\]
Thus words and phrases live in discrete measurable spaces, while their
probability vectors live in a subset of $\mathbb R^3$ equipped with its Borel
$\sigma$-algebra.

\subsection*{3.8 May the paper's spaces be assumed topological?}

They may be given topologies, but no single topology should be imposed on all
of them without justification. Convenient choices are:
\begin{center}
\small
\begin{tabularx}{\textwidth}{@{}l l X@{}}
\toprule
Paper space & Topology & Consequence \\
\midrule
$\mathcal X$ (finite states) & Discrete & Every collection of states is Borel. \\
$\mathcal R$ (finite strings) & Discrete & Every declared phrase group is Borel. \\
$\mathcal C$ (prompts) & Discrete & Every declared prompt family is measurable. \\
$\mathcal Y\subseteq\mathbb R^d$ & Euclidean subspace & Numerical evidence events are Borel. \\
$\Delta_K\subseteq\mathbb R^K$ & Euclidean subspace & Continuity, distances and error bounds are available. \\
\bottomrule
\end{tabularx}
\end{center}

For the discrete topology, every singleton is open, so every subset is open
and hence Borel. Thus words create no special difficulty for basic
measurability. For example, if
\[
 \mathcal R=\{r_1,r_2,r_3\},
\]
then $\{r_1,r_2\}$ is automatically Borel.

The discrete topology does \emph{not} encode semantic similarity. It does not
say that ``urgent review'' is closer in meaning to ``immediate escalation''
than to ``routine follow-up.'' An embedding-induced topology could encode a
model-dependent notion of proximity, but would change with the fitted model
and context. The paper therefore declares and validates a semantic map rather
than defining scientific meaning through an embedding topology.

Mixed numerical and textual evidence can be represented, for example, as
\[
 \mathcal Y=\mathbb R^d\times\mathcal W,
 \qquad
 \mathscr Y=\mathcal B(\mathbb R^d)\otimes\mathcal B(\mathcal W),
\]
where $\mathcal W$ is a finite or countable string space with the discrete
topology. A realization could be
\[
 y=((39.4,89),\text{``increasing respiratory distress''}).
\]
The product $\sigma$-algebra is beyond Chapters~2--3 of the chosen primer.
Accordingly, the paper can state explicitly that the combined evidence space
is standard Borel without developing product-measure theory in the setup.

\noindent\textbf{Where the paper is more complex.} A continuation can contain
several tokens, its probability depends on the complete prompt, and the
reported alternatives may omit probability mass. Information-equivalent
prompts can also produce different phrase probabilities. These are statistical
and semantic complications, not failures of Borel measurability.

\subsection*{3.9 Discrete language spaces and the geometry actually used}

Topology, geometry and semantic equivalence play different roles. Let
$\Sigma$ be a finite or countable vocabulary. The space of complete finite
continuations
\[
 \mathcal R=\bigcup_{m\geq0}\Sigma^m
\]
is countable and may be given the discrete topology. Then
$\mathcal B(\mathcal R)=\mathcal P(\mathcal R)$, so every declared phrase
group is measurable. This is the only topological structure required to
define the continuation law and its semantic pushforward.

The discrete topology supplies no useful graded notion of linguistic
similarity. Its associated discrete metric,
\[
 d_{\rm disc}(r,r')=\mathbf 1_{\{r\ne r'\}},
\]
assigns the same distance to ``urgent review'' versus ``immediate
escalation'' as to ``urgent review'' versus an unrelated phrase. It therefore
solves the measurability problem but not the semantic-stability problem.

For prompt perturbations, the paper instead uses a prespecified experimental
relation. Write $c\sim_E c'$ when two contexts contain the same evidence,
target and horizon and differ only through a declared nuisance change in
wording, order or numerical format. This is an information-equivalence
relation supplied by the experimental protocol; it is not implied by the
discrete topology. The effect of the change is then measured through the
induced probability laws, for example by
\[
 d_{\rm TV}\{Q_u(\cdot\mid c),Q_{u'}(\cdot\mid c')\}
\]
at the language layer and by Jensen--Shannon divergence between the calibrated
state-probability vectors. Thus the paper tests finite-design prompt stability
without claiming a universal metric on natural language.

An embedding map $e:\mathcal C\to\mathbb R^d$ could induce the pseudometric
\[
 d_e(c,c')=\|e(c)-e(c')\|_2.
\]
This may be useful as a supplementary diagnostic, but it is model-dependent
and is not used to define scientific equivalence: a small lexical change such
as inserting a negation may preserve embedding proximity while reversing the
scientific meaning.

After semantic grouping, the substantive geometry is no longer a geometry of
strings. The vector
\[
 \boldsymbol p_u(c)\in\Delta_K
\]
lies on the probability simplex, where total variation, Jensen--Shannon
divergence, Hellinger distance, Aitchison log-ratio geometry and, for
filtering, Hilbert's projective metric have established statistical meanings.
The paper therefore uses the following division of labour:
\[
 \boxed{\text{discrete topology makes language events measurable}},
\]
\[
 \boxed{\text{declared equivalence identifies nuisance rewordings}},
\]
\[
 \boxed{\text{probability geometry quantifies their inferential effect}}.
\]

\section*{4. Measurable functions}

A function
\[
   f:(S,\mathscr S)\longrightarrow(T,\mathscr T)
\]
is \emph{measurable} if
\[
   f^{-1}(B)\in\mathscr S
   \qquad\text{for every }B\in\mathscr T.
\]
Thus every measurable statement about the output corresponds to a measurable
event in the input space.

If $S$ and $T$ carry their Borel $\sigma$-algebras, such a function is called
\emph{Borel measurable}. Every continuous function between ordinary Euclidean
spaces is Borel measurable, although a Borel measurable function need not be
continuous.

\noindent\textbf{Paper notation.} The map
\[
   C_u:(\mathcal Y,\mathscr Y)\to(\mathcal C,\mathscr C)
\]
is measurable when
\[
   C_u^{-1}(B)\in\mathscr Y
   \qquad\text{for every }B\in\mathscr C.
\]
This condition says nothing about whether $C_u$ is one-to-one, onto or
information preserving. Those are separate properties.

\noindent\textbf{Worked example.} Take $S=T=\mathbb R$ and $f(s)=s^2$.
For $B=[1,4]$,
\[
 f^{-1}(B)=[-2,-1]\cup[1,2]\in\mathcal B(\mathbb R).
\]
The same inverse-image property holds for every Borel set, so $f$ is Borel
measurable. It is not one-to-one because $f(-1)=f(1)$ and is not onto
$\mathbb R$ because it never takes a negative value. Measurability therefore
does not imply invertibility.

\noindent\textbf{Paper connection and added complexity.} The presentation map
$C_u:\mathcal Y\to\mathcal C$ and semantic map
$\phi_u:\mathcal R\to\mathcal X_0$ are measurable. The paper also asks whether
presentation loses relevant information and whether semantic grouping permits
stable posterior recovery; measurability alone does not answer either question.

\section*{5. Measures and probability measures}

A \emph{measure} $\mu$ on $(S,\mathscr S)$ assigns a number
$\mu(A)\in[0,\infty]$ to every $A\in\mathscr S$ such that
\[
  \mu(\varnothing)=0,
  \qquad
  \mu\!\left(\bigcup_{j=1}^{\infty}A_j\right)
  =\sum_{j=1}^{\infty}\mu(A_j)
\]
for every sequence of pairwise disjoint measurable sets $A_j$.
A \emph{probability measure} additionally satisfies $\mu(S)=1$.

\noindent\textbf{Worked example.} Let $S=\{a,b\}$,
$\mathscr S=\{\varnothing,\{a\},\{b\},S\}$, and define
\[
 \mathbb P(\{a\})=0.3,\qquad \mathbb P(\{b\})=0.7.
\]
Then $\mathbb P(S)=0.3+0.7=1$ and $\mathbb P(\varnothing)=0$.
For the disjoint events $\{a\}$ and $\{b\}$, additivity gives
$\mathbb P(\{a,b\})=0.3+0.7=1$.

The reference model is written
\[
   (\Omega,\mathscr F,\mathbb P),
\]
where $\mathbb P$ is a probability measure on $(\Omega,\mathscr F)$.
The notation $\mathbb P(A)$ is meaningful only when $A\in\mathscr F$.

\subsection*{5.1 Probability need not be defined on a Borel $\sigma$-algebra}

A probability measure may be defined on any measurable space
$(\Omega,\mathscr F)$. The $\sigma$-algebra $\mathscr F$ might be Borel,
Lebesgue, discrete, or another application-specific $\sigma$-algebra. The
probability axioms require only that
\[
 \mathbb P:\mathscr F\longrightarrow[0,1]
\]
be countably additive and satisfy $\mathbb P(\Omega)=1$.

\noindent\textbf{Worked Lebesgue example.} Let
\[
 \Omega=[0,1],
 \qquad
 \mathscr F=\mathcal L([0,1]),
\]
where $\mathcal L([0,1])$ is the Lebesgue $\sigma$-algebra, and let
$\lambda$ denote Lebesgue measure. Define
\[
 \mathbb P(A)=\lambda(A),
 \qquad A\in\mathcal L([0,1]).
\]
Since $\lambda([0,1])=1$, this is a probability measure. For example,
\[
 \mathbb P([1/4,3/4])=\lambda([1/4,3/4])=\frac12.
\]
Thus
\[
 ([0,1],\mathcal L([0,1]),\lambda)
\]
is a probability space whose event $\sigma$-algebra is Lebesgue rather than
merely Borel.

The Lebesgue $\sigma$-algebra contains the Borel $\sigma$-algebra and also all
subsets of Borel sets having Lebesgue measure zero:
\[
 \mathcal B([0,1])\subsetneq\mathcal L([0,1]).
\]
This enlargement is called the completion of the Borel $\sigma$-algebra with
respect to Lebesgue measure. Some sets are therefore Lebesgue measurable but
not Borel measurable.

The phrase \emph{Lebesgue measurable set} means precisely
\[
 A\in\mathcal L(\mathbb R).
\]
It is better to say \emph{Lebesgue $\sigma$-algebra} than ``Lebesgue
algebra,'' because closure under countable operations is essential. A function
\[
 Z:(\Omega,\mathcal L(\Omega))\longrightarrow
   (\mathbb R,\mathcal B(\mathbb R))
\]
is Lebesgue measurable when
\[
 Z^{-1}(B)\in\mathcal L(\Omega)
 \qquad\text{for every }B\in\mathcal B(\mathbb R).
\]

\subsection*{5.2 Almost-sure statements do not require Lebesgue completion}

The notation ``$\mathbb P$-almost surely'', abbreviated $\mathbb P$-a.s., is
meaningful on any probability space $(\Omega,\mathscr F,\mathbb P)$.  A
statement holds $\mathbb P$-a.s. when its measurable failure event has
probability zero.  A Lebesgue $\sigma$-algebra is therefore not required.

\noindent\textbf{Concrete numerical example on the Borel space.} Consider
\[
 (\Omega,\mathscr F,\mathbb P)
 =([0,1],\mathcal B([0,1]),\lambda)
\]
and define two random variables by
\[
 X(\omega)=\omega,
 \qquad
 Y(\omega)=
 \begin{cases}
  99, & \omega=1/2,\\
  \omega, & \omega\ne1/2.
 \end{cases}
\]
Both are Borel measurable.  They disagree on exactly the Borel event
\[
 \{\omega:X(\omega)\ne Y(\omega)\}=\{1/2\},
\]
whose probability is
\[
 \mathbb P(\{1/2\})=\lambda(\{1/2\})=0.
\]
Consequently
\[
 X=Y\qquad\mathbb P\text{-a.s.},
\]
even though $X(1/2)=1/2$ and $Y(1/2)=99$.  This almost-sure statement is
already valid on the Borel probability space; no Lebesgue completion has been
used.

\noindent\textbf{What completion adds.} Let $N\subset[0,1]$ be the Cantor
set.  It is Borel and $\lambda(N)=0$.  There exist subsets $A\subseteq N$
that are not Borel.  On the Borel probability space, the indicator
\[
 Z(\omega)=\mathbf 1_A(\omega)
 =\begin{cases}1,&\omega\in A,\\0,&\omega\notin A\end{cases}
\]
is not measurable, because
\[
 Z^{-1}(\{1\})=A\notin\mathcal B([0,1]).
\]
After completing the space to
\[
 ([0,1],\mathcal L([0,1]),\lambda),
\]
every subset of the null set $N$, including $A$, is measurable and has
probability zero.  The same $Z$ is then measurable and
\[
 Z=0\qquad\lambda\text{-a.s.}
\]
Thus completion is not what creates the notion of almost-sure equality.  It
ensures that arbitrary changes made inside a measurable null event remain
measurable.  In compact form,
\[
 \boxed{\mathbb P\text{-a.s. requires a measurable null exceptional event}}
\]
whereas
\[
 \boxed{\text{completion makes every subset of that null event measurable}.}
\]

\subsection*{5.3 Why some authors write ``$\mathbb P$-measurable''}

The phrase ``$\mathbb P$-measurable'' is used in more than one way.  The
meaning must therefore be read from the author's definition.

First, since a probability measure is itself specified as
\[
 \mathbb P:\mathscr F\longrightarrow[0,1],
\]
some authors use ``$\mathbb P$-measurable'' merely as shorthand for
``$\mathscr F$-measurable.''  The measure already reveals which
$\sigma$-algebra is intended.

Second, the phrase may mean measurable with respect to the
$\mathbb P$-completion
\[
 \mathscr F^{\mathbb P}
 =\sigma\!\left(
   \mathscr F\cup
   \{A:A\subseteq N,\ N\in\mathscr F,\ \mathbb P(N)=0\}
  \right).
\]
This is a genuinely larger domain when $\mathscr F$ does not already contain
every subset of each $\mathbb P$-null event.  Third, in constructions from an
outer measure, ``measure-measurable'' may refer to the Carath\'eodory
criterion.  That is a separate construction and will not be needed in the
paper.

\noindent\textbf{Numerical example: when the distinction matters.} Return to
\[
 ([0,1],\mathcal B([0,1]),\lambda).
\]
For the non-Borel subset $A$ of the null Cantor set used above,
$Z=\mathbf 1_A$ is not $\mathcal B([0,1])$-measurable.  It is, however,
measurable with respect to the completed $\sigma$-algebra
$\mathcal B([0,1])^{\lambda}=\mathcal L([0,1])$.  Thus an author who calls
$Z$ ``$\lambda$-measurable'' may specifically mean that $Z$ is measurable
only after completion.  Here one must check the convention.

\noindent\textbf{Text example: when the distinction does not matter.} Let
the possible verbal continuations be
\[
 \mathcal R
 =\{\text{``urgent review''},\text{``routine review''},
      \text{``insufficient information''}\}
\]
with the discrete $\sigma$-algebra $\mathscr R=\mathcal P(\mathcal R)$, and
suppose
\[
 Q(\text{``urgent review''})=0.60,
 \quad
 Q(\text{``routine review''})=0.40,
 \quad
 Q(\text{``insufficient information''})=0.
\]
The null event is
\[
 N=\{\text{``insufficient information''}\}.
\]
Its only subsets are $\varnothing$ and $N$, and both already belong to
$\mathscr R$.  In fact, every subset of every null event belongs to
$\mathcal P(\mathcal R)$.  The space is already complete, so completion adds
nothing.  The same conclusion holds for the paper's finite semantic state
space and for a finite or countable continuation space equipped with its full
power-set $\sigma$-algebra.

The practical rule is therefore simple.  For finite categorical states and
declared continuation lists, ordinary measurability is sufficient and
completion requires no extra care.  Completion becomes relevant for
continuous or otherwise uncountable spaces when functions are modified on
arbitrary subsets of null events, or when results are represented only up to
$\mathbb P$-almost-sure equality.  In this paper, the phrase
``$\mathbb P$-measurable'' should be avoided unless the
$\mathbb P$-completion is explicitly intended.

\noindent\textbf{Paper connection.} The paper deliberately begins with a
general event $\sigma$-algebra $\mathscr F$. It therefore does not require the
underlying probability space to be Borel. Borel structure is imposed on the
finite state and evidence codomains where it is needed for measurable mappings
and conditional distributions.

\noindent\textbf{Running example with a measure.} Let $\lambda$ denote length
on $(\mathbb R,\mathcal B(\mathbb R))$. Then
\[
 \lambda([1,4])=4-1=3
\]
and, because the two intervals are disjoint except for no shared endpoint,
\[
 \lambda([-2,-1]\cup[1,2])
 =\lambda([-2,-1])+\lambda([1,2])=1+1=2.
\]
The Borel $\sigma$-algebra specifies which sets may be measured; $\lambda$
then assigns their numerical lengths.

\noindent\textbf{Paper connection and added complexity.} The paper uses
probability measures rather than length: $Q_u(\cdot\mid c)$ assigns mass to
verbal continuations and $\Pi^\star(\cdot\mid y)$ assigns mass to states. Since
each varies with a conditioning value, each is a probability kernel rather
than one fixed measure.

\section*{6. Random variables}

A \emph{random variable} is a measurable function from a probability space
into a measurable space. A real-valued random variable has the familiar form
\[
 Z:(\Omega,\mathscr F)\longrightarrow
   (\mathbb R,\mathcal B(\mathbb R)).
\]
For each outcome $\omega\in\Omega$, the realization $Z(\omega)$ is a real
number. The codomain need not be $\mathbb R$; categorical, vector-valued and
phrase-valued measurable maps are also called random variables here. Thus
\[
 X:(\Omega,\mathscr F)\longrightarrow(\mathcal X,\mathcal B(\mathcal X)),
 \qquad
 Y:(\Omega,\mathscr F)\longrightarrow(\mathcal Y,\mathscr Y)
\]
are random variables with different codomains.

The distribution of $Y$ is the probability measure
\[
   \mathbb P_Y(B)=\mathbb P\{Y\in B\}
   =\mathbb P\{Y^{-1}(B)\},
   \qquad B\in\mathscr Y.
\]
More compactly, $\mathbb P_Y=\mathbb P\circ Y^{-1}$ is the pushforward of
$\mathbb P$ through $Y$.

\noindent\textbf{Worked example.} Let
$\Omega=\{-2,-1,0,1,2\}$, assign probability $1/5$ to each point, and define
$Y(\omega)=\omega^2$. For $B=[1,4]$,
\[
 Y^{-1}(B)=\{-2,-1,1,2\},
 \qquad
 \mathbb P_Y(B)=\mathbb P\{Y\in B\}=\frac45.
\]
Here $Y$ is the random variable and $\mathbb P_Y$ is its induced
distribution.

\noindent\textbf{Continuous version of the running example.} Let $U$ be
uniformly distributed on $[-2,2]$ and set $Y=f(U)=U^2$. For
$B=[1,4]$,
\[
\begin{aligned}
 \mathbb P_Y(B)
 &=\mathbb P\{Y\in[1,4]\}\\
 &=\mathbb P\{U\in f^{-1}([1,4])\}\\
 &=\mathbb P\{U\in[-2,-1]\cup[1,2]\}\\
 &=\frac{\lambda([-2,-1]\cup[1,2])}{\lambda([-2,2])}
 =\frac{2}{4}=\frac12.
\end{aligned}
\]
This single calculation connects the Borel set $B$, its inverse image, the
underlying probability measure and the induced distribution of $Y$.

\noindent\textbf{Paper connection and added complexity.} The latent state
$X$ is categorical and the observed evidence $Y$ may combine numbers and
text. The target is the conditional distribution of $X$ given $Y$, not merely
the marginal distribution of either random variable.

\section*{7. Pushforward measures}

Let $f:(S,\mathscr S)\to(T,\mathscr T)$ be measurable and let $\mu$ be a
measure on $(S,\mathscr S)$. The \emph{pushforward measure} $f_{\#}\mu$ on
$(T,\mathscr T)$ is
\[
   (f_{\#}\mu)(B)=\mu\{f^{-1}(B)\},
   \qquad B\in\mathscr T.
\]
The pushforward transfers a probability measure from one measurable space to
another through a measurable function.

\noindent\textbf{Paper notation.} The semantic map
\[
   \phi_u:(\mathcal R,\mathscr R)\to
          (\mathcal X_0,\mathcal B(\mathcal X_0))
\]
pushes the conditional response measure $Q_u(\cdot\mid c)$ onto the finite
semantic state space. For a state $x_k$,
\[
 [\phi_{u\#}Q_u](\{x_k\}\mid c)
 =Q_u\{\phi_u^{-1}(x_k)\mid c\}.
\]
This operation simply adds the probabilities of all responses mapped to the
same state.

\noindent\textbf{Worked semantic pushforward.} Suppose
\[
 Q(\text{urgent review})=0.45,\quad
 Q(\text{immediate escalation})=0.15,\quad
 Q(\text{routine follow-up})=0.40.
\]
Let $\phi$ map the first two phrases to \emph{urgent} and the last to
\emph{routine}. Then
\[
 (\phi_\#Q)(\{\text{urgent}\})=0.45+0.15=0.60,
 \qquad
 (\phi_\#Q)(\{\text{routine}\})=0.40.
\]
Thus the induced state distribution is $(0.60,0.40)$.

\noindent\textbf{Paper connection and added complexity.} This pushforward is
the uncalibrated language-derived state distribution. The paper does not assume
that it equals the reference posterior: prompt dependence, semantic
misclassification and omitted continuation mass create additional errors.

\section*{8. Sigma-algebras generated by observations}

For a random variable $Y$, the notation $\sigma(Y)$ denotes the smallest
$\sigma$-algebra on $\Omega$ that makes $Y$ measurable:
\[
   \sigma(Y)=\{Y^{-1}(B):B\in\mathscr Y\}.
\]
It represents precisely the information revealed by observing $Y$.

If $Y$ records temperature and oxygen saturation, then events such as
``temperature exceeds $39^\circ$C'' or ``oxygen saturation is below $92\%$''
belong to $\sigma(Y)$. Events involving an unobserved quantity need not.

\noindent\textbf{Worked example.} Let $\Omega=\{1,2,3,4\}$ and define
\[
 Y(1)=Y(2)=0,
 \qquad
 Y(3)=Y(4)=1.
\]
The observation distinguishes $\{1,2\}$ from $\{3,4\}$ but not the points
within either pair. Therefore
\[
 \sigma(Y)=\{\varnothing,\{1,2\},\{3,4\},\Omega\}.
\]

\noindent\textbf{Paper connection and added complexity.} The reference
posterior conditions on the information in $Y$. Prompt wording is not meant to
redefine that target information. The empirical question is whether equivalent
presentations preserve enough information to recover the same posterior.

\section*{9. The one extension beyond Chapters 2--3}

Chapters~2 and~3 provide the foundations needed to define the observable
variables and their induced distributions. The paper additionally conditions
a state distribution on observed evidence. General conditioning with respect
to a $\sigma$-field and regular conditional distributions are later material;
they should not be attributed to Chapters~2 and~3.

\subsection*{9.1 Conditional expectation}

Let $Z$ be an integrable random variable and let
$\mathscr G\subseteq\mathscr F$ be a $\sigma$-algebra. The conditional
expectation $\mathbb E(Z\mid\mathscr G)$ is a $\mathscr G$-measurable random
variable satisfying
\[
 \int_G\mathbb E(Z\mid\mathscr G)\,d\mathbb P
 =\int_G Z\,d\mathbb P
 \qquad\text{for every }G\in\mathscr G.
\]
It is defined uniquely up to sets of $\mathbb P$-probability zero.

For an event $A\subseteq\mathcal X$,
\[
 \mathbb P\{X\in A\mid\sigma(Y)\}
 =\mathbb E\!\left[\mathbf 1_{\{X\in A\}}\mid\sigma(Y)\right].
\]
This is a random variable on $\Omega$. A regular conditional distribution
expresses the same quantity as a measurable function of the observed value
$Y$.

\noindent\textbf{Worked example.} On the preceding four-point space, assign
probability $1/4$ to each point and let
$Z=\mathbf 1_{\{1,3,4\}}$. Then
\[
 \mathbb E(Z\mid\sigma(Y))=
 \begin{cases}
  1/2,&\omega\in\{1,2\},\\
  1,&\omega\in\{3,4\}.
 \end{cases}
\]
Within $\{1,2\}$ the values of $Z$ are $1,0$; within $\{3,4\}$ they are
$1,1$. The conditional expectation is the average within each group that the
observation $Y$ cannot distinguish further.

\subsection*{9.2 Probability kernels}

Let $(S,\mathscr S)$ and $(T,\mathscr T)$ be measurable spaces. A
\emph{Markov kernel} or \emph{probability kernel}
\[
   K:S\times\mathscr T\to[0,1]
\]
satisfies two conditions:
\begin{enumerate}[label=(\roman*)]
\item for every $s\in S$, $B\mapsto K(B\mid s)$ is a probability measure on
      $(T,\mathscr T)$;
\item for every $B\in\mathscr T$, $s\mapsto K(B\mid s)$ is
      $\mathscr S$-measurable.
\end{enumerate}
A kernel is therefore a probability distribution that varies measurably with
a conditioning value.

\noindent\textbf{Simple example.} For $S=\mathbb R$ and
$T=\{0,1\}$, define
\[
 K(\{1\}\mid s)=\frac{e^s}{1+e^s},
 \qquad
 K(\{0\}\mid s)=\frac{1}{1+e^s}.
\]
For every $s$, these values define a Bernoulli distribution; as functions of
$s$, they are measurable.

\noindent\textbf{Paper notation.} Two kernels appear:
\[
 \Pi^\star:\mathcal Y\times\mathcal B(\mathcal X)\to[0,1],
 \qquad
 Q_u:\mathcal C\times\mathscr R\to[0,1].
\]
$\Pi^\star(\cdot\mid y)$ is the reference conditional distribution of the
state given evidence. $Q_u(\cdot\mid c)$ is the conditional distribution of a
response given covariate $c$.

\noindent\textbf{Numerical evaluation.} In the Bernoulli kernel above,
\[
 K(\{1\}\mid0)=K(\{0\}\mid0)=\tfrac12,
\]
while
\[
 K(\{1\}\mid\log3)=\tfrac34,
 \qquad K(\{0\}\mid\log3)=\tfrac14.
\]
For every fixed input, the two probabilities sum to one.

\noindent\textbf{Paper connection and added complexity.} The language kernel
$Q_u$ is observed through returned token or phrase probabilities. If the list
is truncated, unobserved mass must be recorded rather than silently assigned
probability zero.

\subsection*{9.3 Regular conditional distributions}

A \emph{regular conditional distribution} of $X$ given $Y$ is a kernel
\[
   \Pi^\star:\mathcal Y\times\mathcal B(\mathcal X)\to[0,1]
\]
such that, for every $A\subseteq\mathcal X$,
\[
   \Pi^\star(A\mid Y)
   =\mathbb P\{X\in A\mid\sigma(Y)\}
   \qquad\mathbb P\text{-almost surely}.
\]
The word \emph{regular} means that the conditional probability can be selected
as an actual measurable function of the observed value $y$. On standard Borel
spaces such a version exists. It is unique for $\mathbb P_Y$-almost every
$y$, rather than necessarily at every point that has zero probability under
$\mathbb P_Y$.

For the finite state space, define
\[
 \pi_k^\star(y)=\Pi^\star(\{x_k\}\mid y),
 \qquad
 \boldsymbol\pi^\star(y)
 =\bigl(\pi_1^\star(y),\ldots,\pi_K^\star(y)\bigr).
\]
Thus $\Pi^\star$ is the full conditional probability measure and
$\boldsymbol\pi^\star(y)$ is its vector of state probabilities.

\noindent\textbf{Worked example.} Let $X,Y\in\{0,1\}$ with joint
probabilities
\[
\begin{array}{c|cc}
 &Y=0&Y=1\\ \hline
X=0&0.40&0.20\\
X=1&0.10&0.30
\end{array}.
\]
Since $\mathbb P(Y=1)=0.20+0.30=0.50$,
\[
 \Pi^\star(\{1\}\mid1)
 =\mathbb P(X=1\mid Y=1)=\frac{0.30}{0.50}=0.60.
\]
Likewise, $\Pi^\star(\{1\}\mid0)=0.10/(0.40+0.10)=0.20$. With the
coordinate order $(X=0,X=1)$, the two posterior vectors are therefore
$(0.40,0.60)$ and $(0.80,0.20)$.

\noindent\textbf{Paper connection and added complexity.} The reference
posterior $\boldsymbol\pi^\star(y)$ comes from an external probabilistic model
or reference procedure. The language-derived vector is a separate measurement.
Identification asks when the latter contains sufficient information to recover
the former under stated support and equivalence conditions.

\section*{10. Conditioning and normalization}

Suppose a probability measure assigns masses
$\mu_1,\ldots,\mu_K$ to declared states and mass $1-M$ elsewhere, where
$M=\sum_{k=1}^K\mu_k>0$. Conditioning on membership in the declared states
gives
\[
   p_k=\frac{\mu_k}{M},
   \qquad k=1,\ldots,K.
\]
The vector $\boldsymbol p=(p_1,\ldots,p_K)$ lies in the probability simplex
\[
   \Delta_K=\left\{\boldsymbol z\in[0,1]^K:
                     \sum_{k=1}^Kz_k=1\right\}.
\]
This normalization does not recover probability mass that was omitted. The
quantity $1-M$ must be retained separately because small $M$ can make the
normalized composition unstable.

\noindent\textbf{Worked example.} Suppose $\mu_1=0.45$, $\mu_2=0.15$ and
the unexpressed mass is $0.40$. Then $M=0.60$ and
\[
 p_1=\frac{0.45}{0.60}=0.75,
 \qquad
 p_2=\frac{0.15}{0.60}=0.25.
\]
The normalized vector is $(0.75,0.25)$, but it must be accompanied by
$M=0.60$; otherwise the omitted mass $0.40$ is concealed.

\noindent\textbf{Paper connection and added complexity.} Normalization places
the measured composition on the simplex, where calibration is performed. The
covered mass $M$ remains a separate diagnostic because normalization can make
a poorly covered candidate set appear falsely decisive.

\section*{11. Almost-sure statements}

An equality holds \emph{$\mathbb P$-almost surely}, abbreviated
$\mathbb P$-a.s., if it may fail only on an event of probability zero. This is
the natural equality used for random variables and conditional expectations.
Conditional distributions are generally determined only up to such null sets.

This distinction matters when a continuous evidence value has
$\mathbb P(Y=y)=0$. The expression $\Pi^\star(\cdot\mid y)$ is still supplied
by a regular conditional kernel, but changing its value on a
$\mathbb P_Y$-null set does not change the statistical model.

\noindent\textbf{Worked example.} Let $Y$ be uniform on $[0,1]$ and define
$Z=Y$ except at $1/2$, where $Z(1/2)=99$. Because
$\mathbb P(Y=1/2)=0$,
\[
 \mathbb P(Z=Y)=1.
\]
Thus $Z=Y$ almost surely, even though the functions differ at one point.

\noindent\textbf{Paper connection and added complexity.} Existence and
uniqueness of the reference conditional distribution are understood up to
evidence values having probability zero. Empirical claims are therefore tied
to the evidence distribution represented in calibration and test data, not to
every imaginable text or prompt.

\section*{Extra background notes}

\subsection*{Vector-valued evidence and finite-valued states}

The main bridge permits evidence to contain several recorded quantities. In
the simplest Euclidean case,
\[
 Y:\Omega\longrightarrow\mathbb R^d,
 \qquad Y(\omega)=(39.4,89),
\]
where the two coordinates might record temperature and oxygen saturation.
Here $d$ is the number of evidence components, not the sample size. The paper
uses the more general evidence space $\mathcal Y$ because an evidence record
may also contain categories or text.

The latent state may likewise be categorical. For example,
\[
 \mathcal X
 =\{\text{urgent},\text{routine},\text{insufficient information}\},
 \qquad
 X:(\Omega,\mathscr F)longrightarrow
   (\mathcal X,\mathcal B(\mathcal X)).
\]
With the discrete topology,
$\mathcal B(\mathcal X)=\mathcal P(\mathcal X)$, so every collection of
states is measurable. Thus
\[
 X^{-1}\{\text{urgent},\text{insufficient information}\}\in\mathscr F.
\]
The variable $X$ is the latent scientific state whose conditional
distribution is sought; it is not a generated phrase. By contrast,
$\mathcal R$ is the continuation space and contains phrases such as
``urgent review''. Even when a state and a phrase use similar words, they are
different mathematical objects. They are connected only by the prespecified
semantic map
\[
 \mathcal R\xrightarrow{\ \phi\ }\mathcal X_0,
 \qquad \mathcal X_0=\mathcal X\cup\{x_0\}.
\]

For completeness, if
$\mathcal R=\{r_1,r_2,r_3\}$ is finite, then
$\mathscr R=\mathcal B(\mathcal R)=\mathcal P(\mathcal R)$. A measurable map
\[
 W:(\Omega,\mathscr F)\longrightarrow(\mathcal R,\mathscr R)
\]
represents a random complete verbal continuation. The paper works directly
with its conditional law $Q_u(\cdot\mid C_u(y))$, rather than treating $W$ as
the scientific state.

\subsection*{The measurable class of prompt-construction maps}

For a more formal treatment, the indexed collection
\[
 \mathfrak C=\{C_u:u\in\mathcal U\}
 \subseteq
 \mathcal M_{\mathscr Y,\mathscr C}(\mathcal Y,\mathcal C)
\]
is the fixed family of prompt-construction maps, where
\[
 \mathcal M_{\mathscr Y,\mathscr C}(\mathcal Y,\mathcal C)
 :=\left\{f:\mathcal Y\to\mathcal C:
       f^{-1}(D)\in\mathscr Y\text{ for every }D\in\mathscr C\right\}.
\]
This function-class notation is not required for the main statistical setup;
it is recorded here for later theoretical or book-length developments.

\subsection*{Numerical random variables and the half-line measurability test}

Begin with a measurable
map
\[
 Y:(\Omega,\mathscr F)\longrightarrow
   (\mathbb R,\mathcal B(\mathbb R)).
\]
For example, $Y(\omega)$ may be a measured temperature. Following the usual
textbook shorthand, we say that $Y$ is an $\mathscr F$-measurable random
variable, with the Borel $\sigma$-algebra on $\mathbb R$ understood. This means
\[
 Y^{-1}(B)=\{\omega\in\Omega:Y(\omega)\in B\}\in\mathscr F
 \qquad\text{for every }B\in\mathcal B(\mathbb R).
\]
Here $Y$ represents the observed numerical evidence and $Y(\omega)$ is the
value observed under outcome $\omega$. For example, if outcome $\omega$
describes one possible patient, then $Y(\omega)=39.4$ may be that patient's
temperature. The symbol $Y$ denotes the random variable before observation;
the lower-case symbol $y=39.4$ denotes one realised evidence value.
Capi\'nski and Kopp equivalently require
\[
 \{\omega\in\Omega:Y(\omega)\geq a\}\in\mathscr F
 \qquad\text{for every }a\in\mathbb R.
\]
Why do they use a one-sided condition rather than immediately requiring
\[
 \{\omega\in\Omega:a\leq Y(\omega)\leq b\}\in\mathscr F?
\]
Either formulation can be used. The closed half-lines $[a,\infty)$, as $a$
ranges over $\mathbb R$, generate $\mathcal B(\mathbb R)$, so it is enough to
check their inverse images. The bounded-interval event follows from the
$\sigma$-algebra operations. First,
\[
 \{Y\leq b\}
 =\bigcap_{n=1}^{\infty}\{Y<b+1/n\}
 =\bigcap_{n=1}^{\infty}\{Y\geq b+1/n\}^{c}
 \in\mathscr F.
\]
Therefore
\[
\begin{aligned}
 \{a\leq Y\leq b\}
 &=\{Y\geq a\}\cap\{Y\leq b\}\\
 &\in\mathscr F.
\end{aligned}
\]
Thus the one-sided test is not a restriction; it is a compact generating
criterion.

\noindent\textbf{Why do the half-lines generate all Borel sets?} A single
half-line is not sufficient. The generating family is
\[
 \mathcal H=\{[a,\infty):a\in\mathbb R\}.
\]
Starting from all members of $\mathcal H$, the $\sigma$-algebra operations
produce every open interval. First,
\[
 (-\infty,b)=[b,\infty)^c.
\]
Second,
\[
 (a,\infty)
 =\bigcup_{n=1}^{\infty}[a+1/n,\infty).
\]
To interpret this carefully, define
\[
 A_n=[a+1/n,\infty).
\]
The sets increase as $n$ increases:
\[
 A_1\subseteq A_2\subseteq A_3\subseteq\cdots.
\]
For an increasing sequence of sets, its set-theoretic limit is the union, so
one may write
\[
 \lim_{n\to\infty}A_n
 :=\bigcup_{n=1}^{\infty}A_n
 =(a,\infty),
\]
provided this meaning of the set limit is stated. It is not correct to replace
$1/n$ informally by $1/\infty$. The expression $1/\infty$ is not a real
number, and such a substitution would misleadingly suggest the closed
half-line $[a,\infty)$.

The endpoint explains the difference. For every finite $n$,
\[
 a<a+1/n,
\]
so $a\notin A_n$ and hence
\[
 a\notin\bigcup_{n=1}^{\infty}A_n.
\]
On the other hand, if $x>a$, then one can choose $n$ sufficiently large that
$1/n\leq x-a$. It follows that $x\geq a+1/n$, so $x\in A_n$. Therefore every
$x>a$ enters the union, while $a$ never does, proving
\[
 \bigcup_{n=1}^{\infty}[a+1/n,\infty)=(a,\infty).
\]
Therefore
\[
 (a,b)=(a,\infty)\cap(-\infty,b)
\]
belongs to $\sigma(\mathcal H)$. Every open subset of $\mathbb R$ can be
written as a countable union of open intervals with rational endpoints.
Consequently, every open set belongs to $\sigma(\mathcal H)$, and hence
\[
 \mathcal B(\mathbb R)\subseteq\sigma(\mathcal H).
\]
Conversely, each $[a,\infty)$ is closed and therefore Borel, so
\[
 \sigma(\mathcal H)\subseteq\mathcal B(\mathbb R).
\]
The two inclusions give
\[
 \boxed{\mathcal B(\mathbb R)=\sigma\{[a,\infty):a\in\mathbb R\}}.
\]

This also explains the measurability test. Inverse images commute with the
same operations:
\[
 Y^{-1}(B^c)=\{Y^{-1}(B)\}^c,
 \qquad
 Y^{-1}\!\left(\bigcup_nB_n\right)=\bigcup_nY^{-1}(B_n).
\]
If the inverse image of every generating half-line is in $\mathscr F$, then
the inverse image of every set constructed from those half-lines is also in
$\mathscr F$. Since those constructed sets are precisely the Borel sets, the
half-line condition is sufficient.

The sets $[a,\infty)$ are closed half-lines, rather than half-open bounded
intervals. One could instead use the family of all half-open intervals
$[a,b)$; that family also generates $\mathcal B(\mathbb R)$, but it is not the
particular criterion used in the cited chapter.

\noindent\textbf{Worked construction of $[a,b]$.} The lower half-line
$[a,\infty)$ imposes the lower bound $x\geq a$. To impose the upper bound
$x\leq b$, first construct
\[
 (b,\infty)
 =\bigcup_{n=1}^{\infty}[b+1/n,\infty),
\]
and then take its complement:
\[
 (-\infty,b]
 =(b,\infty)^c
 =\left(\bigcup_{n=1}^{\infty}[b+1/n,\infty)\right)^c.
\]
The required closed interval is therefore
\[
\boxed{
 [a,b]
 =[a,\infty)\cap(-\infty,b]
 =[a,\infty)\cap
 \left(\bigcup_{n=1}^{\infty}[b+1/n,\infty)\right)^c.}
\]
The endpoint in the second half-line must be $b$, not $a$:
$[a,\infty)$ supplies the lower endpoint and $(-\infty,b]$ supplies the upper
endpoint. If one instead used $(-\infty,a]$, the intersection would be
\[
 [a,\infty)\cap(-\infty,a]=\{a\},
\]
which constructs only the singleton at the lower endpoint.

The Borel $\sigma$-algebra is not merely a list of intervals and their unions.
It is
\[
 \mathcal B(\mathbb R)
 =\sigma\{(a,b):a<b\},
\]
the smallest $\sigma$-algebra containing all open intervals. It consequently
contains complements, countable unions and countable intersections formed from
them. For example,
\[
 [1,4]=( (-\infty,1)\cup(4,\infty))^c
 \in\mathcal B(\mathbb R),
\]
and
\[
 \mathbb Q=\bigcup_{q\in\mathbb Q}\{q\}
 \in\mathcal B(\mathbb R).
\]
It contains all open and closed subsets of $\mathbb R$ and many more
complicated sets obtained by repeated countable operations, but not every
subset of $\mathbb R$.
For a real-valued random variable
\[
 Y:(\Omega,\mathscr F)\longrightarrow
   (\mathbb R,\mathcal B(\mathbb R)),
\]
standard probability terminology calls $Y$ \emph{$\mathscr F$-measurable}
when
\[
 Y^{-1}(B)\in\mathscr F
 \qquad\text{for every }B\in\mathcal B(\mathbb R).
\]
Because the sets pulled back from the range are Borel sets, the same random
variable may also be called Borel measurable. These are two descriptions of
the same inverse-image condition, not two different requirements.

\section*{Summary}

The entire construction rests on four operations:
\begin{enumerate}[label=\textbf{\arabic*.}]
\item declare measurable spaces and probability measures;
\item represent observations by measurable random variables;
\item describe conditional distributions by probability kernels;
\item transfer a response distribution to semantic states by a measurable
      pushforward.
\end{enumerate}
These operations are standard measure theory. The new statistical question is
whether the resulting semantic composition identifies and accurately estimates
the independently defined reference posterior.

\section*{Sources and further reading}
\begingroup
\sloppy

\begin{enumerate}[label=\textbf{\arabic*.}]
\item M. Capi\'nski and E. Kopp, \emph{Measure, Integral and Probability},
Springer, 2nd ed., 2004. Chapter~2 supplies the treatment of Borel sets and
probability spaces; Chapter~3 supplies measurable functions, random variables,
generated $\sigma$-fields and induced probability distributions:
\url{https://zaco.au/lib/math/analysis/Marek-Capinski-Peter-E.-Kopp-Measure-integral-and-probability-Springer-2004.pdf}.

\item R. L. Wolpert, \emph{STA 711: Probability and Measure Theory, Week 1},
Duke University, p.~2. This supplementary reference discusses the power set;
the paper uses $\mathcal P(S)$ whenever that collection must be mentioned:
\url{https://www2.stat.duke.edu/courses/Fall15/sta711/lec/wk-01.pdf}.

\item O. Kallenberg, \emph{Foundations of Modern Probability}, Springer,
2nd ed., 2002, Chapter~1, for measurable generating classes and probability
kernels. Publisher record:
\url{https://link.springer.com/book/10.1007/978-1-4757-4015-8}.

\item P. L\'opez, \emph{Introduction to Probability and Statistics}, UCLA,
Definition~1.7, p.~6. These notes use the alternative notation
$\mathcal P(\Omega)=\{A:A\subseteq\Omega\}$:
\url{https://www.math.ucla.edu/~plopez/Fall_25/170E/Lecture_notes.pdf}.

\item C. Wallace, \emph{Probability 1: Axioms of Probability}, Durham
University, for the power set and the largest possible $\sigma$-algebra on a
sample space:
\url{https://www.maths.dur.ac.uk/users/clare.wallace/Prob1/01-axioms.html}.
\end{enumerate}
\endgroup

\newtoks\BeliefLensFormerSetup
\BeliefLensFormerSetup={%
\subsection{Basic model setup}
Our goal is to estimate a declared posterior over a finite state space from
an observable conditional response distribution. The setup is as follows.

\begin{figure}[H]
\centering
\resizebox{0.84\linewidth}{!}{%
\begin{tikzpicture}[
  font=\scriptsize,
  arr/.style={-{Latex[length=1.5mm]},thin,black}]
\node (evidence) at (0,1.25) {$y\in\mathcal Y$};
\node (target) at (8.8,1.25) {$\boldsymbol\pi^\star(y)\in\Delta_K$};
\node (law) at (2.1,0) {$Q_u\{\,\cdot\mid C_u(y)\}\in\mathcal M_1(\mathcal R)$};
\node (lexical) at (5.6,0) {$\boldsymbol p_u\{C_u(y)\}\in\Delta_K$};
\node (estimate) at (8.8,0) {$\widehat{\boldsymbol\pi}_u(y)\in\Delta_K$};
\draw[arr] (evidence)--node[above]{reference conditioning}(target);
\draw[arr] (evidence)--node[below left]{$C_u,Q_u$}(law);
\draw[arr] (law)--node[above]{$\phi_u$}(lexical);
\draw[arr] (lexical)--node[above]{calibration}(estimate);
\draw[arr,dashed] (estimate)--(target);
\end{tikzpicture}
}
\caption{The target posterior is defined by the reference probability law.
The lower path constructs an estimator from an observable conditional law on
responses. The dashed line denotes posterior-recovery error.}
\label{fig:formal-setup}
\end{figure}

Let $(\Omega,\mathscr F,\mathbb P)$ be a probability space, with generic
outcome (sample point) $\omega\in\Omega$. Following standard probability
terminology, $\Omega$ is the sample space and the members $A\in\mathscr F$ are
events; neither $\omega$ nor $\Omega$ is called an event. Let
$\mathcal X=\{x_1,\ldots,x_K\}$, equipped with its discrete Borel $\sigma$-algebra $\mathcal B(\mathcal X)$, and let
$(\mathcal Y,\mathscr Y)$ be a standard Borel space. Equivalently, one may
realise $\mathcal Y$ as a Borel subset of a complete separable metric space and
write $\mathscr Y=\mathcal B(\mathcal Y)$. Define the measurable random
variables
\[
 X:(\Omega,\mathscr F)\longrightarrow(\mathcal X,\mathcal B(\mathcal X)),
 \qquad
 Y:(\Omega,\mathscr F)\longrightarrow(\mathcal Y,\mathscr Y)
\]
as the finite-valued latent state and the observed evidence. Our target is a
regular conditional probability kernel
$\Pi^\star:\mathcal Y\times\mathcal B(\mathcal X)\to[0,1]$ satisfying
\[
 \Pi^\star(A\mid Y)
 =\mathbb P\{X\in A\mid\sigma(Y)\}\quad\mathbb P\text{-a.s.},
 \qquad A\subseteq\mathcal X.
\]
Following common statistical usage, we call $X$ and $Y$ random variables even
when their codomains are finite or structured. Capi\'nski and Kopp use the term
literally for maps into $\mathbb R$; a finite-valued $X$ can be placed in that
form by any one-to-one coding of its $K$ states into
$\{1,\ldots,K\}\subset\mathbb R$. For each $\omega\in\Omega$, $X(\omega)$
and $Y(\omega)$ are realizations of the state and evidence variables.
For $k=1,\ldots,K$, define the state probability
$\pi_k^\star(y)=\Pi^\star(\{x_k\}\mid y)$ and collect these probabilities in
\[
 \boldsymbol\pi^\star(y)
 =\bigl(\pi_1^\star(y),\ldots,\pi_K^\star(y)\bigr)\in\Delta_K,
 \qquad
 \Delta_K=\left\{\boldsymbol z\in[0,1]^K:\sum_{k=1}^Kz_k=1\right\}.
\]
Thus $\Pi^\star$ denotes the conditional probability measure on arbitrary
state events, whereas $\boldsymbol\pi^\star(y)$ denotes its vector of state
probabilities. Regularity means that $y\mapsto\Pi^\star(A\mid y)$ is measurable
for every $A\subseteq\mathcal X$; existence follows from the standard Borel
assumption, with uniqueness for $\mathbb P_Y$-almost every $y$.

To define the observation, let $(\mathcal C,\mathscr C)$ be a measurable
covariate space and $(\mathcal R,\mathscr R)$ a measurable continuation space. For each
$u\in\mathcal U$, let
\[
C_u:(\mathcal Y,\mathscr Y)\to(\mathcal C,\mathscr C)
\]
be measurable, and let
$Q_u:\mathcal C\times\mathscr R\to[0,1]$ be a Markov kernel. Hence
$Q_u(A\mid C_u(y))$ is the conditional probability of response event
$A\in\mathscr R$ given evidence $y$. Finally, a measurable map
$\phi_u:\mathcal R\to\mathcal X_0$, where $\mathcal X_0=\mathcal X\cup\{x_0\}$ and $x_0$ denotes an unexpressed continuation, induces a probability vector
$\boldsymbol p_u\{C_u(y)\}\in\Delta_K$ after conditioning on
$\phi_u\ne x_0$. The statistical problem is a semiparametric inverse problem:
determine whether this observable vector identifies
$\boldsymbol\pi^\star(y)$ and, if it does, estimate the inverse relation with
controlled error. The specific interpretation of $C_u$, $Q_u$, $\mathcal R$ and
$\phi_u$ for verbal data is deferred to Sections~2.3 and~2.4.

\noindent\textbf{Patient-diagnostics example.}
Let $\mathcal X=\{\text{urgent},\text{routine},\text{unsure}\}$ and suppose the
evidence is a temperature of $39.4^\circ$C with oxygen saturation $91\%$. A
declared clinical reference model might give
$\boldsymbol\pi^\star(y)=(.55,.35,.10)$. Under a fixed prompt, suppose the
continuation law assigns probabilities $.45$, $.15$, $.25$ and $.10$ to
``urgent review'', ``immediate escalation'', ``routine follow-up'' and
``insufficient information'', with $.05$ elsewhere. The semantic map combines
the first two phrases and retains the last $.05$ as unexpressed mass, giving
\[
\boldsymbol p_u(c)
=\frac{(.60,.25,.10)}{.95}
=(.632,.263,.105).
\]
The language-derived composition $\boldsymbol p_u(c)=(.632,.263,.105)$ and the
reference posterior $\boldsymbol\pi^\star(y)=(.55,.35,.10)$ live on the same
simplex but are not assumed equal. Calibration uses many such paired scenarios
to estimate the inverse relationship from $\boldsymbol p_u(c)$ to
$\boldsymbol\pi^\star(y)$; identification asks whether materially different
reference posteriors remain distinguishable through this language-derived
measurement.

The example completes the basic setup: the reference law supplies the target
and the observable response law supplies a second probability vector on the
same simplex. No sampling or independence assumption is needed to define these
quantities. A measurable set
$\mathcal Y_0\in\mathscr Y$ will denote the \emph{operating domain}: the
prespecified evidence values on which calibration and recovery claims are made.

The interior of the probability simplex is
\[
\operatorname{int}(\Delta_K)=\left\{\boldsymbol z\in\Delta_K:z_k>0\ \text{for every }k\right\}.
\]
Thus $\Delta_K$ contains all probability vectors over the $K$ states, including distributions with zero components. The interior restriction is imposed here because the log-ratio coordinates used below are finite only when every component is positive; it is not required merely to define a categorical probability distribution.
The target may come from any declared probabilistic model for which exact or independently certified inference is available. For example, a Bayesian network is one convenient data-generating process for conditionally structured evidence; a finite mixture or expert-validated likelihood model may serve the same role. No particular graphical model is required. What is required is a reproducible evidence-generating law and a reference posterior whose assumptions are visible to the evaluator.

\subsection{Calibration experiment and sampling assumptions}
The preceding construction concerns one generic pair $(X,Y)$ and makes no
independence assumption. For the baseline repeated-sampling analysis, scenario
pairs $(X_i,Y_i)$, $i=1,\ldots,n$, are independent copies of $(X,Y)$ within
each prespecified design stratum. Calibration uses these scenarios and, where
replication is available,
repeated language measurements $r=1,\ldots,R_i$. Scenario $i$ contributes
evidence $Y_i$, reference posterior $\boldsymbol\pi_i^\star$, and one or more
lexical compositions obtained under prespecified prompt rules. The calibration
sample is therefore a collection of paired probability compositions,
\[
\mathcal D_{\rm cal}
=\left\{\left(\boldsymbol\pi_i^\star,
  \boldsymbol p_{u,ir}\{C_u(Y_i)\}\right):
  i\in I_{\rm cal},\ r=1,\ldots,R_i\right\}.
\]
Repeated observations from one scenario remain in the same cluster.
Calibration, specification validation, conformal calibration and final testing
use disjoint scenario index sets fixed before the untouched test results are
examined.
In practical terms, five measurements of the same case remain one case; repetition measures stability but does not create five independent observations.

\noindent\textbf{Dependence assumptions.}
The basic construction, existence results and deterministic error bounds do
not require this i.i.d. sampling condition. Repeats within a scenario are not
treated as i.i.d.; they may be arbitrarily dependent and remain one cluster.
The reported cluster bootstrap uses the independent scenario clusters within
each declared stratum. Split-conformal coverage requires
exchangeability of calibration and future nonconformity scores within the
population or stratum for which coverage is claimed. If temporal, adaptive or
cross-scenario dependence is material, these inferential procedures must be
replaced by blocked, weighted or sequential analogues; the semantic
pushforward itself remains well defined.

The reference experiment and declared state definitions are fixed throughout; its
posterior is exact conditional on the declared model, not asserted to be
open-world truth. Continuation probabilities are observed rather than
estimated from sampled response frequencies, complete phrases are scored by
the chain rule, and unobserved probability mass remains explicit. These
conditions separate model definition from statistical estimation and prevent
replication from being mistaken for additional independent scenarios.

\subsection{Observable continuation law}
We next specify the observation experiment. Statistically, the prompt is a
designed textual covariate, the continuation is a discrete sequence-valued
response, and the fitted model supplies its conditional distribution. The
prompt must be represented explicitly because wording and ordering can alter
that distribution even when the intended information is unchanged.

\begin{definition}[Prompt specification for the experiment; cf.\ \citealp{jiang2021calibration,tian2023justask}]
A \emph{prompt specification} is a design rule for representing evidence as
text. Let $\mathcal U$ denote the admissible specifications and let
$(\mathcal C,\mathscr C)$ be the measurable space of finite contexts. For each
$u\in\mathcal U$, the measurable map
$C_u:(\mathcal Y,\mathscr Y)\to(\mathcal C,\mathscr C)$ fixes the state
definitions, evidence selection and ordering, numerical format, instructions
and response form. Thus $u$ indexes the design, $c=C_u(y)$ is the realised
context, and $C_u(Y)$ is the corresponding random covariate. The fitted model's
tokenisation and conditional response mechanism are part of the observation
process, not further design choices available to the evaluator.
Measurability means $C_u^{-1}(B)\in\mathscr Y$ for every
$B\in\mathscr C$. No one-to-one or onto property is presumed: two evidence
values may yield the same context, and many syntactically possible contexts
need never be produced. Any stronger claim that $C_u$ preserves the relevant
information in $Y$ is therefore an empirical restriction, not a consequence
of measurability.
\end{definition}
The prompt is part of the statistical design because two wordings that convey similar information need not induce the same conditional distribution over responses. For example, the rules ``choose urgent, routine, or insufficient information'' and ``choose immediate escalation, ordinary follow-up, or unresolved'' encode the same intended three-way distinction, but use different words and may therefore produce different response probabilities. Presentation equivalence is consequently an empirical hypothesis, not an assumption made by wording alone.

In the experimental design below, the notion of \emph{information equivalence} refers to presentation variants constructed from the same frozen evidence record with the same target, horizon, state meanings and response task. No evidential value is added, removed or changed; only wording, order or lossless formatting may differ. This is prespecified design terminology, not a formal equivalence relation or an assumption that the fitted language model responds identically. The presentation experiment tests that latter question: Theorem~\ref{thm:perturb} bounds the resulting change in state probabilities when the continuation laws differ.

The distinction between target and measurement is exact. The reference posterior $\boldsymbol\pi^\star(y)$ carries no presentation index $u$: for a fixed evidence realization $y$ and fixed reference probability model, changing only the wording, ordering or numerical format of an information-preserving prompt does not change the target. The language law $Q_u(\cdot\mid C_u(y))$, by contrast, generally does depend on $u$ and on which components of $y$ are included in the context. This is the sense in which the target is presentation-independent. It is not a claim that the posterior is independent of evidence selection, measurement, or the reference model. If a prompt omits or changes evidence, its information set is no longer equivalent; the unchanged full-evidence posterior may then be retained as a benchmark target specifically to measure the resulting curation loss. The semantic map does not reconstruct information that was never supplied. It instead provides a common state scale on which curation loss, presentation sensitivity, semantic mismatch and calibration error can be distinguished. The proposed calibrated map is expected to be stable only over the prespecified family of information-preserving presentations on which that stability is validated.

\begin{definition}[Continuation space and conditional language law; cf.\ \citealp{kallenberg2021foundations}]
We use the standard probability-kernel formalism \citep{kallenberg2021foundations} to represent conditional language generation. The kernel formalism is standard; the continuation space, its interpretation as observable language, and its later connection to scientific states through the semantic map are specific to the present construction.
Let $\mathcal V$ be the finite or countable vocabulary used to represent text. Its elements are model tokens, which may be words, parts of words or punctuation rather than complete dictionary words. The continuation space is the disjoint union
\[
\mathcal R=\mathcal V^{<\infty}=\bigcup_{m\ge0}\mathcal V^m,
\]
so an element $r=(v_1,\ldots,v_m)\in\mathcal R$ is one realised complete finite continuation: the text that could appear next after the prompt, such as ``urgent review'' or ``routine follow-up''. The continuation space $\mathcal R$ is an auxiliary observation space, distinct from the underlying probability sample space $\Omega$. We give it the discrete topology and write $\mathscr R=\mathcal B(\mathcal R)$ for its Borel $\sigma$-algebra; because the topology is discrete, every collection of continuations is measurable. Under prompt specification $u$, the language law is the Markov kernel $Q_u:\mathcal C\times\mathscr R\to[0,1]$. Hence $A\mapsto Q_u(A\mid c)$ is a probability measure on continuations for every realised context $c\in\mathcal C$, while $c\mapsto Q_u(A\mid c)$ is measurable for every event $A\in\mathscr R$.
\end{definition}
Let $W$ be the $\mathcal R$-valued random variable representing the complete continuation, with realization $r$. Conditional on prompt $c$ under specification $u$, this definition may be written compactly as $W\mid c\sim Q_u(\cdot\mid c)$. Thus $Q_u$ is simply the conditional continuation distribution; the auxiliary observation is a finite string rather than a scalar or Euclidean vector.
If $W_j$ denotes the random token generated at position $j$, the probability of a particular complete continuation is
\begin{equation}
Q_u(\{r\}\mid c)=\prod_{j=1}^m
\mathbb P(W_j=v_j\mid c,W_{1:j-1}=v_{1:j-1}).\label{eq:sequence-law}
\end{equation}
This is the familiar next-token mechanism written as a probability for a whole phrase. Complete-sequence probabilities matter because two state descriptions may share an opening token and diverge only later; using only the first token would then measure the shared prefix rather than the intended meaning.
For example, suppose the prompt ends with ``The patient requires'' and, for illustration, ``urgent'' and ``review'' are separate tokens. If
\[
\mathbb P(W_1=\text{``urgent''}\mid c)=0.60,
\qquad
\mathbb P(W_2=\text{``review''}\mid c,W_1=\text{``urgent''})=0.75,
\]
then the complete continuation ``urgent review'' has probability $0.60\times0.75=0.45$. The first number alone is not the probability of the phrase: it also includes continuations such as ``urgent testing'' or ``urgent monitoring.''

\subsection{Semantic pushforward}
The continuation law is observable, but its outcomes are strings rather than
the states in the reference experiment. The next step is therefore a
measurable coarsening. It preserves the probability calculus whilst changing
the outcome space: meaning-equivalent continuations are assigned to the same
declared state, and language outside the declared partition remains visible
rather than being silently discarded.

\begin{definition}[Semantic coarsening; adapted from \citealp{farquhar2024semantic,kuhn2023semantic}]
Following the semantic-equivalence literature, different verbal continuations may be grouped when they express the same meaning. Here, however, the groups are not inferred anew from sampled responses: they are a prespecified, application-defined partition whose induced probabilities will be calibrated against an external reference posterior.
The semantic outcome space is the finite measurable space $(\mathcal X_0,\mathcal B(\mathcal X_0))$ with the discrete topology, where $\mathcal X_0=\mathcal X\cup\{x_0\}$ and $x_0$ denotes an unexpressed or out-of-partition continuation. The symbol $x_0$ is an additional category, not an orthogonal complement or projection. For each $u\in\mathcal U$, the semantic map
\[
\phi_u:(\mathcal R,\mathscr R)\longrightarrow(\mathcal X_0,\mathcal B(\mathcal X_0))
\]
assigns each complete continuation to a semantic state or to $x_0$. Measurability here is with respect to $\mathscr R$ and $\mathcal B(\mathcal X_0)$, not Lebesgue measure. Because both spaces carry discrete $\sigma$-algebras, every such function is measurable. If $\mathcal M_1(S)$ denotes the probability measures on a measurable space $S$, the induced pushforward operator is $\phi_{u\#}:\mathcal M_1(\mathcal R)\to\mathcal M_1(\mathcal X_0)$, where $[\phi_{u\#}Q](A)=Q\{\phi_u^{-1}(A)\}$.
\end{definition}

The lexical state-mass functions $\mu_{u,k}:\mathcal C\to[0,1]$ and expressed-mass function $M_u:\mathcal C\to[0,1]$ are
\[
\mu_{u,k}(c)=Q_u(\phi_u^{-1}(x_k)\mid c),\qquad M_u(c)=\sum_{k=1}^K\mu_{u,k}(c),
\]
and, whenever $M_u(c)>0$, define
\[
p_{u,k}(c)=\frac{\mu_{u,k}(c)}{M_u(c)}
=Q_u\!\left(\phi_u=x_k\mid \phi_u\ne x_0,c\right),
\qquad
\boldsymbol p_u(c)=\bigl(p_{u,1}(c),\ldots,p_{u,K}(c)\bigr)\in\Delta_K.
\]
Thus $p_{u,k}:\mathcal C_u^+\to[0,1]$ is the scalar probability assigned to state $x_k$, whereas $\boldsymbol p_u:\mathcal C_u^+\to\Delta_K$, with $\mathcal C_u^+=\{c\in\mathcal C:M_u(c)>0\}$, is the entire probability vector over the $K$ expressed states. The composite map $\boldsymbol p_u\circ C_u:\mathcal Y_{0,u}^+\to\Delta_K$, where $\mathcal Y_{0,u}^+=\{y\in\mathcal Y_0:C_u(y)\in\mathcal C_u^+\}$, is the lexical measurement as a function of realised evidence. It is the pushforward distribution induced by $Q_u$ and $\phi_u$, conditional on the continuation falling inside the declared semantic partition.

For the numerical example in Section~2.1, these definitions give semantic
masses $(.60,.25,.10)$, expressed mass $M_u=.95$ and lexical composition
$(.632,.263,.105)$. Thus the continuation law induces a well-defined
multinomial observation on the declared state space, whilst the reference
posterior remains the separate calibration target.

Figure~\ref{fig:push} is not a second statistical setup. It enlarges the lower
middle portion of Figure~\ref{fig:formal-setup}, from evidence $y$ to the
lexical composition $\boldsymbol p_u\{C_u(y)\}$. In particular,
Figure~\ref{fig:formal-setup} shows
the complete inferential problem, including the reference posterior and the
calibrated estimate; Figure~\ref{fig:push} shows the internal construction of the lexical
measurement before calibration. The intermediate semantic-mass vector and its
normalisation are suppressed in Figure~\ref{fig:formal-setup} only to keep that overview legible.

\begin{figure}[H]\centering
\resizebox{0.98\linewidth}{!}{%
\begin{tikzpicture}[node distance=8mm and 7mm,font=\small]
\node[bbnode] (e) {evidence\\$Y$};
\node[bbnode,right=of e] (c) {context rule\\$C_u(Y)$};
\node[bbnode,right=of c] (q) {language law\\$Q_u(\cdot\mid c)$};
\node[bbnode,right=of q] (m) {semantic mass\\$\phi_u\push Q_u$};
\node[bbnode,right=of m] (p) {lexical composition\\$\boldsymbol p_u\in\Delta_K$};
\draw[bbmap] (e)--(c);\draw[bbmap] (c)--(q);\draw[bbmap] (q)--(m);\draw[bbmap] (m)--(p);
\node[below=2mm of e,align=center,text width=27mm,font=\scriptsize] (xe) {e.g. high fever,\\oxygen $91\%$};
\node[below=2mm of c,align=center,text width=27mm,font=\scriptsize] (xc) {``Evidence: fever, low oxygen.\\Choose urgent, routine or unsure.''};
\node[below=2mm of q,align=center,text width=27mm,font=\scriptsize] (xq) {$P$(urgent review)$=.45$\\$P$(escalation)$=.15$};
\node[below=2mm of m,align=center,text width=27mm,font=\scriptsize] (xm) {urgent $=.60$; routine $=.25$\\unsure $=.10$; $x_0=.05$};
\node[below=2mm of p,align=center,text width=27mm,font=\scriptsize] (xp) {$(.60,.25,.10)/.95$\\$=(.632,.263,.105)$};
\node[below=23mm of m,align=center,text=black,font=\scriptsize] (a) {assess: semantic error, probability error,\\unexpressed mass, repeated variation};\draw[bbdash] (a)--(xm);
\end{tikzpicture}
}
\caption{Expansion of the lexical-measurement path in Figure~\ref{fig:formal-setup}. The annotations follow one realization from evidence to $\boldsymbol p_u\{C_u(y)\}$: clinical evidence becomes a fixed prompt, the language law assigns probabilities to complete phrases, synonymous phrases are combined into semantic masses, and the expressed masses are normalised. The $0.05$ outside the declared categories remains visible as the additional category $x_0$. Calibration from $\boldsymbol p_u\{C_u(y)\}$ to $\widehat{\boldsymbol\pi}_u(y)$ is shown in Figure~\ref{fig:formal-setup} and developed in the next subsection.}
\label{fig:push}
\end{figure}

\begin{definition}[Lexical measurement experiment; introduced here]
The triple $\mathcal E_u=(Q_u,C_u,\phi_u)$ is a lexical measurement experiment for $\boldsymbol\pi^\star$ on domain $\mathcal Y_0$ if probabilities under $Q_u$ are observable, $C_u$ and $\phi_u$ are fixed before final evaluation, and $M_u\{C_u(y)\}>0$ for every $y\in\mathcal Y_0$.
\end{definition}
}

\subsection{Identification and inverse recovery}
The semantic pushforward supplies a probability vector on the correct state
space, but not yet a calibrated posterior. The remaining inferential problem
is an inverse one: learn the relation between lexical and reference
compositions on a calibration sample, and establish that distinct reference
posteriors remain distinguishable through the observation channel.

\begin{definition}[Quantitative observability; cf.\ \citealp{aitchison1982statistical,blackwell1953equivalent,engl1996regularization}]
Following the cited compositional-data and inverse-problem literature, probability vectors are treated as compositions and compared through log ratios rather than through unconstrained Euclidean coordinates. With state $K$ as the prespecified reference component, the additive log-ratio map and its inverse are
\[
\boldsymbol\ell:\operatorname{int}(\Delta_K)\to\mathbb R^{K-1},
\quad [\boldsymbol\ell(\boldsymbol z)]_k=\log(z_k/z_K),
\qquad
\boldsymbol\ell^{-1}:\mathbb R^{K-1}\to\operatorname{int}(\Delta_K).
\]
Let $\Pi_0\subset\operatorname{int}(\Delta_K)$ be the declared posterior domain and let $\mathcal L=\boldsymbol\ell(\Pi_0)\subset\mathbb R^{K-1}$ be its log-ratio image. Write $\boldsymbol\ell^\star=\boldsymbol\ell(\boldsymbol\pi^\star)$ and $\boldsymbol\lambda_u=\boldsymbol\ell(\boldsymbol p_u)$. The forward regression function $h_u:\mathcal L\to\mathbb R^{K-1}$ gives the conditional mean lexical log ratio associated with reference coordinate $\boldsymbol\ell$, namely $h_u(\boldsymbol\ell)=\mathbb E(\boldsymbol\lambda_u\mid\boldsymbol\ell^\star=\boldsymbol\ell)$ under the declared experiment. Unless a narrower class is stated, $h_u$ belongs to $C(\mathcal L;\mathbb R^{K-1})$, the space of continuous vector-valued functions with the uniform norm. It is quantitatively observable on $\mathcal L$ with modulus $c_u>0$ when
\[
\|h_u(\boldsymbol\ell)-h_u(\boldsymbol\ell')\|_2\ge c_u\|\boldsymbol\ell-\boldsymbol\ell'\|_2
\quad\text{for every }\boldsymbol\ell,\boldsymbol\ell'\in\mathcal L.
\]
In plain terms, materially different reference posteriors must not become nearly indistinguishable after language measurement; $c_u$ records the weakest retained direction.
\end{definition}

At this stage, we are ready to state the paper's main inferential problem formally. The semantic construction produces a language-derived probability vector, while the reference experiment supplies the target posterior. We must learn whether the former can be mapped back to the latter, and whether that inverse remains stable on new evidence rather than merely fitting the calibration sample.

\begin{definition}[Calibrated posterior estimator and affine specification]
An inverse calibration map is a measurable map
$\psi_u:\mathbb R^{K-1}\to\mathcal L$. Given $\psi_u$, the calibrated
posterior estimator is the measurable composition
\[
 \widehat{\boldsymbol\pi}_u
 =\boldsymbol\ell^{-1}\circ\psi_u\circ\boldsymbol\ell
  \circ\boldsymbol p_u\circ C_u:
 \mathcal Y_{0,u}^+\longrightarrow\operatorname{int}(\Delta_K).
\]
The empirical analysis uses the affine specification
$h_u(\boldsymbol\ell)=B_u\boldsymbol\ell+\boldsymbol d_u$, with an analogous affine form for $\psi_u$.
\end{definition}

A uniform nonlinear observability claim would additionally require a prespecified regularity class, such as $C^{1,1}(\mathcal L;\mathbb R^{K-1})$ with bounded derivative and Lipschitz derivative; finite observations cannot certify a uniform lower modulus over the unrestricted continuous class.

No independence is assumed among tokens within a continuation: its probability is the product of next-token probabilities, each conditional on the prompt and preceding tokens. The principal structural results are conditional or deterministic; independence or exchangeability is needed only for the repeated-sampling interpretation of bootstrap intervals and conformal coverage.

The fitted language-model service reveals only the $k$ largest conditional probabilities at each token position, so omitted continuations are not assigned probability zero and the unrestricted law $Q_u$ is not fully observed. An absent continuation is unknown, not impossible. For the controlled full-law assessment, we instead use a finite response grammar $G\subset\mathcal R$: three admissible continuations and an explicit residual category define $Q_u^G$, which is treated as effectively complete only when its covered mass exceeds a prespecified threshold. Thus the theory concerns $Q_u$, whereas this empirical assessment concerns the task-constrained law $Q_u^G$.

\subsubsection{Decomposition of uncertainty and measurement error}
This subsection separates two sources of discrepancy between a calibrated
estimate and its reference posterior: variation across repeated measurements
of the same scenario and recovery error that persists after those repetitions
are averaged. The distinction matters because repetition may reduce the first,
whereas the second requires improved representation or calibration. The decomposition is conditional on the fixed semantic map,
prompt specification, reference model and calibration procedure. It does not
treat posterior entropy, alternative state definitions or external validity
as measurement-error components. To formulate the decomposition in a geometry
appropriate for probability vectors, choose an orthonormal contrast matrix
\[
 V\in\mathbb R^{K\times(K-1)},
 \qquad V^\top V=I_{K-1},
 \qquad V^\top\boldsymbol 1_K=0,
\]
and define the basis-explicit log-ratio coordinate map
\[
 \boldsymbol\ell_V:\operatorname{int}(\Delta_K)\longrightarrow\mathbb R^{K-1},
 \qquad
 \boldsymbol\ell_V(\boldsymbol p)=V^\top\log\boldsymbol p,
\]
where the logarithm is componentwise. This map is an isometry from the simplex under Aitchison distance to Euclidean space. Its inverse normalizes the positive vector $\exp(V\boldsymbol z)$:
\[
 \boldsymbol\ell_V^{-1}(\boldsymbol z)
 =\frac{\exp(V\boldsymbol z)}{\boldsymbol 1_K^\top\exp(V\boldsymbol z)}.
\]
Let $\boldsymbol z_{ir}=\boldsymbol\ell_V\{\widehat{\boldsymbol\pi}_{ir}\}$ be the realised coordinate of the calibrated estimate for scenario $i$ and repeat $r$, let $\boldsymbol\theta_i=\boldsymbol\ell_V(\boldsymbol\pi_i^\star)$, and write $\overline{\boldsymbol z}_i=R^{-1}\sum_r\boldsymbol z_{ir}$. The usual within--between identity for the realised observations gives
\begin{equation}
 \frac1R\sum_{r=1}^R\|\boldsymbol z_{ir}-\boldsymbol\theta_i\|^2
 =\frac1R\sum_{r=1}^R\|\boldsymbol z_{ir}-\overline{\boldsymbol z}_i\|^2
 +\|\overline{\boldsymbol z}_i-\boldsymbol\theta_i\|^2.\label{eq:decomp}
\end{equation}
Under squared log-ratio, or Aitchison, distance, equation~\eqref{eq:decomp} separates error exactly into within-scenario variation and stable recovery error. The first asks whether repeated measurements move; the second asks whether their average remains wrong. Resampling calibration scenarios diagnoses inverse-estimation uncertainty but is nested within recovery error and is not added again. Missing mass and discrepancies between probability calculations assess the observation procedure, while posterior entropy measures target dispersion rather than measurement error. The decomposition is therefore conditional on the fixed semantic map and reference experiment.

The claims are likewise conditional on the fitted language model and observed response law: a finite grammar does not validate unrestricted prose, average accuracy does not establish uniform invertibility, and service drift or dependent adaptive sampling can invalidate ordinary bootstrap and split-conformal guarantees. The estimand is an observable application-specific composition, not an inaccessible internal belief. New states, evidence classes or presentations require validation; low expressed mass calls for expanded coverage or rejection, weak conditioning for a more informative presentation, and failed coverage for investigation of distribution shift. The next section formalises these conditions.

\section{Theory}
The mathematical development follows the order in which a statistician must justify the measurement channel. Following the comparison-of-experiments tradition of \citet{blackwell1953equivalent}, \citet{lecam1964sufficiency}, and \citet{torgersen1991comparison}, we ask whether a transformed observation preserves information about the target experiment. Our contribution is to pose that question for a probability law over language and to connect it to calibrated recovery in log-ratio coordinates. First, does the language distribution induce a well-defined probability over the declared response states? Second, how much error is introduced by approximating semantic classes with a finite set of scored sequences? Third, do different reference posteriors remain distinguishable after passing through the language channel? These are the static requirements that must be established before any dynamic use is considered.

Some conclusions below are theorems because they follow from explicit mathematical assumptions. Whether their conditions hold for a particular fitted language distribution is a separate statistical question and must be assessed from data. We therefore state each mathematical implication and then describe its empirical assessment on a finite prespecified design.

\subsection{Existence, uniqueness, and semantic approximation}
The first question is foundational but easy to overlook. Continuation probabilities concern phrases, while users reason about diagnoses, risks, hypotheses, or operational states. A semantic map connects these sample spaces. Measurability is sufficient for the pushforward distribution to exist, but uniqueness is conditional on the map: two defensible partitions can induce two different lexical compositions from the same language distribution.

Lemma~\ref{lem:exist} recasts the standard image-measure result of \citet{kallenberg2021foundations} for the semantic measurement defined here. The general theorem gives the pushforward measure, but it does not supply the paper's additional unexpressed category or the conditional distribution over declared states. That second step requires a positive expressed mass and can fail even when the pushforward itself exists. Kallenberg supplies the measure-theoretic result; the semantic map, augmentation and normalised lexical composition belong to the present construction.

\begin{lemma}[Existence and conditional uniqueness under a fixed semantic map; cf.\ \citealp{kallenberg2021foundations}]\label{lem:exist}
If $Q_u(\cdot\mid c)$ is a probability measure on $(\mathcal R,\mathscr R)$ and $\phi_u:(\mathcal R,\mathscr R)\to(\mathcal X_0,\mathcal B(\mathcal X_0))$ is measurable, then $\phi_u\push Q_u$ is a unique probability measure on $(\mathcal X_0,\mathcal B(\mathcal X_0))$. If $M_u(c)>0$, the normalised lexical composition $\boldsymbol p_u(c)$ exists and is unique conditional on $(Q_u,\phi_u,c)$.
\end{lemma}

Uniqueness is conditional on the fixed representation $(Q_u,\phi_u,c)$: the continuation law, semantic map and realised context. Changing any component may change the induced probability distribution. The lemma does not establish the scientific validity or identification of these components from data.

Existence is also insufficient for estimation because only finitely many, possibly truncated, phrases can generally be evaluated. What is needed is a quantitative account that separates semantic misassignment, probability-calculation error and insufficient expressed mass, and shows how they affect the normalised state probabilities. Such an account should identify whether the appropriate response is redesign, improved measurement or rejection rather than conceal distinct failures within one error score. The next theorem provides this decomposition.

\begin{definition}[Ideal and scored semantic cells]\label{def:semantic-cells}
For a declared state $x_k$, its \emph{ideal semantic cell} is
\[
 A_k=\phi_u^{-1}(\{x_k\})\in\mathscr R.
\]
Thus $A_k$ contains every continuation assigned to $x_k$ by the fixed semantic map. A \emph{scored semantic cell} $B_k\in\mathscr R$ is the finite set of candidate continuations actually evaluated for state $x_k$. We require $B_1,\ldots,B_K$ to be pairwise disjoint. Ideally $B_k=A_k$; in practice, $B_k$ may omit valid continuations or contain continuations assigned to a different ideal cell.
\end{definition}

\begin{definition}[Semantic and probability-calculation errors]\label{def:approx-errors}
At a fixed context $c$, let
\[
 q_k=Q_u(A_k\mid c),\qquad
 \widetilde q_k=\widetilde Q_u(B_k\mid c),
\]
where $Q_u(\cdot\mid c)$ is the full conditional continuation law and $\widetilde Q_u(\cdot\mid c)$ is the measured probability law used to score the finite candidate sets. Write
\[
 \boldsymbol q=(q_1,\ldots,q_K),\quad
 \widetilde{\boldsymbol q}=(\widetilde q_1,\ldots,\widetilde q_K),\quad
 M=\sum_{k=1}^Kq_k,\quad
 \widetilde M=\sum_{k=1}^K\widetilde q_k.
\]
The \emph{semantic approximation error} and \emph{probability-calculation error} are, respectively,
\[
 e_{\rm sem}=\sum_{k=1}^K Q_u(A_k\triangle B_k\mid c),\qquad
 e_{\rm score}=\sum_{k=1}^K
 \left|Q_u(B_k\mid c)-\widetilde Q_u(B_k\mid c)\right|.
\]
The first measures incorrect or incomplete approximation of the semantic cells; the second measures error in the probabilities assigned to those scored cells.
\end{definition}

\begin{theorem}[Finite lexical approximation and semantic-error decomposition]\label{thm:approx}
Suppose $M=\sum_kq_k\ge m_0>0$ and $e=e_{\rm sem}+e_{\rm score}<m_0$. With $\boldsymbol p=\boldsymbol q/M$ and $\widetilde{\boldsymbol p}=\widetilde{\boldsymbol q}/\widetilde M$ in $\Delta_K$,
\[
\|\widetilde{\boldsymbol p}-\boldsymbol p\|_1\le \frac{2e}{m_0}.
\]
\end{theorem}
Theorem~\ref{thm:approx} converts empirical bounds on semantic and probability-calculation error into a deterministic bound on the final state probabilities, with $m_0^{-1}$ quantifying amplification by limited expressed mass. This decomposition is a primary theoretical contribution. For example, $m_0=0.80$, $e_{\rm sem}=0.02$ and $e_{\rm score}=0.01$ give an $L^1$ bound of $0.075$; if $m_0=0.20$, the same errors give $0.30$.

\noindent\textbf{Discussion.} The theorem turns semantic coverage, probability accuracy and missing mass into a common error budget. It detects apparently confident estimates created by discarding unsupported alternatives and renormalising the remainder---a failure that modal accuracy can conceal. If the bound is unacceptable, the state vector should be withheld from downstream decisions; its components also identify the remedy: revise the semantic cells, improve probability calculation or expand verbal coverage. It thereby supplies both a quantitative runtime-governance rule and a diagnosis of why intervention is required.

\subsection{Perturbation and lexical-presentation stability}
A presentation change can alter the continuation law before calibration. This is problematic because unchanged evidence can then produce different measured state probabilities: the estimator responds to wording rather than information, weakening reproducibility and potentially invalidating a calibration fitted under another presentation. The relevant statistical quantity is the distance between the resulting laws, not an unobserved contextual representation or the amount of textual editing: a long paraphrase may preserve meaning, whereas one negation may reverse it. For example, for probabilities over urgent, routine and unsure,
\[
 \boldsymbol q=(.80,.15,.05),\quad
 \boldsymbol q_{\rm para}=(.79,.16,.05),\quad
 \boldsymbol q_{\rm neg}=(.15,.80,.05),
\]
the respective total-variation distances from $\boldsymbol q$ are $.01$ and $.65$.

Theorem~\ref{thm:approx} controls approximation of a single semantic composition. Presentation stability additionally requires propagating the observable distance between two continuation laws through semantic grouping, normalisation and calibration. For probability measures $\mu$ and $\nu$ on a common measurable space, write
\[
 d_{\rm TV}(\mu,\nu)=\sup_A|\mu(A)-\nu(A)|
\]
for their total-variation distance, where the supremum is over measurable events \citep{kallenberg2021foundations}. To formalise presentation stability, fix $\phi$, let $Q=Q_u(\cdot\mid c)$ and $Q'=Q_{u'}(\cdot\mid c')$, and denote their normalised compositions by $\boldsymbol p$ and $\boldsymbol p'$. The next theorem applies contraction of total variation under measurable maps and stability of conditioning; the corollary carries the resulting bound through calibration.

\begin{theorem}[Stability of semantic coarsening under context perturbation; cf.\ \citealp{kallenberg2021foundations}]\label{thm:perturb}
Suppose the same measurable semantic map is used for $Q$ and $Q'$, and their expressed masses satisfy $M,M'\ge m_0>0$. Then
\[
\TV(\boldsymbol p,\boldsymbol p')
\le \min\left\{1,\frac{2\TV(Q,Q')}{m_0}\right\}.
\]
\end{theorem}

The data-processing inequality behind Theorem~\ref{thm:perturb} controls the unnormalised pushforward and is standard. It does not by itself control the distribution conditional on falling inside the declared semantic states. The factor $2/m_0$ arises from this additional conditioning step and exposes the paper-specific failure mode in which small perturbations are amplified by low expressed mass.

\begin{corollary}[Stability after calibration]\label{cor:calibrated-perturbation}
If a calibrated map $\psi:\Delta_K\to\Delta_K$ satisfies
\[
d_{\Pi}\{\psi(\boldsymbol a),\psi(\boldsymbol b)\}
\le L_\psi\TV(\boldsymbol a,\boldsymbol b)
\quad\text{for all }\boldsymbol a,\boldsymbol b\text{ in the operating domain},
\]
for a posterior metric $d_{\Pi}$, then
\[
d_{\Pi}\{\psi(\boldsymbol p),\psi(\boldsymbol p')\}
\le \frac{2L_\psi}{m_0}\TV(Q,Q').
\]
\end{corollary}

\noindent\textbf{Discussion.} Return to the three-vector example and suppose, for clarity, that the three outcomes exhaust the continuation law, the semantic map is the identity and hence $M=M'=m_0=1$. Theorem~\ref{thm:perturb} bounds the change after the paraphrase by $2(.01)=.02$, but gives only the trivial bound $1$ after the negation because $2(.65)>1$. The actual distances, $.01$ and $.65$, lie within these bounds. If calibration is unit-Lipschitz in total variation, Corollary~\ref{cor:calibrated-perturbation} carries the $.02$ paraphrase bound to the posterior scale; for the negation its displayed bound is $1.30$ and therefore supplies no nontrivial guarantee. The example shows the intended distinction: close continuation laws with adequate expressed mass imply controlled state-level changes, whereas a large change in the law need not be stable. Low expressed mass and an unstable calibration map provide two further amplification mechanisms.

\noindent\textbf{Empirical assessment.} Under free-text generation, top-$k$ probabilities identify only a subprobability measure unless omitted mass is bounded. The controlled experiment instead augments three outcomes under $Q_u^G$ by the residual mass $1-\sum_{k=1}^3Q_u^G(v_k\mid c)$ and evaluates total variation on that finite law. The resulting claim concerns the constrained experiment, not unrestricted prose. Perturbations are prespecified as nuisance changes that preserve information, relevant changes that alter evidence, or irrelevant additions that test distraction. A useful channel must be stable to the first and responsive to the second; constancy alone is uninformative.

The presentation experiment therefore compares information-equivalent rewrites with prespecified changes to relevant evidence. The former should induce small state-probability changes; the latter should produce changes commensurate with those in the reference posterior. These are properties of the declared perturbation design, not universal constants of a fitted language model.

\subsection{Posterior recovery as an inverse problem}
We have now established what presentation stability can and cannot mean: small changes in the full response law remain controlled after semantic grouping when expressed mass is adequate, but textual similarity by itself guarantees nothing. This completes the forward part of the argument. We are now ready to state the paper's principal inverse problem formally.

Once the language-derived composition is well defined, the next question is whether it contains the right information. A measurement can be perfectly reproducible yet scientifically useless if substantially different reference posteriors produce nearly identical response probabilities. This is an inverse problem. Quantitative observability measures how strongly the language distribution separates posterior states, and the recovery theorem shows how that separation controls amplification of repeated-measure and probability error.

\subsubsection{Why the full distribution matters}
The modal response is a many-to-one summary: very different probability vectors can have the same largest component. For example, $(.51,.48,.01)$ and $(.98,.01,.01)$ select the same state, although the first is nearly tied and the second is highly concentrated. That distinction determines whether to abstain, seek more evidence or take a consequential action, and it enters any decision based on asymmetric losses, thresholds or the value of information. Modal accuracy cannot reveal whether the language channel preserves it.

Inverse recovery therefore asks for more than correct classification: do meaningful changes anywhere in the reference posterior remain distinguishable in the language-derived probabilities? If the forward map collapses such differences, no subsequent calibration can reconstruct them reliably. The lower observability modulus measures the weakest sensitivity of the language channel to a change in posterior log odds and hence whether full-distribution recovery is possible rather than merely modal agreement.

\subsubsection{Stable inverse recovery}
To state that calibration problem, assume the observable semantic log odds satisfy
\[
\boldsymbol\lambda_u=h_u(\boldsymbol\ell^\star)+\boldsymbol r,
\]
where $\boldsymbol r:=\boldsymbol\lambda_u-h_u(\boldsymbol\ell^\star)\in\mathbb R^{K-1}$ is the residual disturbance for the fixed presentation rule $u$. It records departures from the conditional-mean relation, including repeated-measure and probability-calculation variation; presentation changes are represented separately by the index $u$ and the perturbation analysis above. Calibration estimates an inverse of $h_u$, not a hidden model state. Thus the earlier triage example becomes a collection of paired observations: language-derived log odds computed from phrases such as ``urgent review'' are used to estimate the known reference log odds for \emph{urgent}, \emph{routine}, and \emph{insufficient information}. The scientific question is whether distinct reference odds remain sufficiently separated after this forward map to permit stable inversion.

\begin{theorem}[Stable nonlinear recovery; cf.\ \citealp{engl1996regularization}]\label{thm:recover}
If $h_u\in C(\mathcal L;\mathbb R^{K-1})$ has lower modulus $c_u>0$, $\|\boldsymbol r\|_2\le\delta$, and
\[
\widehat{\boldsymbol\ell}\in\arg\min_{\boldsymbol\ell\in\mathcal L}\|\boldsymbol\lambda_u-h_u(\boldsymbol\ell)\|_2,
\]
then $\|\widehat{\boldsymbol\ell}-\boldsymbol\ell^\star\|_2\le2\delta/c_u$.
\end{theorem}

For example, if two posterior log-ratio vectors are one unit apart and their mean lexical log ratios are at least $0.5$ units apart, then $c_u=0.5$. If measurement error is bounded by $\delta=0.05$, recovery error is at most $2\delta/c_u=0.20$. A larger $c_u$ means that genuinely different posteriors remain more visibly different after language measurement.

This is a specialization of the general identification and inverse-modulus
theory in \citet{dixon2026semanticobservation}. The inverse inequality itself
is therefore not claimed as new here. Corollary~\ref{cor:combined} instead
turns its abstract perturbation radius into an auditable semantic error budget
derived from semantic-cell mismatch, probability calculation and unexpressed
mass. The radius $\delta$ measures noise entering the language-derived
coordinates, while $c_u^{-1}$ measures how strongly inversion magnifies that
noise. A small average held-out error can arise on a
favourable test distribution even when $c_u$ is nearly zero in an untested
direction. Correlation and predictive error are therefore evidence about
sampled behaviour, not proofs of identification.

\subsubsection{The combined semantic-to-posterior error budget}
\begin{corollary}[Semantic-to-posterior error budget]\label{cor:combined}
Suppose the conditions of Theorems~\ref{thm:approx} and~\ref{thm:recover}
hold, the ideal and estimated compositions lie in a compact interior simplex
on which the selected log-ratio map is $L_{\rm lr}$-Lipschitz from
$\ell_1$ to $\ell_2$, and other lexical-coordinate disturbance is bounded by
$\delta_0$.  Then
\[
\|\widehat{\boldsymbol\ell}-\boldsymbol\ell^\star\|_2
\le
\frac{2}{c_u}
\left\{\frac{2L_{\rm lr}}{m_0}
(e_{\rm sem}+e_{\rm score})+\delta_0\right\}.
\]
\end{corollary}

The corollary is the operational consequence needed here.  Semantic
misclassification and probability-calculation error are first amplified by
normalization when expressed mass is small, then amplified again when the
inverse is poorly conditioned. The proposed measurement procedure estimates or
bounds these components separately; the companion identification study
\citep{dixon2026semanticobservation} treats the inverse modulus, convergence
rates and weak-identification limits in full generality. For example, take
$e_{\rm sem}+e_{\rm score}=0.03$, $m_0=0.80$, $L_{\rm lr}=1$ and
$\delta_0=0$. Normalisation first gives a disturbance bound of $0.075$; if
$c_u=0.50$, inversion raises the log-ratio recovery bound to
$2(0.075)/0.50=0.30$. With expressed mass $m_0=0.20$, the same errors would
instead give a bound of $1.20$.

\subsubsection{Affine observability and its empirical diagnostic}
\begin{corollary}[Affine semantic-to-posterior error budget; cf.\ \citealp{horn2012matrix}]\label{cor:affine-budget}
Suppose the conditions of Corollary~\ref{cor:combined} hold and
$h_u(\boldsymbol\ell)=B_u\boldsymbol\ell+\boldsymbol d_u$, where $B_u$ has full column rank. Then
\[
 \|h_u(\boldsymbol\ell)-h_u(\boldsymbol\ell')\|_2
 \ge \sigma_{\min}(B_u)\|\boldsymbol\ell-\boldsymbol\ell'\|_2
\]
for all $\boldsymbol\ell,\boldsymbol\ell'$ in the declared domain, and the oracle affine inverse satisfies
\[
 \|\widehat{\boldsymbol\ell}-\boldsymbol\ell^\star\|_2
 \le
 \frac{2L_{\rm lr}(e_{\rm sem}+e_{\rm score})/m_0+\delta_0}
 {\sigma_{\min}(B_u)}.
\]
Thus $\sigma_{\min}(B_u)$ is a valid global Euclidean lower modulus for the affine forward map.
\end{corollary}

\noindent\textbf{Discussion.} Corollary~\ref{cor:affine-budget} gives the empirical singular-value diagnostic its theoretical interpretation. A small $\sigma_{\min}(B_u)$ means that the language measurement nearly erases some change in posterior log odds, allowing its inverse to amplify modest error. A scenario-bootstrap lower confidence bound for $\sigma_{\min}(B_u)$ therefore assesses the denominator of the affine error bound; it is not an unrelated diagnostic. Mean prediction loss alone cannot reveal such a weakly observed direction.

Figure~\ref{fig:inverse} illustrates this mechanism. Distinct reference posteriors must remain separated after language measurement; the calibrated inverse then maps the retained lexical differences back to the posterior scale.

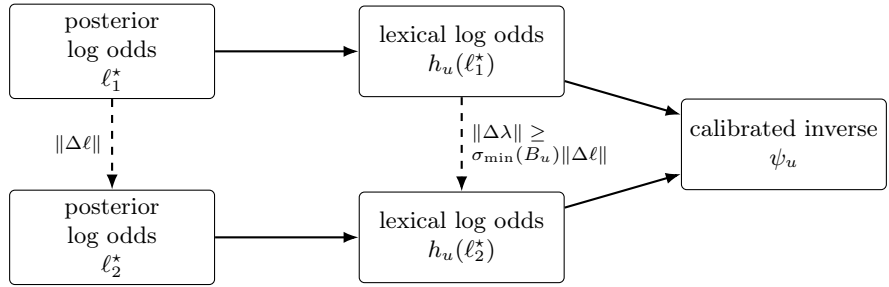
\begin{figure}[H]\centering
\begin{tikzpicture}[node distance=12mm and 19mm,font=\small]
\node[bbnode] (l1) {posterior log odds\\$\ell_1^\star$};\node[bbnode,below=of l1] (l2) {posterior log odds\\$\ell_2^\star$};
\node[bbnode,right=of l1] (z1) {lexical log odds\\$h_u(\ell_1^\star)$};\node[bbnode,right=of l2] (z2) {lexical log odds\\$h_u(\ell_2^\star)$};
\node[bbnode,right=of $(z1)!0.5!(z2)$,xshift=10mm] (inv) {calibrated inverse\\$\psi_u$};
\draw[bbmap] (l1)--(z1);\draw[bbmap] (l2)--(z2);\draw[bbdash] (l1)--node[left,font=\scriptsize]{$\|\Delta\ell\|$}(l2);\draw[bbdash] (z1)--node[right,align=left,font=\scriptsize]{$\|\Delta\lambda\|\ge$\\$\sigma_{\min}(B_u)\|\Delta\ell\|$}(z2);\draw[bbmap] (z1)--(inv);\draw[bbmap] (z2)--(inv);
\end{tikzpicture}
\caption{Quantitative observability and inverse recovery. For the affine forward map, the smallest singular value $\sigma_{\min}(B_u)$ is a lower modulus: it prevents two materially different reference posteriors from collapsing to the same lexical measurement, while $1/\sigma_{\min}(B_u)$ controls worst-direction amplification by the inverse. Good average correlation can still coexist with a nearly unobserved direction.}
\label{fig:inverse}
\end{figure}

\subsubsection{Finite-design observability and uncertainty assessment}\label{sec:empirical-observability}
The experiment asks whether the fitted affine map retains every log-ratio direction on the sampled design and recovers untouched posteriors. Both properties are necessary: a collapsed direction makes distinct posterior changes observationally indistinguishable, while held-out recovery tests whether the fitted inverse generalizes beyond its calibration sample. A scenario-clustered bootstrap lower bound for the smallest singular value addresses the first question; held-out losses and conformal coverage address the second. These results support finite-design observability, not unrestricted nonlinear recovery or extrapolation beyond the studied distribution. This qualification follows the broader nonparametric-inference literature: honest uniform statements generally require regularity restrictions and cannot be inferred adaptively from finite observations alone \citep{genovese2008adaptive,low1997nonparametric}. The companion paper, \emph{Identification and Learning of Semantic Observation Kernels: Partial Observation, Uniform Recovery, and Minimax Limits} \citep{dixon2026semanticobservation}, develops the stronger recovery questions.

\section{Experimental design for assessing the theoretical properties}
The experimental programme is organised around the same three questions introduced in Section~1. \emph{First}, does semantic measurement add value beyond printed numerical confidence? \emph{Second}, can it recover a known target posterior on new cases, with calibrated uncertainty? \emph{Third}, is it stable under presentation changes while still responding to genuine changes in evidence? The approximation theorem requires coverage and scoring diagnostics, the inverse-recovery theorem requires posterior-space variation and held-out calibration, and the perturbation theorem requires controlled changes to presentation and evidence. Organising the design in this order ensures that each reported statistic answers a declared scientific question.

The experimental logic is domain-general. For concreteness, one empirical study uses financial-market evidence because it is heterogeneous, time-stamped and professionally consequential; finance is the test bed, not part of the mathematical definition. The same design could study diagnostic states from patient records, incident states from cybersecurity reports or fault states from engineering logs, provided that the states, evidence, reference posterior and validation partitions are specified independently of the language measurement. Put differently, the application vocabulary changes, but the statistical questions do not. Transfer to another domain nevertheless requires new validation, so no claim of automatic external validity is made.

\paragraph{Information timing and look-ahead bias.} Two forms of look-ahead are excluded. First, evidence is restricted to information available at the stated information time; later outcomes cannot enter the prompt or reference construction. In a medical application, for example, a later confirmed diagnosis could not silently be inserted into an earlier triage record. In the market study, records use only levels and changes ending on or before their information dates. The five-trading-day horizon defines the requested interpretation, not a future return used as the response. Second, semantic states, response phrases, prompt specification, calibration class and scenario partitions are frozen before the untouched test data are opened. Validation may select a specification and conformal calibration may set a radius, but neither sees final-test outcomes. In other words, neither future information nor future test performance is allowed to leak backwards into the measurement. Because the locally cached Federal Reserve series do not reconstruct historical data revisions, the study supports a contemporaneous semantic-measurement claim rather than real-time trading or revision-free forecasting.

\paragraph{Discovery and confirmation.} Candidate phrases, semantic equivalence classes and presentation forms may be explored during discovery because language permits many plausible representations. For example, ``urgent review'' and ``immediate escalation'' may be tested as expressions of one state before confirmation begins. Once chosen, however, these choices become part of the measuring instrument. Altering them after seeing final-test errors would resemble changing a laboratory scale after observing the desired result. Prespecification preserves the meaning of held-out evidence.

\paragraph{Partitioned experimental protocol.} The protocol follows the mathematical dependencies rather than presenting an undifferentiated collection of summaries. The semantic partition and presentation rule are selected first; candidate phrases are then checked for shared beginnings, incomplete probability mass and semantic ambiguity. Entire scenarios are assigned to discovery, semantic validation, calibration, model selection, conformal calibration or final testing. The final test remains inaccessible until the calibration rule, reference state, candidate set and analysis are fixed. Put simply, choices may be learned from validation data, but they may not be rescued by the test set.

\paragraph{Question 1: added value beyond printed confidence.} The semantic-orientation experiment asks whether language-derived probabilities agree with a fixed reference more reliably than probabilities printed by the same model. The two channels see identical records and are judged against the same states. This is the most direct practical test: does the semantic construction retain information that numerical elicitation loses?

\paragraph{Manual audit of the benchmark responses.} The professional-text benchmark does not obtain its reference responses from a hidden Markov model or another fitted state-transition model. Each benchmark response was manually audited against the frozen state definitions and the contemporaneous evidence record before confirmatory evaluation. The audit checks whether the assigned response is supported by the declared semantic protocol; it does not use subsequent market outcomes and does not turn the study into a forecasting exercise. The controlled experiment is separate: there the reference posterior is known exactly from the prespecified static data-generating law.

\paragraph{Question 2: posterior recovery and calibrated uncertainty.} The controlled experiment spans interior and near-boundary posteriors, equally uncertain distributions pointing in different directions, and distributions with the same modal state but different residual mass. These cases ask whether the method preserves the whole probability vector rather than merely choosing the right state. Affine calibration is primary because its smallest singular value has a transparent observability interpretation. Diagnostics include held-out log-ratio and Jensen--Shannon error, clustered confidence intervals and split-conformal sets. In other words, the experiment asks both whether the posterior is recovered and whether its remaining error is honestly quantified.

\paragraph{Question 3: presentation stability and evidence responsiveness.} Prespecified paraphrases, reorderings and irrelevant additions test whether the measurement changes when the scientific information does not. Omissions and contradictions test the converse: whether it moves when the evidence genuinely changes. Both sides are necessary. A measurement that changes with every harmless rewrite is unstable, while one that never changes is uninformative.

\subsection{Reproducibility}
The controlled experiment was repeated without changing its scenario partition, semantic map, calibration method, or analysis for two prespecified fitted language models, GPT-4.1-mini and GPT-4.0-mini. Both used deterministic generation settings, five repeated measurements per scenario, and probabilities of the declared continuations. Each model contributed 2,600 observations over 520 scenarios: 240 for calibration, 80 for validation, 100 for conformal calibration, and 100 for untouched testing. Repetitions remain grouped by scenario. The semantic comparison uses 200 untouched records in 60 clusters and 5,000 cluster-bootstrap resamples. A broader descriptive sample contains 1,500 records; every candidate composition satisfied the prespecified probability-coverage rule, and two independent probability calculations agreed whenever both were available.

Table~\ref{tab:design} consolidates the experimental parameters that otherwise become difficult to reconstruct from separate results. The unit of analysis is always the scenario; repeated observations characterize conditional measurement variation and never inflate the effective sample size. The partitions have distinct statistical roles: validation selects a presentation or calibration specification, conformal calibration fixes an error radius, and the untouched test partition estimates final performance. No test scenario is used to select the inverse map or its uncertainty threshold.

\begin{table}[H]\centering\scriptsize
\caption{Experimental design and fixed analysis parameters. ``Per fitted model'' applies separately to both prespecified language distributions. Scenario counts, rather than repeated responses, determine the independent or exchangeable sampling units used for inference.}\label{tab:design}
\begin{tabularx}{\textwidth}{p{25mm}p{28mm}p{23mm}X}\toprule
Study component & Statistical target & Sample and partitions & Fixed measurement and inference settings\\\midrule
Semantic channel comparison & Agreement with a manually audited four-state semantic reference & 200 untouched records in 60 scenario clusters; broader descriptive set of 1,500 records & Risk-on, Mixed, Risk-off and Unsure; frozen state definitions; paired lexical and elicited measurements; 5,000 scenario-cluster bootstrap resamples; McNemar paired test\\
Controlled posterior recovery & Recovery of an exact three-state posterior distribution in log-ratio coordinates & Per fitted model: 520 scenarios split 240 calibration, 80 validation, 100 conformal and 100 untouched test & Five repeats per scenario (2,600 observations); affine log-ratio calibration; scenario clustering; 90\% split-conformal radius\\
Cross-model replication & Stability of conclusions across fitted language distributions & Identical 520-scenario design applied separately to GPT-4.1-mini and GPT-4.0-mini & Same partitions, semantic map, calibration class, deterministic sampling settings and analysis\\
Observability diagnostic & Retention of all fitted affine log-ratio directions & 240 calibration scenarios per model, with all repeats retained within cluster & Smallest singular value of fitted forward map; scenario-cluster bootstrap confidence interval\\
Calibration uncertainty & Sensitivity of recovered test posteriors to fitted inverse & 2,000 resamples of the 240 calibration scenarios per model & Refit calibration on each cluster resample; evaluate movement on the same untouched test scenarios\\
Presentation study & Stability to equivalent wording and sensitivity to omitted information & 24 validation scenarios balanced over three posterior-entropy strata; two repeats & Two information-preserving rewrites and three omissions; prespecified Jensen--Shannon equivalence margin; no final-test reuse\\
Perturbation extension & Resistance to irrelevant text and directional response to contradictory evidence & 12 validation scenarios per model, four from each entropy stratum; two repeats; 48 calls per model & One fixed irrelevant distractor; one observed variable replaced by the alternative maximizing exact posterior change; two primary contrasts\\
Probability calculation audit & Completeness of constrained enum law and agreement of extraction methods & All controlled observations for which both calculations were available & Raw code mass retained with an explicit residual outcome; absent candidates not set to zero; completeness threshold 0.999999; unrestricted free-text law not claimed\\\bottomrule
\end{tabularx}
\end{table}

The design supports three levels of conclusion. The paired semantic study asks whether two observable measurement channels differ under a common reference. The controlled study asks whether calibrated lexical coordinates recover an exact posterior out of sample and with advertised finite-sample coverage. The lexical-presentation and extraction studies locate sensitivity to wording, missing evidence and probability calculation.

Reproduction requires more than a model name. The archive records model versions, observation times, sampling settings, observed and normalised probability mass, presentation rules, scenario partitions and random seeds. The submission materials will include the fixed presentations, semantic definitions, scenario lists, probability tables, analysis programs and figure data. Since a fitted language distribution may change after publication, exact reproduction applies to the archived observations; new observations constitute statistical replications.

The natural-language corpus consists of deterministic short-form market commentary compiled from Federal Reserve Economic Data series covering the S\&P 500, CBOE VIX, two- and ten-year Treasury yields, and investment-grade and high-yield corporate credit spreads. It fixes US broad-equity risk orientation over the next five trading days as its case-study target and uses Risk-on, Mixed, Risk-off, and Unsure as the possible responses. The reference responses were manually audited under definitions frozen before confirmatory analysis, and the evaluation partition was likewise fixed in advance. The paired comparison therefore estimates relative performance against the same prespecified semantic target rather than allowing either measurement channel to define its own success criterion.

The reference construction is deliberately domain-general. A finite Bayesian network is useful when evidence components have a known conditional structure, but it is neither mandatory nor used to assign the professional-text benchmark responses. A medical study might use an expert-reviewed diagnostic model, a reliability study a prespecified degradation model, and a scientific-hypothesis study a finite likelihood family. The essential requirement is that the reference construction be reproducible and independent of the language measurement being evaluated.

Together, these design choices make the empirical claims narrow but reusable. The semantic-orientation study asks whether continuation-derived and numerically elicited probabilities agree with a semantic reference; the controlled-posterior study asks whether the former preserve known reference posterior distributions in log-ratio coordinates. The result is a validated static measurement from which a separate dynamic study could subsequently proceed.

\section{Experimental results}
We now answer the three questions in the order posed in Sections~1 and~5. Semantic measurement first is compared with printed numerical confidence; posterior recovery and uncertainty coverage are then tested on untouched cases; finally, presentation stability is considered alongside responsiveness to changed evidence. Table~\ref{tab:tests} gives the answers at a glance before the supporting detail.

\begin{table}[H]\centering\small
\caption{Answers to the paper's three central questions and their supporting diagnostics. ``Yes'' and ``mostly'' refer to the prespecified experimental design, not every model, prompt or state space.}\label{tab:tests}
\begin{tabularx}{\textwidth}{p{42mm}p{24mm}X}\toprule
Question & Verdict & Main evidence\\\midrule
Does semantic measurement add value beyond printed confidence? & Yes & Higher accuracy and lower log and Brier losses on the same professional-text records\\
Can it recover the target posterior? & Yes, on the tested design & Mean held-out Jensen--Shannon error near $0.04$ across two fitted models\\
Is uncertainty on new cases calibrated? & Yes & The 90\% conformal sets cover 94\% and 90\% of untouched cases\\
Is it stable under harmless presentation changes? & Mostly & Paraphrases are stable; evidence ordering is a material exception\\
Does it respond to genuine evidence changes? & Yes & Omission worsens recovery and contradiction moves estimates in the correct direction\\
Does raw word entropy equal posterior uncertainty? & No & Their correlation is approximately zero before semantic calibration\\\bottomrule
\end{tabularx}
\end{table}

These findings support the paper's measurement theory on the declared design. They do not establish universal prompt invariance or validate unrestricted free-text generation. Each subsection begins with its substantive verdict; readers seeking inferential detail can then inspect the accompanying table or figure.

\subsection{Does semantic measurement improve on printed confidence?}
\paragraph{Verdict: yes, in the professional-text experiment.} Continuation-derived measurement raises accuracy by seven percentage points and improves both proper losses. Asking a model to print a probability therefore discards information present in its distribution over language.

\paragraph{Design and evidence.} The untouched test contains 200 records in 60 scenario clusters and four semantic states, including insufficient information. Candidate-phrase probabilities form the lexical measurement; printed probabilities form the elicited measurement. Both are evaluated on identical records under one reference protocol, with 5,000 cluster-bootstrap resamples. Table~\ref{tab:semantic} reports the paired results summarized in Table~\ref{tab:tests}.

\begin{table}[H]\centering\small
\caption{Untouched comparison between continuation-derived lexical measurement and direct numerical elicitation under the same fixed semantic reference protocol. This table tests the motivating distinction between two observable channels; it is not a direct test of a theorem. Differences are lexical minus elicited; negative values favour lexical measurement for proper losses.}\label{tab:semantic}
\begin{tabular}{lccc}\toprule
Estimand & Lexical & Elicited & Paired difference (95\% CI)\\\midrule
Accuracy & 0.885 & 0.815 & $0.070\ [0.011,0.126]$\\
Log loss & --- & --- & $-0.788\ [-1.371,-0.317]$\\
Brier score & --- & --- & $-0.144\ [-0.214,-0.078]$\\\bottomrule
\end{tabular}
\end{table}

\paragraph{What this checks.} McNemar's test gives $p=0.0243$, and both proper-loss intervals exclude zero. The two channels are thus empirically different, with the semantic measurement preferred here. This motivates the construction; it is not a claim that it must dominate in every application. Transfer requires a new state space and reference protocol.

\paragraph{Channel diagnostics.} Figure~\ref{fig:semantic-results} shows why the measurements should remain separate: their calibration patterns and scenario-level disagreements differ rather than representing interchangeable versions of one probability.

\begin{figure}[H]\centering
\begin{subfigure}{.49\linewidth}\includegraphics[width=\linewidth]{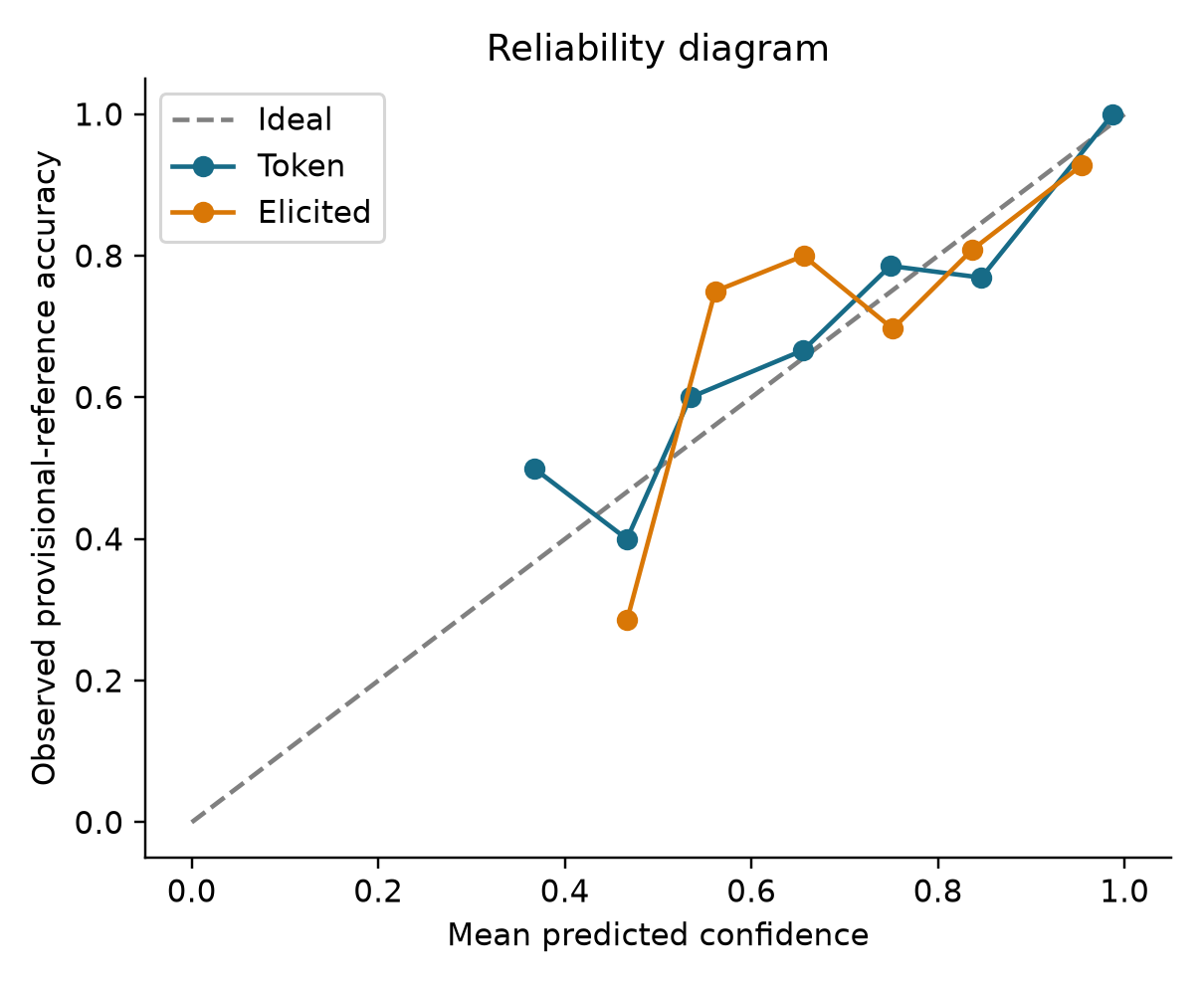}\caption{Classwise calibration diagnostics.}\end{subfigure}\hfill
\begin{subfigure}{.49\linewidth}\includegraphics[width=\linewidth]{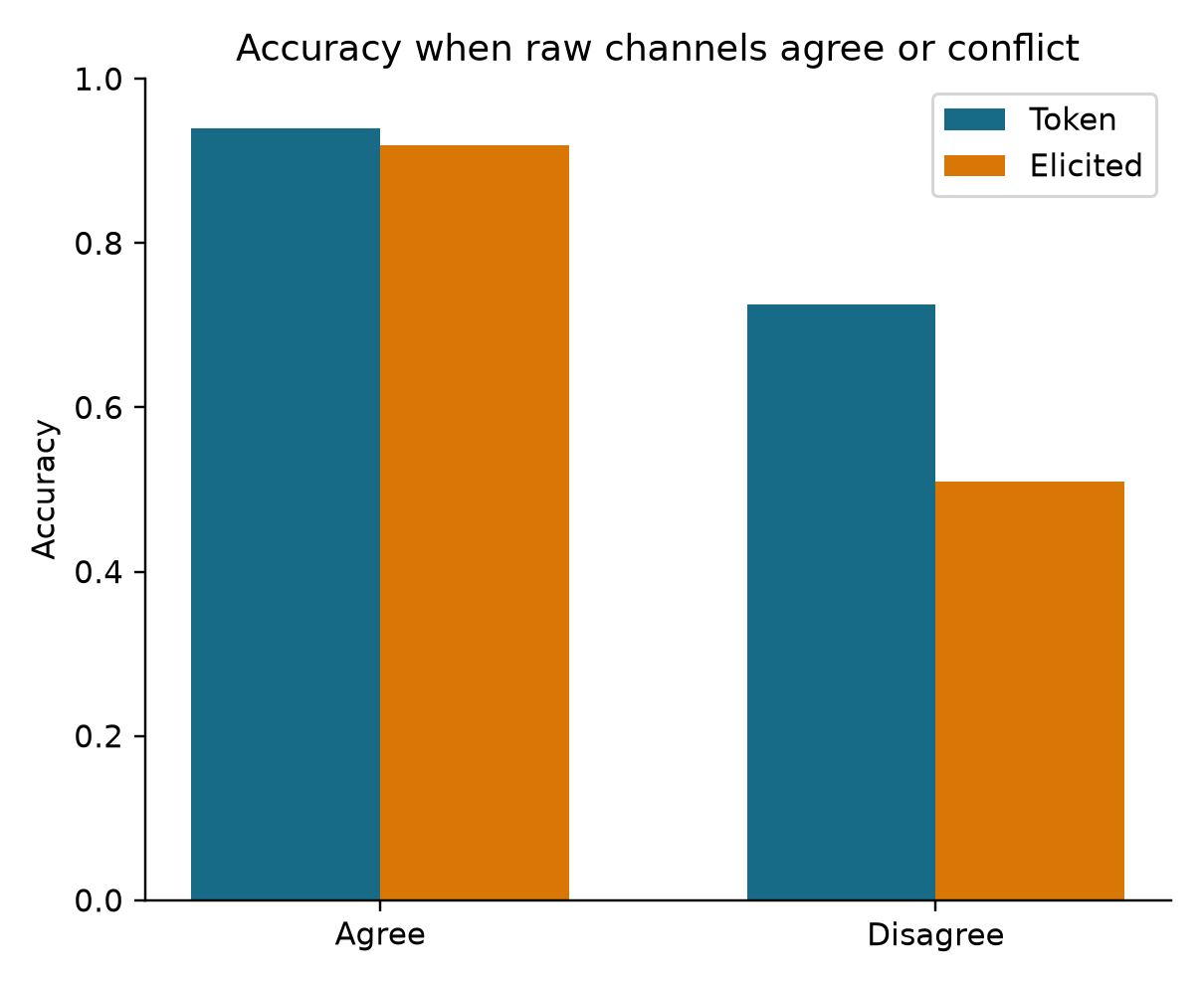}\caption{Lexical--elicited disagreement.}\end{subfigure}
\caption{Two views of measurement quality in the semantic study. These are diagnostics of the paper's motivating channel distinction rather than direct tests of a theorem. The calibration panels assess whether reported probability levels agree with observed response frequencies, while the disagreement plot shows that continuation-derived and elicited probabilities are empirically distinct rather than interchangeable versions of one quantity.}
\label{fig:semantic-results}
\end{figure}

\subsection{Can the semantic measurement recover a posterior?}
\paragraph{Verdict: yes, on the controlled design.} Both fitted models recover untouched exact posteriors with mean Jensen--Shannon error near $0.04$; their 90\% uncertainty sets cover 94\% and 90\% of new cases. This is the paper's central test of posterior recovery rather than modal classification.

\paragraph{What this checks.} Near-complete measured mass supports the premises of Lemma~\ref{lem:exist} and Theorem~\ref{thm:approx}; held-out error checks Theorem~\ref{thm:recover} and Corollary~\ref{cor:combined}; positive singular-value bounds check the affine recovery condition in Corollary~\ref{cor:affine-budget}. Put simply, the language coordinates retain enough posterior variation to permit a stable fitted inverse on the sampled design.

Each controlled replication contains 520 scenarios and 2,600 repeated observations. Calibration is fitted out of sample, and uncertainty is clustered by scenario. Table~\ref{tab:crossmodel} shows that posterior recovery is not confined to the first fitted model: the second produces similar held-out error and nominal conformal coverage.

\begin{table}[H]\centering\scriptsize
\caption{Cross-model exact-posterior recovery under an identical prespecified experiment. Held-out error assesses Theorem~\ref{thm:recover} and Corollary~\ref{cor:combined}; the singular-value interval assesses Corollary~\ref{cor:affine-budget}; conformal coverage is a separate predictive-uncertainty assessment.}\label{tab:crossmodel}
\begin{tabular}{lcccc}\toprule
Fitted model & Mean Jensen--Shannon error (95\% CI) & 90\% radius & Coverage & $\sigma_{\min}$ (95\% bootstrap CI)\\\midrule
GPT-4.1-mini-2025-04-14 & $0.0390\ [0.0295,0.0498]$ & 0.1371 & 0.94 & $1.892\ [1.686,2.061]$\\
GPT-4.0-mini-2024-07-18 & $0.0366\ [0.0286,0.0454]$ & 0.0922 & 0.90 & $1.848\ [1.600,2.028]$\\\bottomrule
\end{tabular}
\end{table}

Table~\ref{tab:budget} quantifies the uncertainty decomposition on untouched test scenarios. The dominant term is stable scenario-level recovery error rather than within-scenario variation, which contributes 1.79\% and 0.72\% of squared Aitchison error. Resampling the 240 calibration scenarios shows that estimation of the inverse map is smaller still. Additional repetitions would therefore remove little of the observed error; improvement should instead target the semantic representation and calibration across scenarios.

\begin{table}[H]\centering\scriptsize
\caption{Empirical uncertainty budget in squared Aitchison distance on 100 untouched scenarios with five repetitions each. The first three columns assess equation~\eqref{eq:decomp}; calibration resampling diagnoses uncertainty in the inverse entering Corollary~\ref{cor:combined}. The calibration column is nested within recovery error and is not added again.}\label{tab:budget}
\begin{tabular}{lcccc}\toprule
Fitted model & Total error & Within-scenario variation (share) & Stable recovery (share) & Calibration estimation\\\midrule
GPT-4.1-mini & 3.859 & 0.069 (1.79\%) & 3.790 (98.21\%) & $0.067\ [0.006,0.231]$\\
GPT-4.0-mini & 3.276 & 0.024 (0.72\%) & 3.252 (99.28\%) & $0.047\ [0.006,0.143]$\\\bottomrule
\end{tabular}
\end{table}

The semantic approximation decomposition clarifies what these experiments do and do not establish. Probability-calculation error and candidate-mass coverage are measured directly; semantic-cell mismatch is governed by the frozen reference protocol and must be examined through independent semantic review or prespecified alternative partitions. Repeated probability measurements cannot identify a mistaken semantic map, because repeating the same assignment reproduces the same mistake. This separation is practically useful: it prevents excellent numerical reproducibility from being misreported as evidence that the state definitions themselves are valid.

The observation procedure contributes no visible discrepancy at the reported precision: the maximum difference between two probability calculations is zero, mean candidate-mass deficit is below $4.8\times10^{-8}$, and the maximum deficit is below $2.0\times10^{-7}$ in both studies. In the terminology of Theorem~\ref{thm:approx}, these findings make the measured contribution of $e_{\rm score}$ and the observed mass deficit negligible for the declared candidate phrases. They do not prove that every semantic expression has been enumerated or that $e_{\rm sem}=0$. Thus truncation and numerical calculation do not explain the test error for the prespecified phrases, while semantic validity remains a distinct design requirement. Because Aitchison distance emphasises relative error near the simplex boundary, its scale is not comparable with Jensen--Shannon divergence in Table~\ref{tab:crossmodel}.

\noindent\textbf{Observability interpretation.} The positive bootstrap lower bounds show that the \emph{fitted affine forward maps} retain both additive log-ratio directions on the sampled calibration design. This is stronger than reporting a condition number alone. It is not a global identification theorem: the empirical pairwise lower modulus is effectively zero for GPT-4.1-mini and zero at machine precision for GPT-4.0-mini, and affine residual root-mean-square errors are 7.62 and 7.35 in lexical log-ratio units. The data support aggregate affine recoverability on the sampled design, not uniform injectivity over an unrestricted function class.

The full-law premise was also audited on the structured response law. The minimum summed probability of its three admissible codes was $0.999999808$ for GPT-4.1-mini and $0.999999816$ for GPT-4.0-mini, both above the frozen $0.999999$ threshold; the residual was retained as a fourth outcome. This supports the observed-law premise used by Lemma~\ref{lem:exist} and Theorem~\ref{thm:approx} for the constrained experiment, not for unrestricted prose.

Figure~\ref{fig:posterior-results} supplies the visual counterpart to Table~\ref{tab:crossmodel}: its left panel concerns recovery under Theorem~\ref{thm:recover}, while its right panel concerns the separate conformal-coverage assessment.

\begin{figure}[H]\centering
\begin{subfigure}{.49\linewidth}\includegraphics[width=\linewidth]{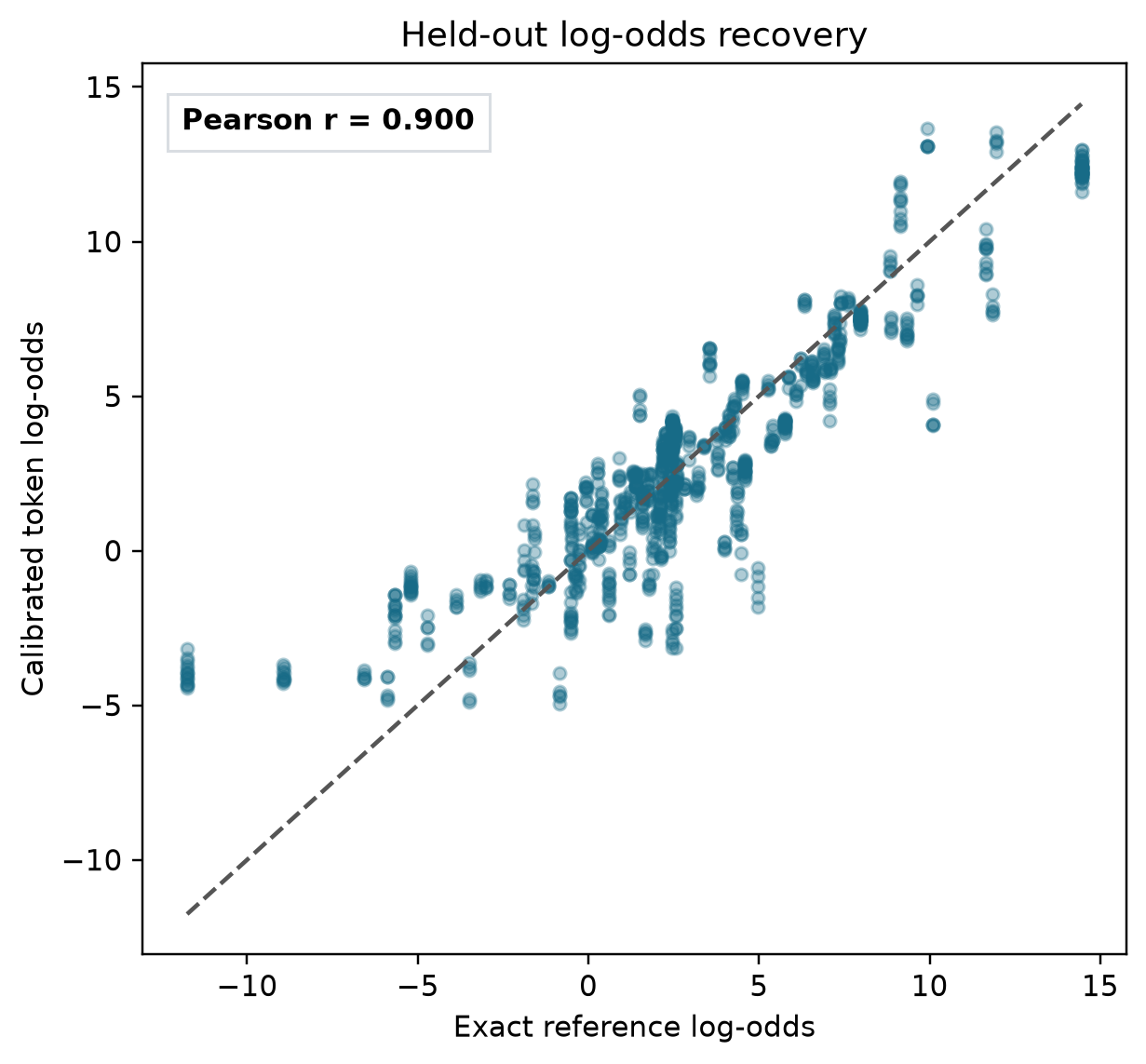}\caption{Reference versus calibrated log odds.}\end{subfigure}\hfill
\begin{subfigure}{.49\linewidth}\includegraphics[width=\linewidth]{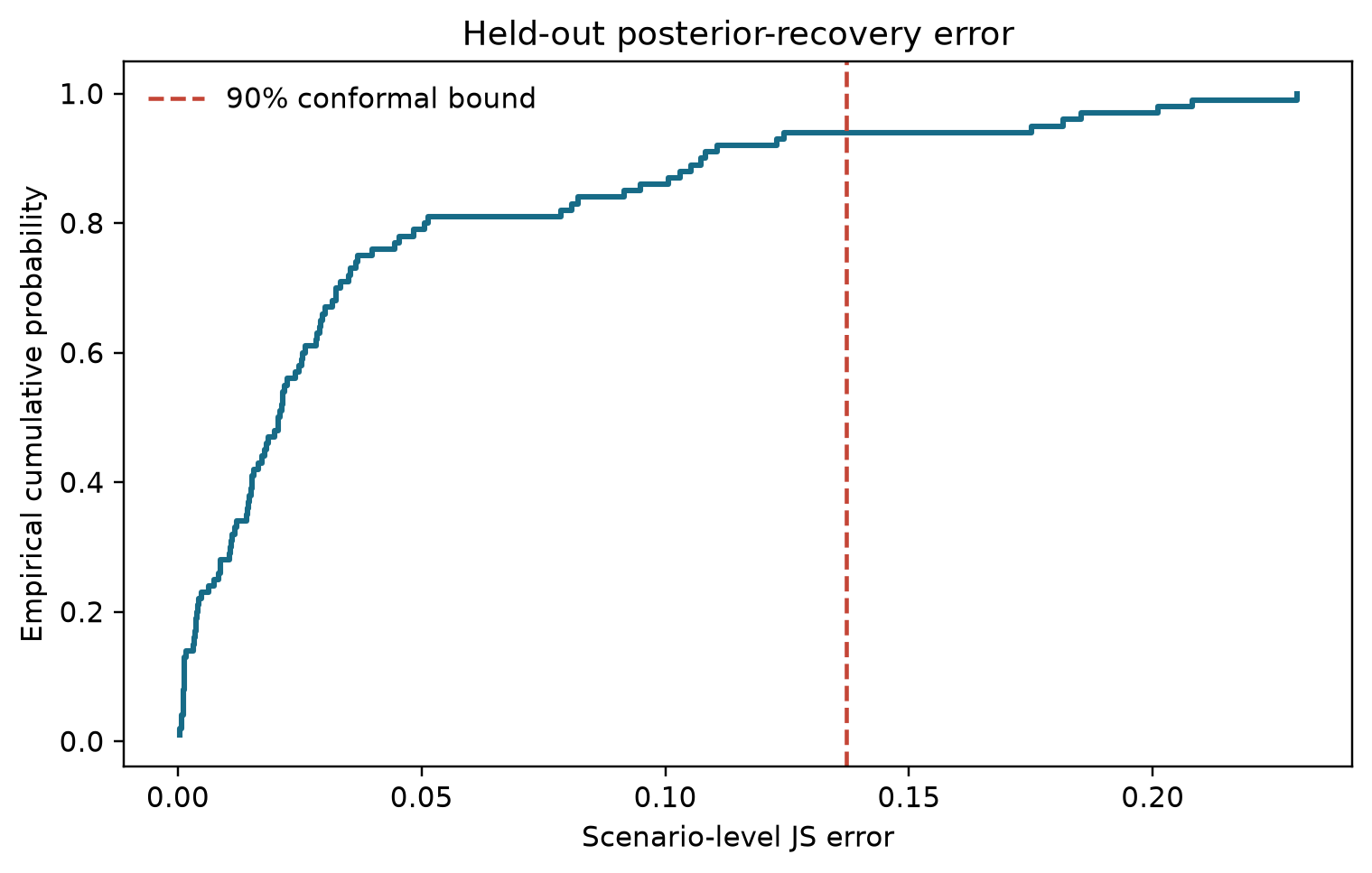}\caption{Error distribution and conformal threshold.}\end{subfigure}
\caption{Evidence for inverse recovery and finite-sample uncertainty. The left panel tests whether lexical coordinates preserve the directions and scale of exact posterior log odds after calibration; the displayed correlation is a summary rather than an identification proof. The right panel evaluates individual-scenario errors against the frozen conformal threshold. Thus only the left panel bears directly on Theorem~\ref{thm:recover}; the right evaluates predictive coverage under exchangeability.}
\label{fig:posterior-results}
\end{figure}

Entropy of the uncalibrated candidate probabilities is essentially uncorrelated with exact posterior entropy ($r\approx-0.001$), whereas entropy of the calibrated lexical posterior correlates $0.742$. Uncertainty over a finite set of words is therefore not automatically uncertainty over the scientific state space. This is a diagnostic of the measurement distinction, not a theorem test: semantic grouping and calibration create the bridge.

\subsection{Is the measurement stable without ignoring evidence?}
\paragraph{Verdict: mostly, with an important qualification.} Ordinary paraphrases are largely stable, and genuine evidence changes move the posterior in the expected direction. Evidence ordering is not harmless, however. The measurement therefore passes the finite perturbation design, but not a universal test of prompt invariance.

\paragraph{What this checks.} Theorem~\ref{thm:perturb} and Corollary~\ref{cor:calibrated-perturbation} control presentation-induced movement. A useful measurement must also respond to relevant evidence; otherwise a constant output would appear perfectly stable. The experiment therefore considers paraphrases and irrelevant text separately from omissions and contradictions.

\paragraph{Main finding: wording was mostly stable, ordering was not.} Across 24 scenarios and two repeats, information-preserving paraphrases remained within the prespecified Jensen--Shannon margin in 93.8\% and 100\% of cases. Merely placing evidence before the state definitions reduced those rates to 70.8\% and 81.3\%. Thus ordinary paraphrase was largely benign in this design, but information-equivalent ordering could not be treated as harmless.

\paragraph{The estimate used the evidence.} Removing evidence increased reference-posterior Jensen--Shannon error by $0.173\ [0.105,0.244]$ and $0.200\ [0.133,0.268]$. These intervals exclude zero. The result supplies the complement to Theorem~\ref{thm:perturb}: the channel was not merely stable; it deteriorated when relevant information was withheld. Removing the prior or detailed definitions did not show a clear mean effect in this small study.

\paragraph{Distraction was modest; contradiction moved the estimate.} A prespecified extension used 12 scenarios per fitted model. An irrelevant administrative sentence produced small average movements: constrained-law total variation was $0.047$ and $0.072$, while posterior Jensen--Shannon movement was $0.0168$ and $0.0084$. Contradictory evidence produced larger law movements, $0.174$ and $0.258$, and the calibrated posterior moved in the correct reference direction in 95.8\% and 100\% of repeats, with mean directional cosines $0.838$ and $0.913$.

\paragraph{What can be claimed.} The results support finite-design stability for the tested paraphrases and limited resistance to distraction, together with strong evidence responsiveness. They do not establish universal prompt invariance: the ordering result is a concrete counterwarning, and the distractor intervals are wide. Table~\ref{tab:perturb-results} collects the evidence and distinguishes the stability claims supported by Theorem~\ref{thm:perturb} from the complementary responsiveness diagnostics.

\begin{table}[H]\centering\scriptsize
\caption{Finite-design presentation stability and evidence responsiveness. Stability rows assess whether the bounds in Theorem~\ref{thm:perturb} and Corollary~\ref{cor:calibrated-perturbation} are empirically useful. Responsiveness rows are complementary diagnostics: they cannot be deduced from an upper stability bound. Results are reported for GPT-4.1-mini and GPT-4.0-mini, respectively.}\label{tab:perturb-results}
\begin{tabularx}{\textwidth}{p{28mm}p{28mm}p{40mm}X}\toprule
Intervention & Theoretical role & Reported result & Interpretation\\\midrule
Information-preserving rewrite & Nuisance stability & Equivalence rates 93.8\% and 100\% & Strong stability for the tested paraphrases\\
Evidence-first reordering & Presentation stability & Equivalence rates 70.8\% and 81.3\% & Information equivalence did not ensure measurement equivalence\\
Evidence omission & Relevant-evidence responsiveness & JS-error increases $0.173\ [0.105,0.244]$ and $0.200\ [0.133,0.268]$ & Removing evidence materially degraded recovery\\
Irrelevant distractor & Nuisance stability & Law TV $0.047$, $0.072$; posterior JS $0.0168$, $0.0084$ & Small average movement, with uncertainty in the law-level estimates\\
Contradictory evidence & Directional responsiveness & Law TV $0.174$, $0.258$; correct direction 95.8\%, 100\% & The estimate generally moved with the changed reference posterior\\\bottomrule
\end{tabularx}
\end{table}

\subsection{Theory-to-evidence synthesis}
The evidence now maps back to Section~4. Near-complete structured response mass makes the existence and normalization premises of Lemma~\ref{lem:exist} credible for the constrained law. Agreement between probability calculations makes the measured contribution of $e_{\rm score}$ in Theorem~\ref{thm:approx} negligible, although the experiments do not identify all semantic-cell error $e_{\rm sem}$. Rewrite and distractor studies provide finite-design evidence for the stability mechanism in Theorem~\ref{thm:perturb} and Corollary~\ref{cor:calibrated-perturbation}, while omissions and contradictions establish the complementary responsiveness that an upper bound cannot provide. Held-out recovery, positive affine singular-value bounds and conformal coverage support Theorem~\ref{thm:recover} and Corollaries~\ref{cor:combined}--\ref{cor:affine-budget} on the sampled design.

The support is deliberately bounded. The results do not establish a uniformly positive nonlinear observability modulus, semantic validity under every alternative state definition or invariance across unrelated fitted-model families. Within those limits, the observable premises of the measurement theory receive direct empirical support.

\section{Conclusion}
This paper asked whether a prompt-dependent distribution over language can yield a reliable posterior over scientifically meaningful states. The semantic map makes that question testable: it declares the meanings, retains unmatched probability mass, calibrates against an independent target and supplies explicit error and recovery conditions. The three questions posed in Section~1 can now be answered directly.

\paragraph{First: does semantic measurement add value beyond printed confidence? Yes.} On professional text, continuation-derived probabilities achieve higher accuracy and lower proper losses than probabilities printed by the same fitted model. The two channels are not interchangeable.

\paragraph{Second: can it recover the target posterior with calibrated uncertainty? Yes, on the tested design.} Across two controlled experiments, the method recovers untouched exact posteriors with mean Jensen--Shannon error near $0.04$, while 90\% conformal sets cover 94\% and 90\% of new cases. Near-complete measured mass and positive affine singular-value bounds support the approximation and stable-recovery conditions on the sampled design.

\paragraph{Third: is it stable under presentation changes while responding to evidence? Mostly.} Ordinary paraphrases are largely stable, and omissions or contradictions move the estimate as expected; changing evidence order, however, remains consequential. Raw word entropy is also essentially unrelated to posterior entropy before semantic calibration. Prompt engineering therefore changes a measurement instrument without establishing stability of the meaning it reports.

The conclusion is conditional but consequential. Probabilities over language can support posterior inference only after meaning has been declared, recovery identified and performance validated out of sample. This is the paper's importance for runtime governance: it replaces trust in a fluent response with an auditable rule for accepting the estimate, requesting more evidence, escalating to human review or abstaining. The rule is conditional on a specified state space, prompt family, fitted model and target population; it is not evidence of an inaccessible internal belief or validation of unrestricted free-text generation. More broadly, the paper shows how uncertainty from a generative system can be brought within a standard statistical framework before it enters scientific or professional use.

\subsection*{Limitations and responsible interpretation}
Lexical measurement can reveal uncertainty that a fluent answer or rounded numerical report conceals. Appropriate uses include abstention, requests for additional evidence, and prioritisation of human review. The method can mislead if conditional calibration is advertised as proof that a model ``knows'' the truth, if the semantic partition excludes clinically or socially meaningful alternatives, or if results are transferred across presentations, populations or fitted models. High-consequence users should retain domain review, monitor coverage over time, document the decision costs attached to every state, and provide a safe alternative outside the validated domain.

The experiments use public or generated market commentary rather than private personal records, and the paper makes no claim of suitability for direct medical, legal, employment, or credit decisions. Future studies involving personal or sensitive data require appropriate privacy review, data governance, and, where applicable, institutional ethics approval. The method evaluates a probabilistic measurement; it does not authorise delegation of the underlying decision.

\section*{Data availability}
The prespecified presentations, scenario lists, observed probability tables, analysis programs, and figure data will be deposited with the accepted article or in a persistent public repository. Where original observations cannot be redistributed, derived probability tables and accompanying sampling information will reproduce every reported statistic. Exact computational reproduction applies to archived observations; newly generated observations constitute statistical replications.

\section*{Funding}
None declared.

\bibliographystyle{abbrvnat}\bibliography{main}

\appendix
\section{Notation and terminology}\label{app:notation}
The two tables below provide a compact cross-reference for readers less
familiar with probabilistic language models. They are not needed to define the
statistical experiment and are therefore kept outside the principal
exposition.

\begin{table}[H]\centering\scriptsize
\caption{Language-model terminology, notation and statistical interpretation.}\label{tab:translation}
\begin{tabularx}{\linewidth}{>{\raggedright\arraybackslash}p{.19\linewidth}>{\raggedright\arraybackslash}p{.25\linewidth}X}\toprule
Language-model term & Notation & Statistical interpretation\\\midrule
Prompt specification and context & $u,\ C_u(Y),\ c=C_u(y)$ & A design rule and the resulting textual covariate constructed from evidence\\
Token & $\mathcal V;\ W_j,\ v_j$ & An element of the discrete vocabulary, possibly a word, subword or punctuation mark\\
Continuation & $\mathcal R;\ W,\ r$ & A finite sequence-valued auxiliary observation generated after conditioning on the context\\
Next-token probability & $\mathbb P(W_j=v_j\mid c,W_{1:j-1})$ & One factor in the conditional probability of the complete response\\
Continuation probability & $Q_u(\{r\}\mid c)$ & Conditional probability of a complete candidate phrase $r\in\mathcal R$, calculated from its successive token probabilities\\
Numerically elicited probability & $\boldsymbol q_u^{\rm elic}(c)$ & A probability vector printed as generated text; a separate observable channel\\
Semantic state & $\mathcal X=\{x_1,\ldots,x_K\};\ X$ & A category in the application-defined state space\\
Semantic grouping & $\phi_u:\mathcal R\to\mathcal X_0$ & A measurable coarsening that maps meaning-equivalent continuations to one state or to $x_0$\\
Unexpressed mass & $x_0;\ 1-M_u(c)$ & Probability assigned outside the declared semantic states\\
Lexical composition & $\boldsymbol p_u(c)\in\Delta_K$ & The probability vector induced by grouping and declared normalisation\\\bottomrule
\end{tabularx}
\end{table}

\begin{table}[H]\centering\small
\caption{Probability quantities and their distinct inferential roles.}\label{tab:objects}
\begin{tabularx}{\linewidth}{>{\raggedright\arraybackslash}p{.18\linewidth}>{\raggedright\arraybackslash}p{.22\linewidth}>{\raggedright\arraybackslash}p{.27\linewidth}X}\toprule
Quantity & Notation & Interpretation & Principal uncertainty\\\midrule
Reference posterior & $\boldsymbol\pi^\star(y)$ & Probability under the declared reference experiment & Evidence model, prior and likelihood\\
Lexical composition & $\boldsymbol p_u(c)$ & Semantic coarsening of the continuation law for a fixed prompt specification & Repeated measurement and semantic partition\\
Calibrated estimate & $\widehat{\boldsymbol\pi}_u(y)$ & Inverse-calibrated estimate of the reference posterior & Calibration sample and misspecification\\
Unexpressed mass & $1-M_u(c)$ & Probability outside the declared semantic cells & Incomplete semantic or verbal coverage\\\bottomrule
\end{tabularx}
\end{table}

\section{Methodological lineage}\label{app:lineage}
Table~\ref{tab:lineage} records the precise division between established
statistical machinery and the paper-specific construction. It complements
the question-led literature review in Section~1 without interrupting the
formal setup.

\begin{table}[H]\centering\small
\caption{Established foundations and paper-specific contributions. ``Standard'' means that the mathematical construction is inherited from the cited literature; novelty lies in the stated specialisation, synthesis or empirical test.}\label{tab:lineage}
\begin{tabularx}{\linewidth}{>{\raggedright\arraybackslash}p{.22\linewidth}>{\raggedright\arraybackslash}p{.31\linewidth}X}\toprule
Component & Established foundation & Paper-specific role\\\midrule
Conditional response law & Markov kernels and pushforward measures \citep{kallenberg2021foundations} & Treat the observable distribution over complete verbal continuations as the sampling law\\
Semantic grouping & Meaning-equivalence classes for language uncertainty \citep{farquhar2024semantic,kuhn2023semantic} & Fix an application-defined partition before evaluation, retain an unexpressed category, and calibrate against an external posterior\\
Finite semantic approximation & Normalisation and total-variation inequalities \citep{kallenberg2021foundations} & Derive the paper's semantic-cell, probability-calculation and expressed-mass decomposition and its state-probability bound\\
Probability-vector representation & Compositions and log-ratio coordinates \citep{aitchison1982statistical,egozcue2003isometric} & Compare semantic and reference probability distributions in unconstrained coordinates\\
Identification and inversion & Comparison of experiments and inverse-problem stability \citep{blackwell1953equivalent,engl1996regularization,lecam1964sufficiency,torgersen1991comparison} & Join the inverse modulus to semantic approximation errors in a combined posterior-recovery bound\\
Predictive uncertainty & Split-conformal coverage under exchangeability \citep{angelopoulos2023conformal,lei2018distribution,vovk2005algorithmic} & Construct individual-scenario error sets for the calibrated semantic posterior\\
\bottomrule
\end{tabularx}
\end{table}

\section{Proofs}
This appendix supplies proofs of the paper-specific formal results in the
order in which they appear.  The general inverse-recovery result is cited from
the companion identification paper; its semantic-to-posterior specialization
is proved here. Throughout, total variation is defined by
\[
d_{\rm TV}(\mu,\nu)=\sup_{A}|\mu(A)-\nu(A)|,
\]
and all displayed conditional probability measures are evaluated at the fixed context appearing in the relevant statement. This convention avoids repeatedly writing the conditioning argument. Vectors are treated as column vectors, and $\|\cdot\|_1$ and $\|\cdot\|_2$ denote the usual vector norms.

\begin{proof}[Proof of Lemma~\ref{lem:exist}]
\noindent\emph{Step 1: define the candidate image measure.}
Fix $u$ and $c$, and abbreviate $Q_u(\cdot\mid c)$ by $Q$. Let $(\mathcal X_0,\mathscr X_0)$ denote the measurable target space, where $\mathcal X_0=\mathcal X\cup\{x_0\}$ and $\mathscr X_0=\mathcal B(\mathcal X_0)$. Define the set function $R:\mathscr X_0\to[0,1]$ by
\[
R(A)=Q\{\phi_u^{-1}(A)\},\qquad A\in\mathscr X_0.
\]
We verify the probability axioms explicitly.

\noindent\emph{Step 2: establish well-definedness and total mass.}
First, because $\phi_u:(\mathcal R,\mathscr R)\to(\mathcal X_0,\mathscr X_0)$ is measurable, $\phi_u^{-1}(A)\in\mathscr R$ for every $A\in\mathscr X_0$. Hence $R(A)$ is well defined. Second,
\[
R(\varnothing)=Q\{\phi_u^{-1}(\varnothing)\}=Q(\varnothing)=0,
\quad
R(\mathcal X_0)=Q\{\phi_u^{-1}(\mathcal X_0)\}=Q(\mathcal R)=1.
\]
\noindent\emph{Step 3: establish countable additivity.}
Let $A_1,A_2,\ldots\in\mathscr X_0$ be pairwise disjoint. Preimages commute with countable unions and preserve disjointness, so
\[
R\!\left(\bigcup_{j\ge1}A_j\right)
=Q\!\left(\bigcup_{j\ge1}\phi_u^{-1}(A_j)\right)
=\sum_{j\ge1}Q\{\phi_u^{-1}(A_j)\}
=\sum_{j\ge1}R(A_j).
\]
Thus $R$ is a probability measure. By definition it is the pushforward $\phi_u\push Q$; the formula specifies its value on every member of $\mathscr X_0$, so no second measure satisfying the same pushforward definition can differ from it. This proves existence and uniqueness of the pushforward conditional on $(Q_u,\phi_u,c)$.

\noindent\emph{Step 4: restrict to the declared states and normalize.}
For the finite declared states $x_1,\ldots,x_K$, put $r_k=R(\{x_k\})$ and
\[
M_u(c)=R(\mathcal X)=\sum_{k=1}^K r_k.
\]
If $M_u(c)>0$, each coordinate
\[
p_{u,k}(c)=\frac{r_k}{M_u(c)}
\]
is well defined and nonnegative, and the coordinates sum to one. Hence $\boldsymbol p_u(c)\in\Delta_K$. Both numerator and denominator are uniquely determined by the pushforward measure, proving conditional uniqueness of the normalized composition. If $M_u(c)=0$, all declared-state masses vanish and the displayed normalization is undefined; this shows why the positivity condition is necessary. Finally, the proof is conditional on $\phi_u$: changing the semantic map changes its preimages and may therefore change $R$ and $\boldsymbol p_u(c)$.
\end{proof}

\begin{proof}[Proof of Theorem~\ref{thm:approx}]
\noindent\emph{Step 1: bound the error before normalization.}
Write $Q(\cdot)=Q_u(\cdot\mid c)$ and $\widetilde Q(\cdot)=\widetilde Q_u(\cdot\mid c)$. We first control each unnormalized coordinate. For every $k$, add and subtract $Q(B_k)$ and apply the triangle inequality:
\[
|q_k-\widetilde q_k|
\le |Q(A_k)-Q(B_k)|+|Q(B_k)-\widetilde Q(B_k)|.
\]
For arbitrary measurable sets $A$ and $B$,
\[
|Q(A)-Q(B)|
=|Q(A\setminus B)-Q(B\setminus A)|
\le Q(A\setminus B)+Q(B\setminus A)
=Q(A\triangle B).
\]
Consequently,
\[
|q_k-\widetilde q_k|
\le Q(A_k\triangle B_k)+|Q(B_k)-\widetilde Q(B_k)|.
\]
Summing over $k$ gives the first key bound
\begin{equation}\label{eq:proof-unnormalized}
\|\boldsymbol q-\widetilde{\boldsymbol q}\|_1
=\sum_{k=1}^K|q_k-\widetilde q_k|
\le e_{\rm sem}+e_{\rm score}=e.
\end{equation}

\noindent\emph{Step 2: control the normalizing mass and prove existence.}
Next consider the normalizing constants $M=\sum_kq_k$ and $\widetilde M=\sum_k\widetilde q_k$. By the triangle inequality and \eqref{eq:proof-unnormalized},
\begin{equation}\label{eq:proof-mass}
|M-\widetilde M|
=\left|\sum_{k=1}^K(q_k-\widetilde q_k)\right|
\le\sum_{k=1}^K|q_k-\widetilde q_k|
\le e.
\end{equation}
Since $M\ge m_0$ and $e<m_0$, this also proves
\[
\widetilde M\ge M-|M-\widetilde M|\ge m_0-e>0.
\]
Thus the estimated normalized vector exists. Pairwise disjointness of the $B_k$ ensures that $\widetilde M\le1$, although only positivity is needed for the following algebra.

\noindent\emph{Step 3: separate coordinate error from normalization error.}
Add and subtract $\widetilde{\boldsymbol q}/M$ to obtain
\begin{align*}
\left\|\frac{\boldsymbol q}{M}-\frac{\widetilde{\boldsymbol q}}{\widetilde M}\right\|_1
&\le
\left\|\frac{\boldsymbol q-\widetilde{\boldsymbol q}}{M}\right\|_1
+\left\|\widetilde{\boldsymbol q}\left(\frac1M-\frac1{\widetilde M}\right)\right\|_1\\
&=\frac{\|\boldsymbol q-\widetilde{\boldsymbol q}\|_1}{M}
+\|\widetilde{\boldsymbol q}\|_1
\frac{|\widetilde M-M|}{M\widetilde M}.
\end{align*}
\noindent\emph{Step 4: substitute the preceding bounds.}
All coordinates of $\widetilde{\boldsymbol q}$ are nonnegative, so $\|\widetilde{\boldsymbol q}\|_1=\widetilde M$. Substituting \eqref{eq:proof-unnormalized} and \eqref{eq:proof-mass} therefore yields
\[
\|\widetilde{\boldsymbol p}-\boldsymbol p\|_1
\le \frac eM+\frac eM
=\frac{2e}{M}
\le\frac{2e}{m_0},
\]
as claimed. Notice that the condition $e<m_0$ is used to establish existence of $\widetilde{\boldsymbol p}$; the final amplification factor depends on the lower bound for the ideal expressed mass.
\end{proof}

\begin{proof}[Proof of Theorem~\ref{thm:perturb}]
\noindent\emph{Step 1: push both continuation laws onto the common state space.}
Let $R=\phi\push Q$ and $R'=\phi\push Q'$ on $\mathcal X_0=\mathcal X\cup\{x_0\}$, and let $E=\mathcal X$ denote the expressed event. We begin by proving contraction under the semantic map. For every measurable $A\subseteq\mathcal X_0$,
\[
|R(A)-R'(A)|
=|Q\{\phi^{-1}(A)\}-Q'\{\phi^{-1}(A)\}|
\le d_{\rm TV}(Q,Q').
\]
Taking the supremum over $A$ gives
\begin{equation}\label{eq:push-tv}
d_{\rm TV}(R,R')\le d_{\rm TV}(Q,Q').
\end{equation}

\noindent\emph{Step 2: express normalization as conditioning.}
The normalized lexical compositions are precisely the conditional laws $R(\cdot\mid E)$ and $R'(\cdot\mid E)$ restricted to the $K$ declared states, with $M=R(E)$ and $M'=R'(E)$. For any $A\subseteq E$, add and subtract $R'(A)/M$:
\[
\left|\frac{R(A)}{M}-\frac{R'(A)}{M'}\right|
\le \frac{|R(A)-R'(A)|}{M}
+R'(A)\left|\frac1M-\frac1{M'}\right|.
\]
\noindent\emph{Step 3: bound the two conditioning errors.}
Because $M\ge m_0$, the first term is no larger than $d_{\rm TV}(R,R')/m_0$. For the second term, $R'(A)\le R'(E)=M'$ and
\[
|M-M'|=|R(E)-R'(E)|\le\TV(R,R'),
\]
so
\[
R'(A)\left|\frac1M-\frac1{M'}\right|
=R'(A)\frac{|M-M'|}{MM'}
\le\frac{|M-M'|}{M}
\le\frac{d_{\rm TV}(R,R')}{m_0}.
\]
\noindent\emph{Step 4: take the supremum and apply data processing.}
Combining the two bounds and taking the supremum over $A\subseteq E$ gives
\[
\TV(\boldsymbol p,\boldsymbol p')
\le \frac{2d_{\rm TV}(R,R')}{m_0}
\le \frac{2\TV(Q,Q')}{m_0},
\]
where the final inequality follows from \eqref{eq:push-tv}. Since total variation between probability measures never exceeds one, intersecting this bound with the trivial bound yields
\[
d_{\rm TV}(\boldsymbol p,\boldsymbol p')
\le\min\left\{1,\frac{2d_{\rm TV}(Q,Q')}{m_0}\right\}.
\]

The common semantic map is essential for \eqref{eq:push-tv}; if the map itself
changes, an additional discrepancy between the two partitions must be
controlled.
\end{proof}

\begin{proof}[Proof of Corollary~\ref{cor:calibrated-perturbation}]
\noindent\emph{Step 1: apply stability of the calibration map.}
The assumed Lipschitz property gives, with $\boldsymbol a=\boldsymbol p$ and $\boldsymbol b=\boldsymbol p'$,
\[
d_\Pi\{\psi(\boldsymbol p),\psi(\boldsymbol p')\}
\le L_\psi d_{\rm TV}(\boldsymbol p,\boldsymbol p').
\]
\noindent\emph{Step 2: insert the semantic-coarsening bound.}
Theorem~\ref{thm:perturb} gives
$d_{\rm TV}(\boldsymbol p,\boldsymbol p')
\le 2d_{\rm TV}(Q,Q')/m_0$. Therefore
\[
d_\Pi\{\psi(\boldsymbol p),\psi(\boldsymbol p')\}
\le \frac{2L_\psi}{m_0}d_{\rm TV}(Q,Q'),
\]
which is the asserted calibrated perturbation bound.
\end{proof}

\begin{proof}[Proof of Corollary~\ref{cor:combined}]
\noindent\emph{Step 1: control semantic approximation in the simplex.}
Theorem~\ref{thm:approx} gives
$\|\widetilde{\boldsymbol p}-\boldsymbol p\|_1
\le2(e_{\rm sem}+e_{\rm score})/m_0$.

\noindent\emph{Step 2: transfer this error to log-ratio coordinates.}
By $L_{\rm lr}$-Lipschitz continuity,
\[
\|\ell(\widetilde{\boldsymbol p})-\ell(\boldsymbol p)\|_2
\le \frac{2L_{\rm lr}(e_{\rm sem}+e_{\rm score})}{m_0}.
\]

\noindent\emph{Step 3: add the remaining coordinate disturbance.}
The triangle inequality and the assumed residual bound $\delta_0$ give the
total disturbance radius
\[
\delta
=\frac{2L_{\rm lr}(e_{\rm sem}+e_{\rm score})}{m_0}+\delta_0.
\]

\noindent\emph{Step 4: invoke stable inverse recovery.}
Substituting this explicit $\delta$ into Theorem~\ref{thm:recover} yields the
claimed posterior-error bound. The affine specialization is established next.
\end{proof}

\begin{proof}[Proof of Corollary~\ref{cor:affine-budget}]
\noindent\emph{Step 1: obtain the lower modulus of the affine map.}
For every vector $\boldsymbol v$, the variational characterization of the
smallest singular value gives
\[
 \|B_u\boldsymbol v\|_2
 \ge \sigma_{\min}(B_u)\|\boldsymbol v\|_2.
\]
Taking $\boldsymbol v=\boldsymbol\ell-\boldsymbol\ell'$ and using
$h_u(\boldsymbol\ell)-h_u(\boldsymbol\ell')
=B_u(\boldsymbol\ell-\boldsymbol\ell')$ proves the lower-modulus statement;
the affine offsets cancel. Full column rank implies
$\sigma_{\min}(B_u)>0$.

\noindent\emph{Step 2: identify the disturbance entering the inverse.}
Corollary~\ref{cor:combined} bounds it by
\[
\delta=\frac{2L_{\rm lr}(e_{\rm sem}+e_{\rm score})}{m_0}+\delta_0.
\]

\noindent\emph{Step 3: invert the lower-modulus inequality.}
If an observed lexical perturbation has norm at most $\delta$, the corresponding
posterior-coordinate perturbation has norm at most
$\delta/\sigma_{\min}(B_u)$. Substitution of the displayed value of $\delta$
gives exactly the asserted affine error budget.
\end{proof}

\end{document}